\RequirePackage{etex} 
\documentclass[ejsv2,noshowframe,nohead]{imsart}
\RequirePackage{amsthm,amsmath,amsfonts}
\RequirePackage[authoryear]{natbib}
\usepackage{algorithm}
\usepackage{graphicx} 
\usepackage{algpseudocode}
\RequirePackage[colorlinks,citecolor=blue,urlcolor=blue]{hyperref}

\usepackage[capitalise,nameinlink
    ,noabbrev]{cleveref}
  \hypersetup{
    colorlinks=true,
    linkcolor=blue,
    citecolor=blue
}
\usepackage{autonum}
\usepackage{xcolor}
\usepackage{xr}
\usepackage{enumitem}
\usepackage{chngcntr}
\usepackage{ulem}

\startlocaldefs

\newcommand{\cov}{\operatorname{Cov}}

\newcommand{\T}{\top}

\newcommand{\op}{o_{\mathbb{P}}}
\newcommand{\Op}{O_{\mathbb{P}}}

\newcommand{\E}{\mathbb{E}}
\newcommand{\F}{\mathcal{F}} 
\newcommand{\FF}{\mathcal{F}}
\newcommand{\R}{\mathbb{R}}

\newcommand{\mf}{\mathbf}
\newcommand{\bs}{\boldsymbol}

\newcommand{\bt}{\mathrm{boot}}
\newcommand{\nt}{\lceil n \tau\rceil}
\newcommand{\proj}{\mathcal{P}}

\theoremstyle{plain}
\newtheorem{theorem}{Theorem}[section]
\newtheorem{assumption}{Assumption}[section]
\newtheorem{definition}{Definition}[section]
\newtheorem{proposition}{Proposition}[section]
\newtheorem{example}{Example}[section]
\newtheorem{remark}{Remark}[section]
\newtheorem{lemma}{Lemma}[section]

\definecolor{darkgreen}{rgb}{0.0, 0.4, 0.1} 

        \newcommand{\HDB}{\textcolor{black}}
            
\endlocaldefs

\begin{document}

\begin{frontmatter}
\title{A portmanteau test for multivariate  non-stationary functional time series with an increasing number of lags} 
\runtitle{Portmanteau test for non-stationary functional time series}
\begin{aug}
\author[A]{\fnms{Lujia}~\snm{Bai}\ead[label=e1]{lujia.bai@rub.de}},
\author[A]{\fnms{Holger}~\snm{Dette}\ead[label=e2]{holger.dette@rub.de}}
\and
\author[B]{\fnms{Weichi}~\snm{Wu}\ead[label=e3]{wuweichi@mail.tsinghua.edu.cn}}
\address[A]{Fakult\"at f\"ur Mathematik, Ruhr-Universit\"at Bochum, 44780 Bochum, Germany\printead[presep={,\ }]{e1,e2}}

\address[B]{Department of Statistics and Data Science, Tsinghua University, 100084 Beijing, China \printead[presep={,\ }]{e3}}
\end{aug}

\begin{abstract}
Multivariate locally stationary functional time series provide a flexible framework for modeling \HDB{functional data} exhibiting both temporal and spatial dependencies while allowing for a time-varying data generating mechanism. In this paper, we introduce a  portmanteau-type test for assessing white noise assumptions  tailored  for multivariate locally stationary functional time series without dimension reduction. A simple bootstrap procedure is proposed to implement the test, \HDB{because it is not clear if the limiting distribution of the test statistic exists. }
 Our approach is based on a Gaussian approximation result for a degenerate $U$-statistic of second-order
functional time series involving  an increasing number of lags,  which is of independent interest.
Through theoretical analysis, simulation studies, \HDB{and real data analysis of energy consumptions}, we  demonstrate the efficacy and adaptability of the proposed method in detecting departures from white noise assumptions in multivariate locally stationary functional time series.  \HDB{Finally, the R package corresponding to our method can be downloaded  from \href{https://github.com/Lujia-Bai/nftsport}{\textbf{nftsport}. }}
\end{abstract}



\end{frontmatter}



\section{Introduction}
Functional time series provide a very  flexible framework  for modeling  temporal dependence between different observations \HDB{of functions} with  numerous applications  including the analysis of electricity prices \citep{liebl2013,7917305}, daily concentration curves of particulate matter and gaseous pollutants at different locations \citep{LI2017235}, cumulative intraday return trajectories \citep{HORVATH201466} among many others. Exploring the complex temporal dynamics of functional time series has   therefore  become a central  topic in applied and mathematical statistics, see for example the  monographs  of   \cite{horvath2012inference} and \cite{hsingeubank2015}.  Since then the field of functional data analysis is continuously growing 
\citep[see, for example,][for some recent developments]{Lundborg2022,chang2023autocovariance,basubastianDette2024,tan2024}.

White noise tests are frequently used to  analyze 
the serial dependence structure of  time series or checking  model
assumptions \citep[see, for example,][for some recent developments]{TSAY2020106, yao2019, chang2023testing}. 
For  functional time series   white noise testing  is challenging, and  several authors have worked under the assumption of stationarity  in recent years. For example,   \cite{ZHANG201676},  \cite{Bagchi2018}, \cite{rice2020}  and \cite{hlavka2021testing} 
have investigated such tests in the frequency domain, and
\cite{Gabrys2007}, \cite{horvath2012inference}, \cite{kokoszka2017inference}, \cite{rice2020tests} and \cite{MESTRE2021107108} 
among others
have developed tests using a   time domain approach.  We refer to \cite{kimkokricesurvey} for a recent review on the subject \HDB{for IID and uncorrelated white noise}.

Although these methods provide powerful tools for white noise testing  in  functional time series,  there are several limitations, which restrict their application. First,  several  portmanteau tests for functional time series such as those in \cite{Gabrys2007} and \cite{kokoszka2017inference}
are based on functional principal component analysis (FPCA).  Although FPCA has been successfully applied in many cases, it has also  been pointed out by \cite{aue2018detecting} that  the application  of FPCA can cause a loss of information when projecting  infinite-dimensional  data onto a finite-dimensional orthogonal basis, especially when  testing for a signal which is orthogonal to the selected basis.
Second, to our best knowledge, the problem of constructing  tests for the white noise assumption in 
multivariate functional time series is left relatively untouched. Even in the literature of multivariate time series analysis,   portmanteau tests have been majorly discussed for ARMA models, see for instance \cite{hosking1980}, \cite{Li1981}, \cite{francq2007} and  \cite{Mahdi2012} among others.  Finally, most of the literature deals with  stationary functional time series where the  data-generating mechanism is not allowed to be  time-varying. A notable exception is the work of \cite{Axel2023}, where univariate locally stationary \HDB{centered} functional time series are  considered. 

\par

In this paper, we propose a fully functional nonparametric  portmanteau-type test  for \HDB{uncorrelated} white noise in a multivariate locally stationary functional time series, where dense signals in different lags and dimensions of unknown form can be accumulated. Our approach is based on an aggregation of non-parametric estimates of the norm of  
functional auto-covariances, where   the number of  lags is increasing  with the sample size.  In contrast to the classical portmanteau tests, the corresponding test statistic considered in this paper converges to $0$ under the null hypothesis of white noise  (for an increasing sample size).

To investigate its  asymptotic behavior, we develop a systematic theoretical framework for a general class of  degenerate $U$-statistics consisting of second-order statistics of  not necessarily stationary functional time series. An essential ingredient of our  approach is  a non-standard  high-dimensional Gaussian approximation  result, which addresses the  increasing dependence caused by the summation of the growing number of lags 
among the nonparametric estimates of the covariance functions. Note that  the existing literature on Gaussian approximation for high-dimensional series cannot be directly used to deal with these intricate dependencies
because these papers make rather strict assumptions on 
the boundedness of the coefficients measuring the   strength of the serial dependence;  see, for example,
the assumptions in Theorem 3.2 of \cite{wu2017}, 
Assumption 2.3 in  \cite{zhang2018}, 
or Assumption  (G.1) of  \cite{miessteland2023}.

Based on this Gaussian approximation, we address the problem of a potentially \HDB{ non-existing limiting distribution } of the test statistic 
by  developing a novel difference-based multiplier bootstrap procedure to obtain critical values, which is adapted to the increasing dependence between the multipliers caused by  the aggregation of a growing number of lags, while filtering out signals via differencing to ensure good power performance.
We prove  the  asymptotic correctness and consistency of the  bootstrap test  and validate its finite sample performance by means of a simulation study.

Our Gaussian approximation results are of independent interest for  degenerate U-statistics of  functional time series with a kernel depending on the sample size $n$. 
The asymptotic properties  of degenerate $U$-statistics have been studied by \cite{zhou2014}  for (univariate) locally stationary time series.  A typical limit distribution is a weighted sum of centered chi-square  distributions with weights depending on $n$, which are difficult to estimate from the data. A common approach 
to obtain the limiting distribution
of  $U$-statistics with kernels depending on $n$ is the  martingale central limit theorem; see    \cite{HALL19841}  for the case of independent data and \cite{xiao2014portmanteau} for time series.  However, \cite{phillipp1991} showed that in general one can not embed vector-valued martingales into a Gaussian process. Meanwhile, \cite{belloni2018high} proposed a high-dimensional Gaussian approximation for the maximum of martingale processes, \HDB{but  as  mentioned in Remark 1 of \cite{belloni2018high}, conditions for Theorem 3.1 in their  paper cannot be obtained via the element-wise  convergence in probability, and the proof of Theorem 3.1 cannot be directly extended to the case of the element-wise  convergence in probability.} In  contrast to these papers, the  Gaussian approximation scheme developed in this paper  
is suitable for a general class of $U$-statistics derived from a locally stationary functional time series.

\section{Test for white noise in multivariate locally stationary functional time series} \label{sec2} 
  
\def\theequation{2.\arabic{equation}}	
  \setcounter{equation}{0}

\subsection{Multivariate locally stationary functional time series}
For $q \geq 1$, let $\|X \|_q = \{\E|X|^q\}^{1/q}$ the $q$-norm of a real-valued random variable $X$
with existing moment of order $q$. In the case $q = 2$, we omit the index $q$ and use $\|X \|$ for $\|X \|_2$. \HDB{Let $b$ be the smoothing parameter for estimating the trend function, while $\tau$ the smoothing parameter for estimating the covariance function. Let $|\mf a|_F$ denote the Frobenius (or Euclidean) norm of the vector $\mf a$, and for any square integrable function $f:[0,1] \to \mathbb{R}$, the norm of $f$ is defined by   $\|f\|_{\mathcal L^2} = \big \{\int_0^1 f^2(t) dt \big \}^{1/2}$.  In the rest of the paper, $n$ is the number of observed functions and $N$ is the number of grid points ${1 \over N} ,{2 \over N} ,  \ldots, {N-1 \over N}, {N \over N}$, where  each function is observed.} Throughout this paper, 
we  consider  a $p$-dimensional functional time series of the  form 
\begin{align}
    \mf X_{i,n}(u) = \mf m(i/n, u) + \bs \epsilon_{i,n}(u), \quad u \in [0,1], \label{eq:model} 
\end{align}
where $\mf X_{i,n}(u) = (X_{i,n, 1}(u), \ldots, X_{i,n, p}(u))^{\top} \in \mathbb{R}^p$, $\mf m:[0,1]\times [0,1] \to \mathbb R^p$ is a  $p$-dimensional  mean function and  for a given $i$, $\bs \epsilon_{i,n}$ lies in the  space, 
\HDB{${\cal L}^2\big ( [0,1],\mathbb{R}^p\big):= 
\big \{ \ f =(f_1, \ldots , f_p )^\top : [0,1] \to \mathbb{R}^p~:~f_i \in {\cal L}^2[0,1] \big \} $, where ${\cal L}^2[0,1]$ denotes the space of square integrable functions on the interval $[0,1]$.} \HDB{In each  $X_{i,n,k}(u) $, the index $i \in \{ 1,\ldots,n \} $ denotes the temporal component,  $k=1,\ldots, p$ corresponds to the spatial component, and $u \in [0,1]$ represents the argument of the function.} For  the vector-valued  function $\bs f(\cdot) = (f_1, \ldots, f_p)^{\top} \in {\cal L}^2\big ( [0,1],\mathbb{R}^p\big)$,  we define the  norm 
$$
\||\bs f|_F\|_{\mathcal L^2} = \HDB{\sqrt{\sum_{i=1}^p \int f_i^2 (t) \mathrm dt} },
$$
then, equipped with  $\||\cdot|_F\|_{\mathcal L^2}$, ${\cal L}^2\big ( [0,1],\mathbb{R}^p\big)$ is a Hilbert space (with the natural inner product).
In  model \eqref{eq:model},  the errors come from a triangular array $\{\bs \epsilon_{i,n}  \colon  i=1, \ldots  , n \}_{n \in \mathbb{N}} $ 
\HDB{of centered random variables  taking values in the space $ {\cal L}^2\big ( [0,1],\mathbb{R}^p\big)$.
In particular, for a fixed $i$ (and $n$), $\boldsymbol \epsilon_{i,n} = (\epsilon_{i,1,n} , \ldots , \epsilon_{i,p,n})^\top $ is a $ {\cal L}^2\big ( [0,1],\mathbb{R}^p\big)$-valued random variable and we assume that it  } can be represented as
$$
\bs \epsilon_{i,n}(u) = \mf H(i/n, u, \F_i).  
$$
Here $\mf H = (H_1, \cdots, H_p)^{\top}:  [0, 1] \times [0, 1] \times {\cal S}^{\mathbb Z}
\to \mathbb R^p$
 is a $p$-dimensional measurable filter function,  $\mathcal F_i = (\cdots, \eta_{i-1}, \eta_{i})$  and   $(\eta_i)_{i \in \mathbb Z}$ 
 is a sequence of independent and identically  distributed (i.i.d)  
 random elements taking values in  a measurable space  ${\cal S}$. \HDB{Note that the randomness enters the different dimensions of  $\bs \epsilon_{i,n}$ through the random elements as well as their indices.} 
\HDB{
The mapping
     $\mathbf H$ denotes the time-varying functional data generating mechanism. More specifically for any fixed $t \in [0,1]$, $\big ( \mathbf H(t, \cdot, \mathcal{F}_i)\big )_{i \in \mathbb{Z}}$ defines a stationary functional time series in $ {\cal L}^2\big ( [0,1],\mathbb{R}^p\big)$ (with a distribution depending on $t$), and we denote by $\boldsymbol \epsilon_{i,n} = \mathbf H(i/n,  \cdot , \mathcal{F}_i)$ 
     the (functional) observation  in $ {\cal L}^2\big ( [0,1],\mathbb{R}^p\big) $ at time point $t= i/n$.   As a consequence the time series $(\mf X_{i,n})_{i=1, \ldots , n} $   in  $ {\cal L}^2\big ( [0,1],\mathbb{R}^p\big) $   defined by \eqref{eq:model} is not necessarily stationary.
     }
     
 \HDB{We observe a sample of $n$ of vector-valued  functions, where  each function is evaluated  on $N$ grid points, say $  1/N, \cdots, N/N$. Moreover,  the  methods developed in this paper are  suitable for continuously observed functions. In this  case,  we can still discretize the functions on the grid $ u =1/N, \cdots, N/N$.} We assume that the stochastic process $\mf H$ is a  multivariate locally stationary functional time series in the following sense. 
 
\begin{definition} \label{defloc}
Let $(\eta_i^{\prime})_{i \in \mathbb Z}$ denote an independent copy of 
of $(\eta_i)_{i \in \mathbb Z}$ and define $\mathcal F^{*}_i = (\cdots, \eta_{-1}, \eta_0^{\prime}, \eta_1,\cdots,\eta_i)$. Define the element-wise $q$-order functional dependence measure as 
    \begin{align}
        \delta_{q, i, j}:= \delta_{q, i}(H_j) =\sup_{u, t \in [0,1]}\mathbb \|H_j(t, u, \mathcal F_i)-H_j(t, u, \mathcal F_i^*)\|_{q},
    \end{align}
where $H_j(t, u, \mathcal F_i)$ is the $j$-th coordinate  of the vector $\mf H(t, u, \mathcal F_i)$. 
Then,  $\mf H(t, u, \mathcal F_i)$ is a multivariate locally stationary functional time series,  if there exists a constant $\chi \in (0,1)$ such that 
\begin{align}
\label{dette1}
    \max_{1 \leq j \leq p} \delta_{1, i,j} = O(\chi^{i}),
\end{align}
and local stationarity is satisfied, namely
    there exists some constant $K > 0$ such that    
    \begin{align}
    \label{dette2}
    \max_{1 \leq j  \leq p} \|H_j(t_1, u, \F_0) -  H_j(t_2, u, \F_0)\| \leq K  |t_1-t_2|,\quad \text{for all}~ u, t_1, t_2 \in [0,1] .
    \end{align}
\end{definition}

It is important to point out that our approach does not rely on a particular Karhunen-Loève expansion. Therefore, it  circumvents the problem  
that the time-varying covariance operator of a multivariate locally stationary process  will lead to time-varying  eigenspaces, which are very difficult  to estimate. In the following discussion we give two examples of multivariate locally stationary processes satisfying the assumptions made in Definition \ref{defloc}.

\begin{example}[Multivariate time-varying Karhunen-Loève-type expansion]
{\rm 
\HDB{Let  $M\in \mathbb N$ and  $\sigma_1, \ldots \sigma_M: [0,1]\to \mathbb{R} $ denote differentiable functions such that   $\sup_{t\in [0,1]} \sum_{m=1}^M (\sigma^{\prime}_m(t))^2 < \infty$. Let $\{\eta_{i,m} ~|~m=1, \ldots , M,~i \in \mathbb{N} \} $ denote  $i.i.d.$  random variables with variance $1$ and define $\eta_{i} = (\eta_{i,1}, \ldots , \eta_{i,M})^\top$ ($i \in \mathbb N$). We consider the (univariate)  model 
\begin{align}
\label{revhd2}
    {H}(t,u, \mathcal F_i) = \sum_{m=1}^{M}  \sigma_m(t)( \eta_{i,m} + 0.5  \eta_{i-1,m})2^{-m} \cos(m\pi u).
\end{align}
Then,  
\begin{align}
\|H(t, u, \mathcal F_1) -  H(t, u, \mathcal F_1^{*})\|^2 & \leq  \sup_{ u\in [0,1]}\Big \{\sum_{m=1}^{M} (2^{-m}\cos(m\pi u))^2 \Big \}  \\
& \qquad \times  \E \Big \{ \sum_{m=1}^{M} \sigma^2_m(t) (     \eta_{1,m} + 0.5  \eta_{0,m}- \eta_{1,m} - 0.5  \eta^{'}_{0,m} )^2 \Big \}   \\
& < 0.25 \sup_{ u\in [0,1]}\Big \{\sum_{m=1}^{M} 4^{-m}  \Big \} \times   \Big \{ \sum_{m=1}^{M} 2\sigma^2_m(t)   <\infty\Big \},
\end{align}
 and $\|H(t, u, \mathcal F_i) -  H(t, u, \mathcal F_i^{*})\| =0$, if $i>1$, where $\eta_{0,m}$ and $\eta_{0,m^{\prime}}$ are $i.i.d.$ random variables. Therefore,  condition  \eqref{dette1} of \cref{defloc}
is satisfied.}
\HDB{Moreover, for $t_1, t_2, u \in [0,1]$, we have 
\begin{align}
   \| {H}(t_1,u, \mathcal F_i) -{H}(t_2,u, \mathcal F_i)\|^2 &= \E \Big \{\sum_{m=1}^{M}  (t_2-t_1)\sigma^{\prime}_m(t^*)( \eta_{i,m} + 0.5  \eta_{i-1,m})2^{-m} \cos(m\pi u)\Big \}^2 \\
   &\leq | t_2-t_1|^2 \sup_{t \in [0,1]}\Big \{\sum_{m=1}^M 1.25 (\sigma^{\prime}_m(t))^2\Big\}   \Big\{\sum_{m=1}^{M} 4^{-m} \Big\} < C | t_2-t_1|^2 ,
\end{align}
where $C$ is a positive constant which is independent of $M$. Therefore, condition \eqref{dette2} 
in Definition \ref{defloc} is also satisfied. \\
The model \eqref{revhd2} is a special case of a  more generous functional time series model:}
\begin{align}
\mf H(t, u, \mathcal F_i) = \sum_{m=1}^{\infty} \lambda_m(t, \mathcal F_i)\bs \psi^{(m)} (t, u),
\end{align}
where for any fixed $t \in [0,1]$, $\bs \psi^{(m)}(t, u)  = (\psi^{(m)}_1(t, u), \ldots,\psi^{(m)}_p(t, u) )^{\top} $, $1\leq m\leq\infty$ is a  complete orthonormal basis of  $ {\cal L}^2\big ( [0,1],\mathbb{R}^p\big)$, $\lambda_m(t, \mathcal F_i) = \int \sum_{j=1}^p H_j(t, u, \mathcal F_i) \psi_j^{(m)} (t, u) du$. If \HDB{the filter $\mathbf H = (H_1, \ldots , H_p)^\top $ (almost surely) is differentiable with respect to its second argument, $\sup_{u, t \in [0,1]}\| H_j(t,u,\mathcal F_i)\| \leq  C\sup_{t \in [0,1]} \| \| H_j(t,\cdot,\mathcal F_i)\|_{\mathcal L_2}\|$},  and 
 \begin{align}
\sup_{t \in [0,1]} \Big \| \Big [\sum_{m=1}^{
\infty} \{\lambda_m(t, \mathcal F_i) - \lambda_m(t, \mathcal F_i^*)\}^2 \Big ]^{1/2} \Big \| &= O(\chi^i)
 \end{align}
  for some $\chi \in (0,1)$, it can be shown that condition  \eqref{dette1} of \cref{defloc}
holds.
Moreover, if the functions $\bs \psi^{(m)}$ are  Lipschitz  continuous and the random variables $\lambda_m$ are 
stochastically  Lipschitz  continuous in the sense that   
\begin{align} 
    \Big \|[\sum_{m=1}^{
\infty} \{\lambda_m(t, \mathcal F_0) - \lambda_m(s, \mathcal F_0)\}^2]^{1/2}  \Big\| &\leq  c_1 |t-s| , \\ 
\Big \|\Big [\sum_{m=1}^{
\infty} \lambda_m^2(t, \mathcal F_0)\Big ]^{1/2} \Big  \| &\leq  c_2,\\
\Big \|[\sum_{m=1}^{
\infty} |\bs \psi_j^{(m)} (t, \cdot)  - \bs \psi_j^{(m)} (s, \cdot) |^2_F]^{1/2}\Big \|_{\mathcal L^2}& \leq  c_3 |t-s| , 
\end{align}
for all $s,t \in [0,1]$, where 
$c_1, c_2$ and $c_3$ are  positive constants,  condition \eqref{dette2} 
in Definition \ref{defloc} is also satisfied, and  the process $\mf H$ is a multivariate 
locally stationary process. }
\end{example}

\begin{example}[Multivariate locally stationary iterated functional time series]
\par 
\HDB{Consider the functional AR(p) model (see also Section 3.2 of \cite{bosq2000linear}), 
\begin{align}
    Z_i^{(t)} = \Psi^{(t)} ( Z_{i-1}^{(t)}) + \eta_i, \label{eq:far} 
\end{align}
where $\eta_i$ are $i.i.d.$ random elements in  $ {\cal L}^2\big ( [0,1],\mathbb{R}^p\big)$ and  $\big \{ \Psi^{(t)} ~|~t \in [0,1] \} $ is a family of linear operators, for example integral operators of the form   
 $$
 \Psi^{(t)} ( Z_{i-1}^{(t)})= \int_0^1 \rho(t) \bs A(\cdot,u)Z_{i-1}^{(t)}(u) du
 $$ with some kernel $\bs A: [0,1]\times[0,1] \to \mathbb R^{p \times p}$ and a function $\rho : [0,1] \to \mathbb{R}$.  If there exists a constant $j^0 \geq 1 $ such that the norm of the linear operator satisfies 
 $$
\| (\Psi^{(t)})^{j_0}\|_{\mathcal L}:= \sup_{x \in {\cal L}^2\big ( [0,1],\mathbb{R}^p\big), \||x|_F\|_{\mathcal{L}_2} \leq 1}\||(\Psi^{(t)})^{j_0}(x) |_F\|_{\mathcal{L}_2} <  1,
$$
then the process \eqref{eq:far} can be represented in the form  
\begin{align}
    Z_i^{(t)} (\cdot) = \mf H(t, \cdot, \mathcal F_i)  :=  \eta_{i}    (\cdot) +  \sum_{j=1}^\infty(\Psi^{(t)})^j \eta_{i-j}  (\cdot)  .
\end{align}
\HDB{
If  $\mf H$ is almost surely differentiable with respect to its second argument with uniformly bounded derivatives, i.e., $\sup_{t,u\in[0,1]}
\left|\partial_u \left|\mathbf H(t,u,\mathcal F_i)\right|_F\right|
\leq L < \infty $ (almost surely), then uniformly for all $j=1,\ldots, p$, $t \in [0,1]$, we have   $ \sup_{u \in [0,1]}\|   H_j (t,u,\mathcal F_i) \| \leq \| \sup_{u \in [0,1]}|\mf H (t,u,\mathcal F_i)|_F\| \leq  c \| \||\mf H (t,\cdot,\mathcal F_i)|_F\|_{\mathcal L_2}\|$ for a positive constant $c$. 
  By the definition of the operator norm, it follows that  uniformly for $t \in [0,1]$, $j = 1,\ldots, p$,
\begin{align}
  \sup_{u \in [0,1]}\|   H_j (t,u,\mathcal F_i)  - H_j (t,u,\mathcal F_i^*) \|^2  &\leq   c^2  \|\| |\mf H(t, \cdot, \mathcal F_i)-\mf H(t, \cdot, \mathcal F_i^{*})|_F\|_{\mathcal L^2}\|^2\\ & \leq  c^2  \E \|| (\Psi^{(t)})^i  (\eta_{0} - \eta_{0}^{\prime})|_F\|_{\mathcal L^2}^2\\ &\leq  c^2  \|  (\Psi^{(t)})^i \|_{\mathcal L }^2  \E\||(\eta_{0} - \eta_{0}^{\prime})|_F\|_{\mathcal L^2}^2 = O(\chi^{2i}),
\end{align}
for some $\chi \in (0,1)$. This means that \eqref{dette1} is satisfied. If we further assume that  $\||\Psi^{(t_1)} - \Psi^{(t_2)} |_F\|_{\mathcal L_2} < C|t_1 - t_2| $ for any $t_1, t_2 \in [0,1]$ (with some positive constant $C$),  we have for another positive constant $C^{\prime}$
\begin{align}
  &\sup_{u \in [0,1]}\|   H_j (t_1,u,\mathcal F_i) - H_j (t_2,u,\mathcal F_i) \|^2 \\ & \leq   c^2 \big  \|  |  \mf   H(t_1, \cdot, \mathcal F_i)-\mf H(t_2, \cdot, \mathcal F_i )|_F\big \|^2_{\mathcal L^2} \\ &=   c^2  \E \Big \|\Big |\sum_{j=0}^\infty(\Psi^{(t_1)} -\Psi^{(t_2)} )^j \eta_{i-j} \Big |_F\Big \|_{\mathcal L^2 }^2 \\ 
      & \leq   c^2\sum_{j=1}^\infty\| (\Psi^{(t_1)} -\Psi^{(t_2)} )^j \|_{\mathcal L }^2 \E \||\eta_{i-j}|_F\|_{\mathcal L^2 }^2\\
      & \leq C^{\prime} |t_1-t_2|^2.
\end{align}}}
{\rm Similar arguments can be used to establish these properties for general  non-linear processes defined by 
\begin{align}\label{eq:lsiterative}
    \mf Z_i^{\HDB{(t)}}(\cdot) = R(t, \mf Z_{i-1}^{\HDB{(t)}}(\cdot), \bs \epsilon_i),
\end{align}
where  $t \in [0,1]$ and 
$\mf Z_{i-1}^{\HDB{(t)}}(\cdot) \in {\cal L}^2\big ( [0,1],\mathbb{R}^p\big)$ is a $p$-dimensional  vector of  square-integrable random functions on the interval $[0,1]$
 and   $R: [0,1]  \times {\cal L}^2\big ( [0,1],\mathbb{R}^p\big)\times {\cal L}^2\big ( [0,1],\mathbb{R}^p\big) \to {\cal L}^2\big ( [0,1],\mathbb{R}^p\big)$ 
  is a measurable function. 
  If the dimension is $1$ and  stationarity (which means that the function $R$ in \eqref{eq:lsiterative} does not depend on $t$) is assumed, the process defined by \eqref{eq:lsiterative} includes both the functional AR($p$) model \citep{kokoszka2017introduction} and the functional GARCH model \citep{functionalgarch2020}.

\par 
Note that the Hilbert space ${\cal L}^2\big ( [0,1],\mathbb{R}^p\big)$ is a complete and separable metric space. Therefore, if there exists a square-integrable function $
\mf Z_0(\cdot)$ 
\begin{align}
    &\sup_{t \in [0,1]} \|\| |R(t, 
\mf Z_0(\cdot), \bs \epsilon_i)|_F\|_{\mathcal{L}^2}\| < \infty, \label{eq:momentq}\\
    &\chi_1 := \sup_{t \in [0,1]} \sup_{\mf x(\cdot) \neq \mf y(\cdot) \in {\cal L}^2\big ( [0,1],\mathbb{R}^p\big)} \frac{\|\| | R(t, \mf x(\cdot), \bs \epsilon_i) - R(t, \mf y(\cdot), \bs \epsilon_i)|_F\|_{\mathcal{L}^2}\|}{\||\mf x(\cdot) - \mf y(\cdot)|_F\|_{\mathcal{L}^2}} < 1, \label{eq:chi1} 
\end{align}
it follows from 
  Theorem 2 of \cite{wu2004limit} that,  for any $t \in [0,1]$, the iterative equation  \eqref{eq:lsiterative} admits a unique stationary solution
\begin{align}
    \mf Z_i^{\HDB{(t)}}(\cdot) = \mf H(t, \cdot, \mathcal{F}_i), 
\end{align} 
such that 
\begin{align}
    \sup_{t \in [0,1]} \left\|\| |\mf H(t,  \cdot, \mathcal{F}_i) -  \mf H(t, \cdot, \mathcal{F}_i^*)|_F\|_{\mathcal{L}^2} \right\| = \mathcal{O}(\chi_1^i).
\end{align}
 Furthermore, assume that  \HDB{the filter $\mathbf H$ is differentiable with respect to its second argument, so that $\sup_{u \in [0,1]}\| H_j(t,u,\mathcal F_i)\| \leq  C \| \| H_j(t,\cdot,\mathcal F_i)\|_{\mathcal L_2}\|$} and that $\sup_{t \in [0,1]} M(\mf  H(t, \cdot, \mathcal{F}_0))  < \infty$,
 where 
\begin{align}
   M(\mf  x(\cdot)) = \sup_{0 \leq t < s \leq 1} \frac{\|\|| R(t, \mf  x(\cdot), \bs \epsilon_i) - R(s, \mf  x(\cdot), \bs \epsilon_i)|_F\|_{\mathcal{L}^2}\|}{|t - s|}.
\end{align}
 Then,  
 a straightforward but tedious calculations  shows
there exists a positive constant $C$  such that \ 
\begin{align}
    \max_{1 \leq j \leq p} \sup_{u \in [0,1]} \|H_j(t, u, \mathcal{F}_0) - H_j (s, u, \mathcal{F}_0)\| \leq C |t - s|,
\end{align}
and a constant $\chi_2 \in (0,1)$ such that for $i > 0$, $1 \leq j \leq p$, $\delta_{1,i,j} \leq \delta_{2,i,j} = \mathcal{O}(\chi_2^i)$. 
As a consequence, the process \eqref{eq:lsiterative} is  a
multivariate locally stationary functional time series in the sense of Definition \ref{defloc}.

}

\end{example}

\subsection{Testing procedures} 

Under appropriate assumptions  on  the filter $\mf H $ (which will be given below), we can define  the  auto-covariance function 
at lag $k$ as 
\begin{align}
\bs  \Gamma_k:  \left\{ 
\begin{array}{l}
       [0,1]^2 \to  \mathbb R \\
     (t, u) \mapsto     \bs \Gamma_k(t, u) = \mathrm{Cov}(\mf H(t, u,  \mathcal F_0), \mf H(t, u,\mathcal  F_k))
\end{array}
\right.  ~, 
\nonumber
\end{align}
which directly measures the cross-covariance of lag $k$ of the random observations at each time point in the functional time series.

In order to construct a 
test for white noise of the  multivariate locally stationary functional time series  $\{ \mf X_{i,n} \colon  i=1 , \ldots , n \} $,  
we consider the hypothesis
\begin{align}
    H_0:  \bs \Gamma_k(t, u) = \mf 0, \text{   for all  } ~k \in \mathbb N ,\ (t,u)^\top \in [0,1]^2. \label{eq:H0}
\end{align}
\HDB{In \eqref{eq:H0}, the difference of time is shown in the index $k$, which is the time-lag in the covariance. To see this,  note that $\mathrm{Cov} (\boldsymbol \epsilon_{i,n}(u), \boldsymbol  \epsilon_{j,n}(u)) \approx \Gamma_{|j-i|}(i/n, u)$, if $|j-i|/n \to 0$. }
An extension to testing the hypothesis
\begin{align}
    H_0: \mathrm{Cov}(\mf H(t, u,  \mathcal F_0), \mf H(t, v,\mathcal  F_k)) = \mf 0, \text{ for all }  k\in \mathbb{N} \text{ and all } (t,u,v)^\top  \in [0,1]^3, \label{eq:extended}
\end{align}
which includes the  \HDB{cross dependence between functional locations} $u$ and $v$, 
will be discussed in Section \ref{sec:extend}.  
Specifically, we will develop a Box-Pierce portmanteau-type test  using the statistic 
\begin{align}
\max_{1 \leq i \leq n, 1 \leq j \leq N} \sum_{k=1}^{s_n}\mathrm{tr}\big \{\widehat{\mf G}_{k} (i/n, j/N) \big \}, \label{eq:estimator}
\end{align}
where $\mathrm{tr}\{A\}$ denotes the trace of  a square matrix $A$,
$s_n$ is the number of lags under consideration which is increasing with the sample size $n$.  Here  $\widehat{\mf G}_{k} (t, u)$ is a  nonparametric  estimate of 
${\mf G}^{\circ}_{k} (t, u)= 2
\bs \Gamma_k(t,u)\bs \Gamma^{\top}_k(t,u)$ and  $\mathrm{tr}\big \{\widehat{\mf G}_{k} (t, u)\big \}$ is therefore the  nonparametric  estimate of $2|\bs \Gamma_k(t,u)|_F^2$,
where $| \cdot |_F$ denotes the Frobenius norm of a matrix.
In the multivariate functional time series, there could be dense signals on different lags and different dimensions of unknown form. Therefore, the  portmanteau-type statistic \eqref{eq:estimator}
combines a maximum with a sum 
to   enhance the power by summing up all potentially  weak signals.

For the construction of the estimator $\widehat{\mf G}_{k} (t, u)$ we  consider the nonparametric residuals 
\begin{align}
\label{det10}
    \hat{\bs \epsilon}_{i,n}(u) = \mf X_{i,n}(u) - \hat{\mf m}(i/n, u), 
\end{align}
where 
\begin{align}
    \label{det30}
\hat{\mf m}(t,u) = \frac{1}{nb}\sum_{i=1}^n \mf X_{i,n}(u) K_b(i/n-t),
\end{align}
is a 
Nadaraya Watson-type estimate with bandwidth $b$, $K_{b} (x) =K (x/b)  $,  $K$ is a symmetric
 kernel function  
 supported  on the interval on the interval $(-1,1)$ with  $\int K(t) dt = 1$. 
With the notation  
$$ 
\hat{\bs \phi}_{i,k}(u)=\hat {\bs \epsilon}_{i-k}(u) \hat{\bs \epsilon}_i^{\top} (u),
$$
the estimate of $2|\bs \Gamma_k(t,u)|_F^2$ is  defined by 
$
\mathrm{tr}\{\widehat{\mf G}_{k}(t,u)\} ,
$
where 
\begin{align}
\widehat{\mf G}_{k}(t,u)   &= 
\frac{1}{n\tau s_n} \sum_{i=1}^{n} K_{\tau}(i/n-t) \hat{\bs \phi}_{i,k} (u)\sum_{l=1 }^{n} K_{\tau}(l/n-t)\hat{\bs \phi}^{\top}_{l,k}(u) \mf 1(l \in \mathbb L_i) ,\label{eq:hatgamma}
\end{align} 
where $\mathbb L_i = \{l: M_n < |l-i| \leq s_n\}$ 
and  $\mf 1(\cdot)$ is the indicator function. 
Note that
 \begin{align}
     \label{det33} 
 \hat{\bs \Gamma}_k(t,\cdot) = \frac{1}{n\tau} \sum_{i=1}^n  K_{\tau} (i/n-t) \hat{\bs \phi}_{i,k}(\cdot) = \frac{1}{n\tau} \sum_{i=1}^n\hat{\bs \epsilon}_{i-k}(\cdot)\hat{\bs \epsilon}^{\top}_i(\cdot) K_{\tau} (i/n-t)
 \end{align}
 is the Nadaraya-Watson estimator of  $\bs \Gamma_k(t,\cdot)$ and that the set 
$\mathbb L_i $ is used to control the amount of dependence after multiplying the two sums in \eqref{eq:hatgamma}. 
\HDB{Note that the kernel is only a function of $t$,  since our method does not depend on pre-smoothing, and treats each function as an element.}
In particular, an appropriate choice of $M_n  \to \infty $ guarantees  that the fourth-order cumulants are negligible; see
Lemma \ref{lm:4th} in the appendix for more details. Then, the null hypothesis is rejected for large values of  the statistic in \eqref{eq:estimator}. 

In order to obtain critical values we will develop a difference-based block multiplier bootstrap, which is based on a Gaussian approximation result.  In the following remark, we indicate why the derivation of results of this type is extremely challenging. 

\begin{remark} \label{remu}
{\rm 
Note that the statistic in \eqref{eq:estimator} is a 
maximum of $U$-statistics with a kernel depending on the sample size $n$,  as the random variable 
$\sum_{k=1}^{s_n}\mathrm{tr}\{\widehat{\mf G}_{k} (t,u)\} $ admits the representation 
\begin{align}
    \sum_{k=1}^{s_n} \mathrm{tr}\{\widehat{\mf G}_k(t,u)\} =  \frac{1}{n\tau s_n} \sum_{i\neq l, 1\leq i,l\leq n} k_n\big (\{ \hat{\bs \phi}_{i,k} (u)\}_{k=1}^{s_n},  \{\hat{\bs \phi}_{l,k} (u)\}_{k=1}^{s_n} , t\big ) .  \label{eq:U}
\end{align}
Here, for sequences of 
$p\times p$ matrices $\{ x_{i,k} \}_{k=1}^{s_n}$ and $\{ x_{l,k} \}_{k=1}^{s_n}$, the kernel $k_n(\cdot, \cdot, t): (\R^{p\times p})^{s_n}\times (\R^{p\times p})^{s_n}\to \mathbb R$   depends on the sample size $n$ and is defined by 
$$
k_n\big ( \{ x_{i,k} \}_{k=1}^{s_n}, \{ x_{l,k} \}_{k=1}^{s_n}, t \big ) = \sum_{k=1}^{s_n} \mathrm{tr}(x_{i, k} x_{l, k}^{\top}) K_{\tau}(i/n - t)K_{\tau}(l/n - t) \mf 1(M_n <|l-i| \leq s_n).
$$
Note that $k_n(\cdot,\cdot,t)$ depends also on $M_n$ and $s_n$, which appear in the indicator. This dependence is not reflected in the notation.
To see that \eqref{eq:U} can be approximated by a degenerate $U$-statistic,
we first note that the estimation error is negligible, such that $\hat{\bs \phi}_{i,k} (u)$ can be asymptotically substituted by ${\bs \phi}_{i,k} (u)= {\bs \epsilon}_{i-k}(u) {\bs \epsilon}_i^{\top} (u)$
(details for this argument can be found in the  proof of \cref{prop:esti} in the appendices). Under the assumption of short-range dependence, we can approximate ${\bs \phi}_{i,k} (u)$ by the counterpart with  $m$-dependent errors, 
i.e., $\check {\bs \phi}_{i,k} (u)=\check {\bs \epsilon}_{i-k}(u) \check {\bs \epsilon}_i^\top (u)$,  where 
$\check {\bs \epsilon}_i (u) = \E( {\bs \epsilon}_i (u)|(\eta_{i-m}, \ldots, \eta_i)
)$.  In summary,  under mild conditions  and due to the symmetry of $k_n$, we obtain for  $m/s_n \to 0$ the approximation  
\begin{align}
      \sum_{k=1}^{s_n} \mathrm{tr}\{\hat {\mf G}_k(t,u)\}&\approx \frac{1}{n\tau s_n} \sum_{i \neq l, i,l=1}^{n}k_n(\{\check{\bs \phi}_{i,k} (u)\}_{k=1}^{s_n},  \{\check{\bs \phi}_{l,k} (u)\}_{k=1}^{s_n}, t )\\ 
      & = \frac{2}{n\tau s_n} \sum_{i=1}^{n}\sum_{l=1}^{i-1}k_n(\{\check{\bs \phi}_{i,k} (u)\}_{k=1}^{s_n},  \{\check{\bs \phi}_{l,k} (u)\}_{k=1}^{s_n}, t ). 
\end{align}
To see the reason why the right hand side of the above is degenerate, 
we note that a straightforward calculation shows for  $M_n > 4m$, $M_n/s_n \to 0$,
\begin{align}
    &\E[k_n(\{\check {\bs \phi}_{i,k} (u)\}_{k=1}^{s_n},  \{\check {\bs \phi}_{l,k} (u)\}_{k=1}^{s_n}, t)|\{\check {\bs \phi}_{l,k} (u)\}_{k=1}^{s_n}] \\ 
    &= \E\Big\{K_{\tau}(i/n - t)K_{\tau}(l/n - t) \mf 1(M_n <|l-i| \leq s_n)\\ 
    & \times \Big (\sum_{k=3m+1}^{s_n} \E[\mathrm{tr}\{ \check{\bs \epsilon}_{i-k}(u) \check {\bs \epsilon}^{\top}_{i}(u)\check {\bs \epsilon}_{l-k}(u) \check {\bs \epsilon}^{\top}_{l}(u)\}|\F_{l-s_n-m}^{l}] \\ &  +\sum_{k=1}^{3m} \E[\mathrm{tr}\{ \check{\bs \epsilon}_{i-k}(u) \check {\bs \epsilon}^{\top}_{i}(u)\check {\bs \epsilon}_{l-k}(u) \check {\bs \epsilon}^{\top}_{l}(u)\}|\F_{l-s_n-m}^{l}]\Big)|\{\check {\bs \phi}_{l,k} (u)\}_{k=1}^{s_n}
   \Big \}\\
    & = 0,
\end{align}
since under $M_n< |l-i|\leq s_n$, $l\leq i-1$, when $k > 3m$, $\check{\bs \epsilon}_{i}(u)$ is $\F_{i-m}^{i}$ measurable and hence independent of $\F_{l-s_n-m}^{l}$  and  $\F_{i-k-m}^{i-k}$, and when $k \leq 3m$, $\check{\bs \epsilon}_{i-k}(u) \check {\bs \epsilon}^{\top}_{i}(u)$ is $\F_{i-k-m}^{i}$ measurable and hence independent of $\F_{l-s_n-m}^{l}$.

Note  that $\sum_{l=1}^{i-1} k_n \big( \{\check{\bs \phi}_{i,k} (u) \}_{k=1}^{s_n},  \{\check{\bs \phi}_{l,k} (u)\}_{k=1}^{s_n}, t \big )$ is $\F_i$ measurable. Therefore, observing the  degeneracy of the statistic under consideration a natural  approach to derive the asymptotic distribution 
would be constructing martingale differences and using  a martingale central limit theorem \citep[see, for example,][]{HALL19841}.
However, a  multivariate version of  the martingale central limit theorem is in general not available, see \cite{phillipp1991}. 
Moreover, the present situation is even more complicated as the dimension 
of the vectors  $\{\check{\bs \phi}_{i,k} (u) \}_{k=1}^{s_n} \in (\mathbb{R}^{p \times p})^{s_n} $ is increasing with the sample size and the dependence over $u\in(0,1)$ among these  functional vectors 
has to  be taken into account.
}
 \end{remark}
 
Despite the difficulties indicated in the previous remark  we  prove  in  
\cref{thm:ga}  under the null hypothesis  the existence of  a    Gaussian vector  $(\mf Z_i)_{i=1}^{2\nt}  \in \mathbb{R}^{N(n-2\nt+1)}$ such that
\begin{align}
    &\sup_{x \in \mathbb R}\Big  | \mathbb{P} \Big ( \sqrt{n\tau}\underset{\substack{ u\in [0,1] \\ t
    \in [\tau, 1-\tau] } }{\sup}  \Big |\sum_{k=1}^{s_n}\mathrm{tr}\big \{\widehat{\mf G}_{k} (t, u)\big \} \Big |\leq x\Big )   
    - \mathbb{P} \Big  (\frac{1}{s_n\sqrt{n\tau} }\Big   |\sum_{i=1}^{2\nt}\mf Z_i \Big  |_{\infty} \leq x\Big  ) \Big| = o(1),\label{eq:ganulla} 
\end{align}
where $\lceil x \rceil$ is the smallest integer that is larger than or equal to $x$
and $| \cdot  |_\infty$ denotes the maximum norm of a vector (its dimension will always be clear from the context).

It  will be shown below that the statistic defined in \eqref{eq:hatgamma}  is also asymptotically unbiased for ${\mf G}^{\circ}_{k} (t, u)$. 
In this way,
we can obtain a statistic which is close to $0$  under the null hypothesis
and large under the alternative. Moreover, by \eqref{eq:ganulla}, the critical values for the statistic
$$
\sqrt{n\tau}\underset{{ u\in [0,1], t
    \in [\tau, 1-\tau] } }{\sup}  \Big |\sum_{k=1}^{s_n}\mathrm{tr}\big \{\widehat{\mf G}_{k} (t, u)\big \} \Big |
$$
can be obtained from the quantiles of the distribution of $\frac{1}{s_n\sqrt{n\tau}} |\sum_{i=1}^{2\nt}\mf Z_i  |_{\infty}$.
The autocovariance structure of  the sequence $(\mf Z_i)_{i=1}^{2\nt}$  is complicated and will be described in detail in Section \ref{sec3} below,
but we briefly motivate this result here. To be precise, we recall the notation of the nonparametric residuals in 
\eqref{det10} and the formulation of $U$-statistic in \eqref{eq:U}. In order to linearize the  $U$-statistic in \eqref{eq:U}, we introduce the notation
\begin{align}
    \hat{ \mf U}_i (t, u) =\sum_{k=1}^{s_n} \big \{   \hat{\bs \phi}_{i,k}(u) \hat{\bs \psi}^{\top}_{i,k}(t,u) +  \hat{\bs \psi}_{i,k}(t,u) \hat{\bs \phi}^{\top}_{i,k}(u)\big \} ~, 
 \end{align}
 where 
 \begin{eqnarray}
 \hat{\bs \psi}_{i,k}(t,u) & = & \sum_{\ell = (i-s_n) \vee 1}^{i- M_n - 1}  K_{\tau}( \ell/n - t) \hat{\bs \phi}_{\ell, k}(u).  \label{det12}
 \end{eqnarray} 
This yields  the representation 
\begin{align}
   \label{hol123} 
\sum_{k=1}^{s_n} \hat {\mf G}_k(t,u) =  \frac{1}{n\tau s_n} \sum_{i=1}^{n}  K_{\tau}(i/n - t)  \hat{\mf  U}_i(t, u).
\end{align}
Next we consider the partition $\{ j/N \}_{j=1 , \ldots , N} $ of the  interval $[0,1]$, 
define the $N$-dimensional vectors 
\begin{align}
    \hat {\mf V}_{i,l}:= K_{\tau}(i/n - l/n)  \big (\mathrm{tr} (\hat{\mf U}_{i}(l/n, j/N) \big )_{1 \leq j \leq N} \in \mathbb{R}^N  \label{eq:hatV1}
\end{align}
and  introduce the $N(n - 2\nt+1)$-dimensional  vectors
\begin{equation} \label{det25}
\hat{\mf V}_i = (\hat {\mf V}_{i,\nt} , \hat {\mf V}_{i + 1,\nt + 1}, \cdots,\hat {\mf V}_{i + n - 2\nt, n-\nt} )^{\top}  ~~~~~ (i=1, \ldots ,  {2\nt} ). 
\end{equation}

Together with \eqref{hol123}, we obtain   the approximation  
$$
\sup_{t \in [\tau, 1 - \tau], u \in [0,1]}\sum_{k=1}^{s_n}\mathrm{tr}\big \{ \widehat{\mf G}_k(t,u) \big \} \approx \frac{1}{n\tau s_n} \Big |\sum_{i=1}^{2\nt} \hat{\mf V}_i \Big  |_{\infty} =  \frac{1}{n\tau s_n}  \max_{\nt \leq l \leq n- \nt}\Big  |\sum_{i=1}^{2\nt} \hat{\mf V}_{i,l}  \Big  |_{\infty} ,
$$
for which \eqref{eq:ganulla} provides a Gaussian approximation,  where the vector
$({ \mf Z}_i)_{i=1}^{2\nt}$ has the same covariance structure as 
$(\hat{ \mf V}_i)_{i=1}^{2\nt}$.
\\

The Gaussian approximation result \eqref{eq:ganulla} inspires us to consider   the test statistic 
\begin{align}
    Q_n = \underset{\nt \leq i \leq n - \nt, 1 \leq j \leq N}{\max}\Big  |\sqrt{n\tau s_n} \sum_{k=1}^{s_n}\mathrm{tr}\big \{\widehat{\mf G}_{k} (i/n, j/N)\big \} \Big  |,\label{eq:testQ}
\end{align}
for the construction of a test for the white noise hypothesis  \eqref{eq:H0}
and to reject the null hypotheisis for large values of $Q_n$.
\HDB{ Since it is not clear if the limiting distribution of  $Q_n$ exists,} we devise a bootstrap procedure to implement a test for the hypotheses  \eqref{eq:H0}, 
which adapts to the non-stationarity of the underlying high-dimensional functional time series by mimicking the distribution of the high-dimensional Gaussian process $(\mf Z_k)_{k=1}^{2\nt}$ using the estimable process \textbf{$(\hat{ \mf V}_i)_{i=1}^{2\nt}$}.  
More precisely, we propose to reject the null hypothesis in
\eqref{eq:H0} whenever
\begin{align}
    \label{hol1}
    Q_n/\sqrt{s_n} > \tilde  r_{\bt}~,
\end{align}
where $\tilde  r_{\bt}$ is a quantile  
obtained by a multiplier bootstrap procedure yielding an asymptotic level $\alpha$ test which is described in detail \HDB{in  
Algorithm \ref{alg:boot2} of the following section.}

\begin{remark}
    {\rm
Portmanteau-type statistics have been developed  by 
\cite{kokoszka2017inference} and \cite{rice2020tests}  for (stationary) functional time series, among others, who propose a sum of squared $\mathcal L^2$-norms  of estimates of the auto-covariance functions under conditional heteroscedasticity. Our approach differs from the aforementioned works with respect to several aspects.  First, we consider multivariate functional time series. Second, we do not assume stationarity,  working under the locally stationary framework instead. Third,  we consider the maximum of the functional estimates of the Frobenius norms of   the auto-covariance functions  over different lags.  Fourth, we allow  the number of considered lags  $s_n$ in $Q_n$   to diverge as sample size increases.  Theorem \ref{thm:Qn}  below  provides  a Gaussian approximation for the distribution of the statistic $Q_n$ in this case, which motivates the} bootstrap procedure in Algorithm \ref{alg:boot2}.  
\end{remark}

\section{Main results} \label{sec3} 
  \def\theequation{3.\arabic{equation}}	
  \setcounter{equation}{0}
  Throughout this section we assume that the error process  $\{ \bs \epsilon_{i,n} \}_{i=1, \ldots ,n}$ in model \eqref{eq:model} is a multivariate locally stationary process  in sense of Definition \ref{defloc}.
Moreover, for our main results we  require  the following assumptions.   
 
\begin{assumption} ~~ 
\begin{enumerate}
    \item[(1)] There exist constants $t_0 > 0$, $c_0 > 0$ such that $\sup_{u, t \in [0,1]}\mathbb E (\exp(t_0|H_j(t, u, \mathcal F_0)|)) < c_0$  uniformly with respect to $j = 1, \ldots, p$.
    \item[(2)]  
   For any  $t \in (0,1)$ the function $u\to H_j(t, u, \F_i)$ is almost surely differentiable  with partial derivative   
    $L_j(t, u, \FF_i) = \frac{\partial H_j(t, u, \FF_i)}{\partial u} $. For $t=0,1$, we 
    define $L_j(t, 0,  \FF_i) = L_j(t, 0+,  \FF_i)$, $L_j( t, 1, \FF_i) = L_j( t, 1-, \FF_i)$. There exists a constant $\chi_0 \in (0,1)$, and $s^*\geq 4$, for all $j$ such that $$
    \delta_{s^*, i}(L_j) = O(\chi_0^i),
    $$
    and there exists a constant $M_0$ such that 
    \begin{align}
       \sup_{t, u \in [0,1]} \| L_j( t, u, \FF_i)\|_{s^*} < M_0.
    \end{align}
    \item[(3)]For each fixed $u \in [0,1]$, the function $t \to \mf m(t, u)$ is third-order
 differentiable and there exists a positive constant $M_1$ such that
\begin{align}
    \sup_{t, u \in [0,1]}\Big |\frac{\partial^3}{\partial^3 t} \mf m(t,u) \Big |_F \leq M_1.
\end{align}
\item[(4)] There is a positive constant $c$ such that for all $u \in [0,1]$, $\lambda_{\min}(\bs \Gamma_0(t,u)\bs \Gamma_0^{\top}(t,u)) \geq c > 0$, where, for a symmetric matrix $\mf A$,  $\lambda_{\min} (\mf A)$ denotes   the smallest eigenvalue of  $\mf A$. 
\end{enumerate}
\label{ass:expo}
\end{assumption}
\HDB{Assumption \ref{ass:expo}(1) ensures sub-Gaussian tails of the error process. The second condition of Assumption  \ref{ass:expo} guarantees the dependence and moment condition for  the derivative process with respect to the second argument of the function $\mathbf H$ (which is defined almost surely).   It is also satisfied if  $\mf H$ is almost surely differentiable with respect to its second argument with uniformly bounded derivatives, i.e., $\sup_{t,u\in[0,1]}
\left|\partial_u \left|\mathbf H(t,u,\mathcal F_i)\right|_F\right|
\leq L < \infty $ almost surely.  Assumption \ref{ass:expo}(3) requires the smoothness of the mean function with respect to $t$, which is a regularity condition for kernel estimation for the mean function.  Assumption \ref{ass:expo}(4)  guarantees  that the  matrix $\bs \Gamma_0(t,u)\bs \Gamma_0^{\top}(t,u)$ is non-degenerate, which is related to the variance of $Q_n$.}

Under this assumption,  we provide an upper bound for the  joint cumulants of the locally stationary functional time series, which implies a uniform summability condition. This will be crucial for the analysis of the expectation and variance of the proposed statistic and hence the Gaussian approximation; see \cref{lm:cumk} below for more details.

\begin{assumption}
There exists positive constants $g_1, g_2$ such that 
\begin{align}
   \sup_{t \in (0,1), u \in [0,1]}  \sum_{k=1}^{\infty}|\bs \Gamma_k(t,u)|_F^2 < g_1,
\end{align}
for $t \in [0,1]$, the derivatives of $\sum_{k=1}^{\infty}|\bs \Gamma_k(t,u)|_F^2$ with respect to $u$ are uniformly bounded, and for $t_1, t_2 \in (0,1)$,
\begin{align}
    \sup_{ u \in [0,1]}  \left|\sum_{k=1}^{\infty}\{|\bs \Gamma_k(t_1,u)|_F^2 -|\bs \Gamma_k(t_2,u)|_F^2\}\right|\leq g_2|t_1-t_2|.
\end{align}
\label{ass:diff}
\end{assumption}

\HDB{The first part of Assumption \ref{ass:diff}
 guarantees that the autocovariance matrices are summable in terms of the Frobenius norm (uniformly with respect to $u,t$), while the second part  of Assumption \ref{ass:diff} ensures that the sum is uniformly Lipschitz  continuous with respect to $t$. It is  satisfied under  weak spatial-temporal dependence. For example, the auto-covariance function  $\mathbf \Gamma_k (t,u)  = \chi^k  \sin(t) \cos(u) \mathbf \mathbf \mathrm{diag}(\rho, \rho^2, \cdots, \rho^p)$ satisfies the two conditions in Assumption \ref{ass:diff} if $\kappa,\rho \in (0,1)$.  }

In the following, we establish the Gaussian approximation scheme for the distribution of the statistic $Q_n$.
For this purpose, we introduce the analogues of vectors \eqref{eq:hatV1}, 
where the nonparametric residuals have been replaced by the true errors in model
\eqref{eq:model}, that is 
\begin{align}
    {\mf V}_{i,l}:= K_{\tau}(i/n - l/n)  \big (\mathrm{tr} ({\mf U}_{i}(l/n, j/N) \big )_{1 \leq j \leq N} \in \mathbb{R}^N   \label{eq:defVil}
\end{align}
with
\begin{align}
     \mf U_i(t, u) =\sum_{k=1}^{s_n}(\bs \phi_{i,k}(u)\bs \psi^{\top}_{i,k}(t,u) + \bs \psi_{i,k}(t,u)\bs \phi^{\top}_{i,k}(u)) ,\label{eq:defU}
 \end{align}
and 
\begin{align}
   & \bs \phi_{i,k}(u) = \bs \epsilon_{i-k}(u) \bs \epsilon^{\top}_{i}(u) , \\ 
    & \bs \psi_{i,k}(t,u) =\sum_{\ell = (i-s_n) \vee 1}^{i- M_n - 1}  \bs \phi_{\ell ,k}(u)K_{\tau}(\ell /n - t).
    \end{align}
    Finally, we define
the vector 
\begin{align}
\label{det27}
\mf V_i = ({\mf V}_{i,\nt} ,\cdots, {\mf V}_{i + n - 2\nt, n-\nt} )^{\top},
\end{align}
as the analogue of \eqref{det25}.

For some  $q^{\prime} \geq 2$, for some $2\leq q\leq s^*/2$, define the quantity
\begin{align}
      f_n & ={s_n} (nN)^{1/q} ((n\tau)^{-1} + N^{-1}) + s_n \tau^{-2/q^{\prime}} \sqrt{n\tau}\left\{b^2 + (nb)^{-1}\right\} + \tau^{-2/q^{\prime}}\sqrt{s_n \tau/b}.\label{eq:fn}
\end{align}

\begin{theorem}\label{thm:Qn}\label{thm:ga}
If  Assumption \ref{ass:expo} 
is satisfied, $N  = O(n^{\alpha})$, $\alpha \geq 0$, for some $0 \leq \theta < 1/11$, $s_n^{2} = o((n\tau)^{\theta})$,  $\tau \to 0$, $f_n \to 0$,  $(\log n)^4/(nb) \to 0$, $M_n > c \log n$, $M_n = o(s_n)$, $s_n/\log n \to \infty$. Then,  under the null hypothesis there exists a sequence of  Gaussian process $({\mf Z}_k)_{k=1}^{2\nt}$ which shares the autocovariance structure of $({\mf V}_{k})_{k=1}^{2\nt} \in \mathbb R^{N(n-2\nt+1)}$ defined in \eqref{det27} such that 
    \begin{align}
    \sup_{x \in \mathbb R}\Big | \mathbb{P}\Big( \frac{1}{\sqrt{s_n}}Q_n \leq x\Big ) - \mathbb{P}\Big (\frac{1}{ s_n \sqrt{n\tau} }\Big |\sum_{i=1}^{2\nt}{\mf Z}_i \Big |_{\infty} \leq x\Big )\Big| = o(1),\nonumber
\end{align}
where the statistic $Q_n$ is  defined in \eqref{eq:testQ}. 
\end{theorem}

\HDB{ Theorem \ref{thm:Qn} provides the basis for the following Algorithm \ref{alg:boot2}, which calculates the quantile $\tilde r_{\bt}$ for the decision rule \eqref{hol1}.}

\begin{algorithm} 
\caption{Difference-based block multiplier bootstrap}\label{alg:boot2}
\begin{algorithmic}
\Require The statistics $ Q_n$
and the vectors $\hat{\mf V}_{r,l}$ defined in 
\eqref{eq:testQ}  and \eqref{eq:hatV1}, respectively.
\State For window size $L$, $L \leq j \leq 2\nt -L, 0 \leq l \leq n-2\nt$, calculate  the vectors 
\begin{align}
    \tilde{\mf S}_{j,l} \gets \frac{1}{\sqrt{2s_n^2 L}} \Big ( \sum_{r = j}^{j + L - 1} \hat{\mf V}_{r,l} - \sum_{r = j-L}^{j- 1} \hat{\mf V}_{r,l} \Big ).\label{eq:tildeS}
\end{align}
\For{$r = 1, \cdots, B$}
        \State Generate independent standard normal random variables $R_j^{(r)}$, $j=1,\cdots,n$. 
         \State  Calculate 
          \begin{align}
    \label{det26}
          \tilde Z_{\bt}^{(r)} =\frac{\underset{0 \leq l \leq n- 2\nt}{\max}  \left| \sum_{j= L}^{ 2\nt - L} {\tilde{\mf S}}_{j,l} R^{(r)}_{l+j}\right|_{\infty}}{\sqrt{ \nt  - L}}.
                \end{align}
    \EndFor
    \State Define   $\tilde  r_{\bt}$ as  the empirical $(1-\alpha)$-quantile of the bootstrap samples $\tilde Z_{\bt}^{(1)}, \cdots, \tilde Z_{\bt}^{(B)}$.
\end{algorithmic}
\end{algorithm}

\begin{remark}\label{lm:cumk}\label{cor:cumk}
{\rm  An important step in the proof of Theorem \ref{thm:ga} is the calculation of the variance  of the $U$-statistics in \eqref{eq:U}. In particular, this requires  the evaluation  of the covariances $\cov(G_k^{r_1, r_1}(t,u), G_h^{r_2, r_2}(t,u))$, $1 \leq r_1, r_2 \leq p$, $1 \leq k, h \leq s_n$,
where $G_k^{r_1, r_1}(t,u)$ denotes the  element in the position $(r_1, r_1)$  of the matrix $\mf G_k(t, u)$, which is  the counterpart of the matrix $\widehat{\mf G}_k(t, u)$
defined in 
\eqref{eq:hatgamma} with $\hat{\bs \epsilon}_{i}(u)$ substituted by ${\bs \epsilon}_{i}(u)$.
The covariances  can be further written in the form 
\begin{align}
  &\cov(G_k^{r_1, r_1}(t,u), G_h^{r_2, r_2}(t,u))]\\ & =\mathbb{E}(G_k^{r_1, r_1}(t,u) G_h^{r_2, r_2}(t,u)) -  \mathbb{E}(G_k^{r_1, r_1}(t,u))\mathbb{E}( G_h^{r_2, r_2}(t,u)) \\
  &\approx \sum_{i,j=1}^n K_{\tau}(i/n-t) K_{\tau}(j/n-t)   \\ &\times  \sum_{l \in \mathbb L_i, l^{\prime} \in \mathbb L_j} \left\{\mathrm{Cum}(\phi_{i,k}^{r_1, q_1}(t,u), \phi_{l,k}^{r_1, q_1}(t,u), \phi_{j,h}^{r_2, q_2}(t,u), \phi_{l^{\prime},h}^{r_2, q_2}(t,u)) \right.\\ & \left.+ \E(\phi_{i,k}^{r_1, q_1}(t,u)\phi_{j,h}^{r_2, q_2}(t,u))\E( \phi_{l,k}^{r_1, q_1}(t,u)\phi_{l^{\prime},h}^{r_2, q_2}(t,u) )\right.\\ & \left.+ \E(\phi_{i,k}^{r_1, q_1}(t,u)\phi_{l^{\prime},h}^{r_2, q_2}(t,u))\E( \phi_{l,k}^{r_1, q_1}(t,u)\phi_{j,h}^{r_2, q_2}(t,u) ) \right\},
\end{align}
where $  \mathrm{Cum}(\phi_{i,k}^{r_1, q_1}(t,u), \phi_{l,k}^{r_1, q_1}(t,u), \phi_{j,h}^{r_2, q_2}(t,u), \phi_{l^{\prime},h}^{r_2, q_2}(t,u))$
denotes  the cumulant of 
\linebreak 
$\phi_{i,k}^{r_1, q_1}(t,u),$ $ \phi_{l,k}^{r_1, q_1}(t,u),$ $ \phi_{j,h}^{r_2, q_2}(t,u)$ and  $ \phi_{l^{\prime},h}^{r_2, q_2}(t,u))$
\citep[for the definition and simple properties of cumulants we refer to monograph of][]{Brillinger1965}.
To investigate covariances of this and other type we will derive  in the appendices (see \cref{sec73} for details) 
 bounds on the cumulants appearing in these  expressions, which  are of own interest and   generalize corresponding results  for  stationary processes    \citep[see, for example, Proposition 2 in][] {wu2004limit}. 
\\
More precisely, we have under  \cref{ass:expo} for  a fixed number $K$ and $0 \leq m_0 \leq \cdots \leq  m_{K-1}$, $1 \leq r_0 \leq \ldots \leq r_{K-1} \leq p$,  
    \begin{align}
      & \sup_{u \in [0,1]}\sup_{0 \leq t_0,\ldots, t_{K-1} \leq 1} |\mathrm{Cum}(H_{r_0}(t_0, u , \FF_{m_0}) , \ldots, H_{r_{K-1}}(t_{K-1}, u , \FF_{m_{K-1}}))|\\ &\leq C \chi^{(m_{K-1}-m_0)/(2K(K-1))},
    \end{align}
    where the positive constant $C$ is independent of $m_0, \ldots, m_{K-1}$, $\chi\in (0,1)$.
    Consequently, we obtain 
    \begin{align}
        &\sum_{s_1,\ldots, s_{K-1} \in \mathbb Z}  \sup_{u \in [0,1]}\sup_{0 \leq t_0,\ldots, t_{K-1} \leq 1} |\mathrm{Cum}(H_{r_0}(t_0, u , \FF_{m_0}) ,H_{r_1}(t_1, u , \FF_{m_0+s_1}),\\  & \ldots, H_{r_{K-1}}(t_{K-1}, u , \FF_{m_0+s_K}))| < \infty.
    \end{align}
    }
\end{remark}
 
\begin{remark} 
{\rm 
Note that we  derive a portmanteau-type test statistic with asymptotic zero mean under the null hypothesis and derive a Gaussian approximation of the maximum of an $L_2$-type statistic 
to approximate its distribution.
 One of our main technical contributions is  a uniform bound  for the dependence measure of the random variables    $\mf U_1(t,u), \ldots ,  \mf  U_{2\nt}(t,u)$  in \eqref{eq:defU}, which are sums of $s_n$ terms of fourth-order processes, see Lemma \ref{lm:delta1}. In addition, since the condition of sub-exponential tails is too restrictive for fourth-order processes, we extend the Gaussian approximation developed in \cite{zhang2018} to allow for high-moment conditions, which yields a polynomial rate for $N$.  
}

\end{remark}

To investigate the consistency of the test,  we consider for a  positive constant $c_2$,
the  alternative hypothesis
\begin{align}
    H_A:  \underset{\substack{ u\in [0,1] \\ t
    \in [\tau, 1-\tau] } }{\sup} \sum_{k =1}^{s_n} |\bs \Gamma_k(t, u)|_F^2 > c_2. \label{eq:Ha}
\end{align}
  Note that this alternative hypothesis contains dependent non-stationary multivariate functional time series, such as time-varying functional AR (autoregressive) model and the time-varying functional AR-GARCH(1,1) process.

\begin{theorem}\label{thm:alt}
Under \cref{ass:diff}, assume that  the conditions of Theorem \ref{thm:Qn} hold and that  $s_n/\sqrt{n\tau^3} \to 0$. Then,  under the alternative hypothesis \eqref{eq:Ha} with probability approaching $1$, we have 
\begin{align}
    Q_n/\sqrt{s_n (n \tau)^d }\to \infty 
\end{align}
 for any $d \in (0,1)$.
\end{theorem}

By Theorem \ref{thm:Qn}, we obtain an asymptotic level $\alpha$ for the hypothesis \eqref{eq:H0}
by rejecting $H_0$, whenever 
$Q_n/\sqrt{s_n}$ exceeds the $(1-\alpha)$-quantile, say  $q_{1-\alpha}$, of the distribution of
$\frac{1}{n\tau \sqrt{ s_n} }\left|\sum_{i=1}^{2\nt}{\mf Z}_i \right|_{\infty}$. 
\cref{thm:alt} guarantees that $Q_n/\sqrt{s_n}$ will diverge to infinity faster than $(n \tau)^d$, for any $d \in (0,1)$, which implies consistency of this decision rule.
The quantile $q_{1-\alpha}$ is not available and will be estimated by \cref{alg:boot2}, resulting in the bootstrap test \eqref{hol1}. In order to show that this test  is consistent, we 
define for some $q \geq 4$ the quantities 
\begin{align}
 & \vartheta_n = \frac{s^4_n(\log n)^2}{L} + \frac{s_n^3L\log n}{n\tau}+ \sqrt{\frac{L}{n\tau}}s_n^2 (nN)^{4/q} \label{eq:vtheta_n},\\
    & l_n =  (nb)^{-1/2} b^{-1/q} + b^2.
\end{align}
\begin{theorem}
  \label{thm:bs}
 If Assumptions \ref{ass:expo} and \ref{ass:diff} 
 are satisfied, $M_n > c\log n $, for some positive constant $c$, $s_n/(\log n) \to \infty$, then   the following results hold for the statistic $\tilde Z_{\bt}$ defined in \eqref{det26}: 
\par
(i) If the null hypothesis holds and there exists a sequence $\eta_n \to \infty$, $q \geq 4$, such that 
\begin{align}
  \sqrt{L s_n \log(nN) }\{l_n\eta_n(n^2N)^{1/q} +s_n/n\}+\vartheta_n^{1/3} \left\{1 \vee \log (nN/\vartheta_n) \right\}^{2/3} \to 0, \label{eq:bootstrap_ass}
\end{align}
 we have 
\begin{align}
       & \sup_{x \in \R} \Big |\mathbb{P} \Big  (\tilde Z_{\bt} < x|\FF_n \Big  ) -  \mathbb{P}\Big  (\frac{1}{s_n\sqrt{n\tau}}   \Big  | \sum_{i=1}^{2\nt} {\mf Z}_i  \Big  |_\infty  < x\Big )  \Big | = \op( 1  ),
\end{align}
 where $\FF_n$ denotes the $\sigma$-field generated by $\mf X_{1,n} ,  \ldots , \mf X_{n,n}$.

(ii) If  the alternative hypothesis \eqref{eq:Ha} 
holds,  $s_n (\log n)^2 /L^2 \to 0$ and 
\begin{align}
\sqrt{L  s_n\log(nN) }  l_n\eta_n (n^2N)^{1/q} + (L\log (n  N))^{1/2}\{ L/(n\tau) + Ls_n/n\} = o(s_n \log(nN)), \label{eq:bootstrap_ass1}
\end{align} 
we have  
\begin{align}
     \tilde Z_{\bt} 
     = \Op( s_n\log (nN) ).
\end{align}
\end{theorem}
Conditions \eqref{eq:bootstrap_ass} and  \eqref{eq:bootstrap_ass1}  are satisfied for $N = O(n^{\lambda})$, $\lambda > 0$, $b \asymp n^{-1/4} $, $\tau \asymp  n^{-2/5}$, $s_n = o(n^{1/55})$, $L \asymp  \lfloor n^{1/5} \rfloor$ and sufficiently large $q$, where $\lfloor x \rfloor$ is the largest integer that is smaller than or equal to $x$.
 The following result justifies the theoretical validity of the proposed test under the null and alternative hypothesis and  is a direct consequence  of  \cref{thm:Qn}, \cref{thm:alt} and \cref{thm:bs}.

\begin{theorem}
\label{thm:alg}
Suppose that the conditions of \cref{thm:Qn}, \cref{thm:alt} and \cref{thm:bs} are satisfied,   then 
the following assertions are true.
\\
(i) Under the null hypothesis \eqref{eq:H0}, we have
\begin{align}
    P(Q_n/\sqrt{s_n} > \tilde r_{\bt}|\mathcal F_n) \to \alpha
\end{align}
in probability, where $\FF_n$ denotes the $\sigma$-field generated by $\mf X_{1,n} ,  \ldots , \mf X_{n,n}$.
\\
(ii) Under the alternative hypothesis \eqref{eq:Ha}, we have
\begin{align}
   P(Q_n/\sqrt{s_n} > \tilde r_{\bt}|\mathcal F_n) \to 1
\end{align}
in probability.
\end{theorem}


\section{Finite sample properties} 
\label{sec5} 
 \def\theequation{4.\arabic{equation}}	
  \setcounter{equation}{0}

  \subsection{Selection of tuning  parameters}\label{sec:sel}
As a rule-of-thumb  we propose to chose  $s_n$ and $M_n$ as  $s_n = \lfloor (\log n)^2/6 \rfloor$  and  $M_n = \lfloor (\log n)/5 \rfloor$, 
 respectively. To select $b_n$, we adopt the Generalized Cross Validation (GCV) proposed by \cite{craven1978smoothing}.
To be precise, let  $\hat{\mf X}_{i,n}(j/N) =  (\hat{ X}_{i,n,1}(j/N), \ldots, \hat{X}_{i,n,p}(j/N))^\top = \hat {\mf m}(i/n, j/N)$ denote the estimator of $ \mf m (i/n, j/N) $  defined in  \eqref{det30} with bandwidth $b$. For $l=1,\cdots, p$, write $\hat{\mf X}_l(j/N) = (\hat X_{1,n,l}(j/N), \ldots,  \hat X_{n,n,l}(j/N))^{\T}$. Note that we can represent this estimator  in the form  $\hat{\mf X}_l(j/N) = \mf Q(b)\mf X_l(j/N)$ for some square matrix $\mf Q$, where $\mf X_l(j/N) = (  X_{1,n,l}(j/N), \ldots,   X_{n,n,l}(j/N))^{\T}$.
Then we select $\hat b_n$ by
   \begin{equation}
     \hat{b}_{n}=\underset{b\in [b_L ,b_U ]}{\arg \min }\{\operatorname{GCV}(b)\},    \end{equation}
     where 
        \begin{equation}
        \operatorname{GCV}(b)=\frac{n^{-1}\sum_{j=1}^N\sum_{l=1}^p |\mathbf{X}_l(j/N) -\hat{\mathbf{X}}_l(j/N)|^{2}}{[1 - n^{-1}\operatorname{tr}\{\mf Q(b)\}]^{2}}.\label{GCV}
     \end{equation}
Following a similar approach, we can also obtain a data-driven criterion for choosing  $\tau$ in \eqref{eq:ganulla}. 
To be precise, recall the definition of  $ \hat{\bs \Gamma}_k(\cdot, \cdot)$ in \eqref{det33} and denote its vectorization by  
$$
\hat{\hat{\bs \Gamma}}_k(j/N, \tau)=(\mathrm{vec}(\hat{\bs \Gamma}_k({(s_n+1)}/n, j/N)), \ldots, \mathrm{vec}(\hat{\bs \Gamma}_k(n/n, j/N)))^{\T}.
$$
Note that 
for some square matrix $\mf Q_1$,
$$
\hat{\hat{\bs \Gamma}}_k(j/N, \tau) = \mf Q_1(\tau) \mf E_k(j/N)~,
$$
 where $\mf E_k(j/N) = (\mathrm{vec}(\hat{\bs \epsilon}_{{s_n+1}-k,n}(j/N)\hat{\bs \epsilon}^{\top}_{h,n}(j/N) ), \ldots, \mathrm{vec}(\hat{\bs \epsilon}_{n-k,n}(j/N)\hat{\bs \epsilon}_{n,n}^{\top}(j/N)))^{\T}$.
If $p\leq 3 $  we propose to select $\hat \tau_n$ by  
   \begin{equation}
     \hat{\tau}_{n}
     =\underset{\tau\in [\tau_L ,\tau_U]}{\arg \min } \frac{n^{-1}\sum_{j=1}^N\sum_{k=1}^{s_n} |\mathbf{E}_k(j/N) -\hat{\hat{\bs \Gamma}}_k(j/N, \tau)|_F^{2}}{[1-n^{-1} \operatorname{tr}\{\mf Q_1(\tau)\}]^{2}}.\label{GCV1}
     \end{equation}
To save computational time and memory, we use a slightly different rule for $p>3$, i.e.,
\begin{align}
     \hat{\tau}_{n}
     =\underset{\tau\in [\tau_L ,\tau_U]}{\arg \min }\frac{n^{-1}|\sum_{j=1}^N\sum_{k=1}^{s_n} (\mathbf{E}_k(j/N) - \hat{\hat{\bs \Gamma}}_k(j/N, \tau))|_F^{2}}{[1-n^{-1} \operatorname{tr}\{\mf Q_1(\tau)\}]^{2}}.
\end{align}

Finally, for the choice  of $L$, we employ the MV method (extended minimum volatility), proposed in Chapter 9 of \cite{politis1999subsampling}. To be precise, we define 
$$\hat \gamma(L) = \sum_{j= L}^{ 2\nt - L} |{\tilde{\mf S}}_{j,l}|^2/(2\nt - 2L),$$ and calculate this  function on  a grid, say, $\{L_1, \cdots, L_M\}$ and obtain $(\hat \gamma(L_i))_{i=1}^M$. Finally, we choose $L_i$ such that the index $i \in \{ 2, \ldots , M-1 \} $ minimizes the  objective function
\begin{align}
2^{-1}\sum_{k=-1}^1\Big (\hat \gamma(L_{i+k}) - \sum_{k=-1}^1 \hat \gamma(L_{i+k})/3\Big )^2 .
\end{align}
  \subsection{Simulation study}
  \label{sec51}
We consider the following functional time series model
\begin{align}
    \mf X_{i,n}(u) = \mf  m(i/n, u) + \bs \epsilon_{i,n}(u),\label{eq:simulation}
\end{align}
where $\mf m(t, u) = (m_1(t,u), \cdots, m_p(t,u))^{\top}$, $\bs \epsilon_{i,n}(u) = (\epsilon_{i, n, 1}(u), \cdots, \epsilon_{i, n, p}(u))^{\top}$.

Let $(\eta_{i,j})_{i\in \mathbb Z, 1 \leq j \leq p}$ denote a sequence of $i.i.d.$ standard normal random variables. Define $\bs \eta_i = (\eta_{i,1},\ldots, \eta_{i,p})^{\top}$, $\F_{i,j} = (\ldots,  \eta_{i-1, j},  \eta_{i, j})$, and $\F_i = (\ldots, \bs \eta_{i-1}, \bs \eta_i)$. In the simulation study, we consider the following settings for model \eqref{eq:simulation}.  
\medskip

\noindent
\textbf{Model 1:} The trend functions are 
given by
\begin{align}
    m_j(t,u) =(1+u)\{10\sin(\pi(t-0.5)) + 1\},  ~~~(j=1, \ldots , p).\label{eq:trendm1}
\end{align}
Under the null hypothesis, the error processes $\bs \epsilon_{i,n}$ are generated from 
the model 
\begin{align}
   \bs \epsilon_{i,n}(u) = \mf  H(i/n, u, \F_i)  = \{0.1(u-0.5)^2+0.8 \} \bs \eta_{i}. \label{eq:model1}
\end{align}
\HDB{Note that the error process is a stationary functional time series and that the right hand side of \eqref{eq:model1} does not depend on the sample size $n$. Although the index  $n$ could be omitted in this case, we prefer to keep the notation $\boldsymbol{\epsilon}_{i,n}$ to be consistent with the general model consider in this  paper.}
Under the alternative they are given by the functional AR-processes
 \begin{align} \label{alt1}
     \bs \epsilon^{H_A}_{i,n}(u) =  6 A_1 \sin\{\pi i/(2n)\} (u-0.5)^2
     \bs \epsilon^{H_A}_{i-1,n}(u) + \bs \epsilon_{i,n}(u),
 \end{align}
 where 
 $ A_1 =( A_{1,i,j}) $ is a tridiagonal $p \times p$ matrix with elements  $A_{1,j,j} = \delta$, $A_{1,j,j-1} = A_{1,j-1,j} = 0.1$, when $p>3$, and $A_1 = I$ otherwise. 
\medskip

\noindent
\textbf{Model 2:} The trend functions are again given by  \eqref{eq:trendm1}.  For $\bs   \epsilon_{i,n}(u)$, 
 assign each  component $ \epsilon_{i,n,j}(u) = H_j(i/n, u, \FF_i)$ 
 with equal 
probability to one of the two  models: 
\begin{align}
    Z_1(t,u, \FF_{i,j}) = a_1(u) \sigma(t) \eta_{i,j},\quad Z_2(t,u, \FF_{i,j}) = a_2(u) \sigma(t) \eta_{i,j},
    \label{eq:model2}
\end{align}
 where $j=1, \ldots , p$, and 
\begin{align}
    a_1(u) = 0.5\cos(\pi u/3) + 0.2, \quad a_2(u) = 0.4 u, \quad \sigma(t) = 0.8+ 0.2\sin(\pi t/2).
    \label{eq:model2a}
\end{align}
Under the null hypothesis, we set the error processes as $\bs   \epsilon_{i,n}(u)$.
Under the alternative hypothesis, 
the error processes are given by 
\eqref{alt1},  
where the innovation processes  $\bs \epsilon_{i,n}(u)$  are generated according to \eqref{eq:model2}.
\medskip

\noindent
\textbf{Model 3: }
Here we consider a time-varying functional GARCH(1,1) model which is defined as follows. The trend functions are given in \eqref{eq:trendm1}.   
Under the null hypothesis,  each  component $ \epsilon_{i,n,j}(u) = H_j(i/n, u, \FF_i)$ of the vector $ \bs  \epsilon_{i,n}(u)$ is assigned with equal 
probability to one of the two  models 
\begin{align}
    Z_3(t,u, \FF_{i,j}) =  a_1(u) \sigma_{3,i}(t) \eta_{i,j},\quad Z_4(t,u, \FF_{i,j}) = a_2(u) \sigma_{4,i}(t) \eta_{i,j},
\end{align}
 where $j=1, \ldots , p$,  $a_1$ and $a_2$ are defined in \eqref{eq:model2a} and for $k=3,4$,
$$\sigma_{k,i}^2(t) = 0.9+0.1\cos(\pi/3 + 2\pi t) + (0.1+0.2t) (\sigma_{k,i-1}(t)\eta_{i-1,j})^2 + (0.1 + 0.2t)\sigma^2_{k, i-1}(t).
$$  
Under the alternative hypothesis, 
we  consider a  functional AR-type processes of the form  
\begin{align}
     \bs \epsilon^{H_A}_{i,n}(s) =  A_{2}(i/n)\int \exp(-(u^2 + s^2)/2) \bs \epsilon^{H_A}_{i-1,n}(u)du  + \bs \epsilon_{i,n}(s),
 \end{align}
 where $ A_{2}(t) =\delta(1+0.25 \cos(\pi t/2))/c$, with $c = \int \exp(-(u^2 + s^2)/2) duds$.

\bigskip

We first illustrate the performance of the test \eqref{hol1} for small dimension, that is $p=5$.  
In Table \ref{tab1} and \ref{tab2}, we show the empirical rejection probabilities of the test \eqref{hol1} in Model 1  - 3 for various sample sizes under the null hypothesis and one  alternative ($\delta=0.4$), respectively.  We observe that the nominal level is very well approximated and that the test has reasonable power. Note that the alternatives in Model 3 are  a little easier to detect.

\begin{table}[ht]
\centering
\begin{tabular}{rrrrrrr}
  \hline
&\multicolumn{2}{c}{Model 1}&\multicolumn{2}{c}{Model 2}&\multicolumn{2}{c}{Model 3}\\
$n$ & 5\% & 10\% & 5\% & 10\% & 5\% & 10\% \\ 
  \hline
200 & 4.8 & 9.5 & 5.0 & 10.1 & 4.1 & 10.0 \\ 
  400 & 4.3 & 9.4 & 5.7 & 11.9 & 4.8 & 10.5 \\ 
  600 & 3.7 & 7.5 & 4.8 & 9.0 & 3.4 & 8.1 \\ 
  800 & 4.3 & 9.3 & 5.3 & 10.2 & 6.2 & 11.7 \\ 
  1000 & 4.9 & 8.9 & 5.0 & 10.6 & 5.3 & 10.7 \\ 
   \hline
\end{tabular}
\caption{\it Approximation of the nominal level of the test \eqref{hol1} for different sample sizes in Model 1 - 3. The dimension is  $p=5$ and   $N=50$  points have been used for a grid in the interval $[0,1]$.}
\label{tab1}
\end{table}

\begin{table}[ht]
\centering
\begin{tabular}{rrrrrrr}
  \hline
  &\multicolumn{2}{c}{Model 1}&\multicolumn{2}{c}{Model 2}&\multicolumn{2}{c}{Model 3}\\
$n$ & 5\% & 10\% & 5\% & 10\% & 5\% & 10\% \\ 
  \hline
200 & 7.7 & 13.4 & 10.7 & 19.2 & 61.5 & 71.6 \\ 
  400 & 31.1 & 48.7 & 21.0 & 31.9 & 86.9 & 88.9 \\ 
  600 & 77.6 & 88.4 & 45.3 & 61.2 & 92.2 & 93.5 \\ 
  800 & 95.1 & 98.2 & 77.0 & 86.3 & 93.4 & 95.7 \\ 
  1000 & 98.4 & 99.6 & 87.7 & 93.9 & 95.2 & 96.7 \\  
   \hline
\end{tabular}
\caption{\it Empirical rejection probabilities of the test \eqref{hol1} under the alternative for  different sample sizes. Here the dimension is $p=5$, $\delta$ in Model 1 - 3 is chosen as $ 0.4$ and $N=50$  points have been used for a grid in the interval $[0,1]$.}
\label{tab2}
\end{table}

In order to study the finite-sample performance of the test for larger dimensions,  we consider the cases  $p=10$ and $p=20$.  The corresponding results are displayed in Table \ref{tab3}. We observe a reasonable approximation of the nominal level for   dimension $p=10$.  If $p=20$, the simulated level is slightly too large for the sample size $n=200$ in Model 1 and 2, but the approximation improves quickly. with increasing sample size. 
In Model 3, the nominal level  for $p=20$ is overestimated for sample sizes smaller than $n=800$, but the approximation improves for $n=1000$.

As for the effect of the dimension on the power of the test \eqref{hol1},  we display   simulated rejection probabilities of Model 1-3 in Figure \ref{fig1}, where the 
sample size is chosen as  $n=400$ 
and  $N=50$  points have been used for a grid in the interval $[0,1]$.
For all models, the power of the test \eqref{hol1} increases with $\delta$ and 
the  dimension $p$. We observe lower power for the alternatives considered in Model 2.
Comparing Model 1 and 3, we see that for  $\delta \geq 0.5 $ the power is comparable, while for $\delta < 0.5$ the alternatives in model 3 yield larger rejection probabilities.

\begin{table}[ht]
\centering
\begin{tabular}{rrrr r   rrrr   rrrr}
  \hline
  &\multicolumn{4}{c}{Model 1} &\multicolumn{4}{c}{Model 2}&\multicolumn{4}{c}{Model 3}\\
  &\multicolumn{2}{c}{$p=10$}&\multicolumn{2}{c}{$p=20$}  &\multicolumn{2}{c}{$p=10$}&\multicolumn{2}{c}{$p=20$}&\multicolumn{2}{c}{$p=10$}&\multicolumn{2}{c}{$p=20$}\\
    \hline
 & 5\% & 10\% & 5\% & 10\% & 5\% & 10\%  & 5\% & 10\%  & 5\% & 10\% & 5\% & 10\%  \\ 
  \hline
200 & 5.5 & 9.8 & 7.0 & 13.6 & 6.3 & 11.5 & 6.1 & 12.7 & 5.9 & 13.0 & 6.1 & 12.9 \\ 
  400 & 4.2 & 9.8 & 5.8 & 10.7 & 4.7 & 9.6 & 6.1 & 12.5 & 6.0 & 10.8 & 8.8 & 14.8 \\ 
  600 & 4.6 & 8.7 & 4.9 & 9.6 & 4.3 & 9.4 & 5.7 & 10.4 & 3.5 & 8.9 & 7.5 & 15.8 \\ 
  800 & 4.9 & 9.7 & 5.6 & 11.0 & 4.2 & 9.1 & 5.1 & 10.1 & 4.8 & 9.2 & 7.7 & 14.1 \\ 
  1000 & 4.3 & 8.2 & 4.8 & 9.2 & 5.2 & 8.8 & 4.3 & 9.7 & 4.1 & 8.8 & 5.0 & 11.1 \\ 
   \hline
\end{tabular}
\caption{ \it 
 Approximation of the nominal level of the test \eqref{hol1} for different sample sizes in Model 1 - 3. The dimensions are   $p=10$ and $p=20$ and  $N=50$ points have been used for a grid in the interval $[0,1]$.
}
\label{tab3}
\end{table}

\begin{figure}[ht]
    \centering
    \includegraphics[width = 0.45\linewidth]{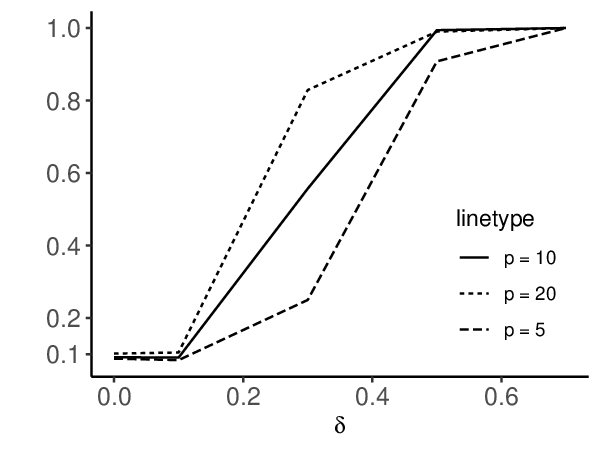} 
     \includegraphics[width = 0.45\linewidth]{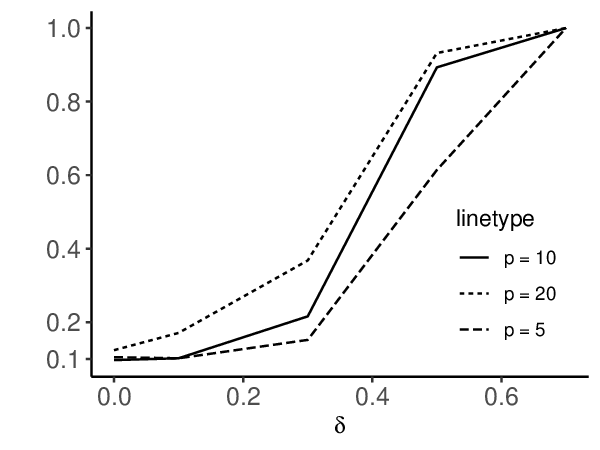} \\
      \includegraphics[width = 0.45\linewidth]{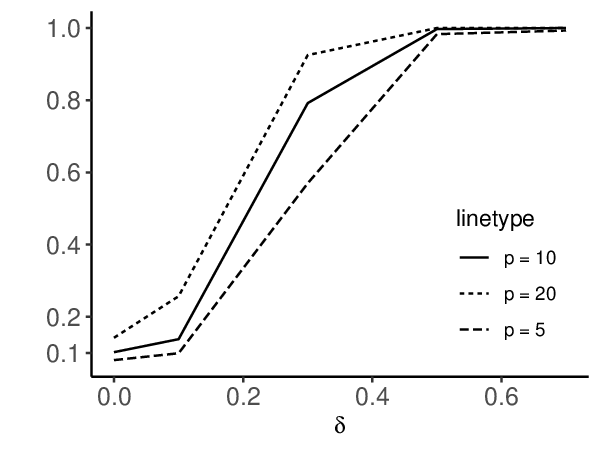}

    \caption{ \it  Simulated rejection rates of the test \eqref{hol1} in  Model 1 (upper left panel), Model 2 (upper right panel) and Model 3 (lower panel). The   nominal level is $\alpha=0.1$, the sample size is  $n=400$ and the number of grid points is $N=50$. }
    \label{fig1}
\end{figure}

\section{An extension to covariance functions including  \HDB{cross dependence between functional locations}  }\label{sec:extend} 
 \def\theequation{5.\arabic{equation}}	
  \setcounter{equation}{0}
The developed methodology can be extended to hypotheses regarding the auto-covariance functions 
\begin{align}
\bs  \Gamma_k:  \left\{ 
\begin{array}{l}
       [0,1]^3 \to  \mathbb R \\
     (t, u, v) \mapsto     \bs \Gamma_k(t, u, v) = \mathrm{Cov}(\mf H(t, u,  \mathcal F_0), \mf H(t, v,\mathcal  F_k))
\end{array}
\right.  ~, 
\nonumber
\end{align}
which include the  \HDB{cross dependence between functional locations} $u$ and $v$.  To be precise, we consider the null hypothesis 
\begin{align}
    H_0:  \bs \Gamma_k(t, u, v) = \mf 0, \text{   for all  } ~k \in \mathbb N ,\ (t,u, v)^\top \in [0,1]^3.\label{eq:H0extend}
\end{align}
and the alternative
\begin{align}
    H_A:  \underset{\substack{ u, v\in [0,1]^2 \\ t
    \in [\tau, 1-\tau] } }{\sup} \sum_{k =1}^{s_n} |\bs \Gamma_k(t, u, v)|_F^2 > c_2, \label{eq:Haextend}
\end{align}
where $c_2$ is a  positive constant.
For testing the hypotheses \eqref{eq:H0extend} against \eqref{eq:Haextend}, we can directly extend the methodology  for  $\boldsymbol \Gamma_k(t,u)$  by adding an extra variable for $v$. As a counterpart of the statistic \eqref{eq:estimator}, we consider the estimator 
\begin{align}
\max_{1 \leq i \leq n, 1 \leq j,j^{\prime} \leq N} \sum_{k=1}^{s_n}\mathrm{tr}\big \{\widehat{\mf G}_{k} (i/n, j/N, j^{\prime}/N) \big \}, \label{eq:estimator_extend}
\end{align}
where 
\begin{align}
\widehat{\mf G}_{k}(t,u, v)   &= 
\frac{1}{n\tau s_n} \sum_{i=1}^{n} K_{\tau}(i/n-t) \hat{\bs \phi}_{i,k} (u, v)\sum_{l=1 }^{n} K_{\tau}(l/n-t)\hat{\bs \phi}^{\top}_{l,k}(u, v) \mf 1(l \in \mathbb L_i) , \\\label{eq:hatgammaextend}
\hat{\bs \phi}_{i,k}(u, v)& =\hat {\bs \epsilon}_{i-k}(u) \hat{\bs \epsilon}_i^{\top} (v),
\nonumber
\end{align}
and  $\mathbb L_i = \{l: M_n < |l-i| \leq s_n\}$. The set 
$\mathbb L_i $ is used to control the amount of dependence after multiplying the two sums as in \eqref{eq:hatgamma}. Note that 
$\mathrm{tr}\{\widehat{\mf G}_{k}(t,u, v)\}$ is an estimator of $2|\bs \Gamma_k(t,u, v)|_F^2$. 
In analogy  to \eqref{det33}, we define
 \begin{align}
 \hat{\bs \Gamma}_k(t,\cdot, \star) = \frac{1}{n\tau} \sum_{i=1}^n  K_{\tau} (i/n-t) \hat{\bs \phi}_{i,k}(\cdot, \star) = \frac{1}{n\tau} \sum_{i=1}^n\hat{\bs \epsilon}_{i-k}(\cdot)\hat{\bs \epsilon}^{\top}_i(\star) K_{\tau} (i/n-t),
 \end{align}
 which is the Nadaraya-Watson estimator of  $\bs \Gamma_k(t,\cdot, \star)$.

Then, the test statistic in \eqref{eq:testQ} and the bootstrap multiplier \eqref{eq:hatV1}  in \cref{alg:boot2} could be modified as 
\begin{align}
    Q_n = \underset{\nt \leq i \leq n - \nt, 1 \leq j, j^{\prime} \leq N}{\max}\Big  |\sqrt{n\tau s_n} \sum_{k=1}^{s_n}\mathrm{tr}\big \{\widehat{\mf G}_{k} (i/n, j/N, j^{\prime} /N)\big \} \Big  |,\label{eq:testQextend}
\end{align}
and 
\begin{align}
    \hat {\mf V}_{i,l}:= K_{\tau}(i/n - l/n)  \big (\mathrm{tr} (\hat{\mf U}_{i}(l/n, j/N, j^{\prime} /N) \big )_{1 \leq j, j^{\prime} \leq N} \in \mathbb{R}^{N^2},  \label{eq:hatV1extend}
\end{align}
respectively, 
where 
\begin{align}
    \hat{ \mf U}_i (t, u, v) =\sum_{k=1}^{s_n} \big \{   \hat{\bs \phi}_{i,k}(u,v) \hat{\bs \psi}^{\top}_{i,k}(t,u,v) +  \hat{\bs \psi}_{i,k}(t,u,v) \hat{\bs \phi}^{\top}_{i,k}(u,v)\big \} ~, 
 \end{align}
and 
 \begin{eqnarray}
 \hat{\bs \psi}_{i,k}(t,u,v) & = & \sum_{\ell = (i-s_n) \vee 1}^{i- M_n - 1}  K_{\tau}( \ell/n - t) \hat{\bs \phi}_{\ell, k}(u,v). 
 \end{eqnarray} 
The portmanteau test for the extended auto-covariance function $\bs \Gamma_k(t,u,v)$ is listed in \cref{alg:boot3}. The major difference between \cref{alg:boot2} and \cref{alg:boot3} is to discretize the extra dimension $v$ by replacing the test statistic \eqref{eq:testQ} and the vector for bootstrap multipliers \eqref{eq:hatV1} by \eqref{eq:testQextend} and \eqref{eq:hatV1extend}, respectively.  The effectiveness of \cref{alg:boot3} is established in the following \cref{thm:extend}, which is an extension of \cref{thm:alg}. Before we present this result we  introduce a counterpart assumption  \cref{ass:expo}(4) and \cref{ass:diff}.

\begin{algorithm} 
\caption{Difference-based block multiplier bootstrap for general covariance functions}\label{alg:boot3}
\begin{algorithmic}
\Require The statistics $ Q_n$
and the vectors $\hat{\mf V}_{r,l}$ defined in 
\eqref{eq:testQextend}  and \eqref{eq:hatV1extend}, respectively.
\State For window size $L$, $L \leq j \leq 2\nt -L, 0 \leq l \leq n-2\nt$, calculate  the vectors 
\begin{align}
    \tilde{\mf S}_{j,l} \gets \frac{1}{\sqrt{2s_n^2 L}} \Big ( \sum_{r = j}^{j + L - 1} \hat{\mf V}_{r,l} - \sum_{r = j-L}^{j- 1} \hat{\mf V}_{r,l} \Big ).
\end{align}
\For{$r = 1, \cdots, B$}
        \State Generate independent standard normal random variables $R_j^{(r)}$, $j=1,\cdots,n$. 
         \State  Calculate 
          \begin{align}
          \tilde Z_{\bt}^{(r)} =\frac{\underset{0 \leq l \leq n- 2\nt}{\max}  \left| \sum_{j= L}^{ 2\nt - L} {\tilde{\mf S}}_{j,l} R^{(r)}_{l+j}\right|_{\infty}}{\sqrt{ \nt  - L}}.
                \end{align}
    \EndFor
    \State Define   $\tilde  r_{\bt}$ as  the empirical $(1-\alpha)$-quantile of the bootstrap samples $\tilde Z_{\bt}^{(1)}, \cdots, \tilde Z_{\bt}^{(B)}$.
\end{algorithmic}
\end{algorithm}

\begin{assumption}
~~

(1) There is a positive constant $c^*$ such that for all $u,v \in [0,1]$, $\lambda_{\min}(\bs \Gamma_0(t,u,v)\bs \Gamma_0^{\top}(t,u,v)) \geq c^* > 0$. 
 
(2) There exists positive constants $g_3, g_4$ such that 
\begin{align}
   \sup_{t \in (0,1), u,v \in [0,1]}  \sum_{k=1}^{\infty}|\bs \Gamma_k(t,u,v)|_F^2 < g_3,
\end{align}
for $t \in [0,1]$, the derivatives of $\sum_{k=1}^{\infty}|\bs \Gamma_k(t,u,v)|_F^2$ with respect to $u$ and $v$ are uniformly bounded, and for $t_1, t_2 \in (0,1)$,
\begin{align}
    \sup_{ u,v \in [0,1]}  \left|\sum_{k=1}^{\infty}\{|\bs \Gamma_k(t_1,u,v)|_F^2 -|\bs \Gamma_k(t_2,u,v)|_F^2\}\right|\leq g_4|t_1-t_2|.
\end{align}
\label{ass:diff_extend}
\end{assumption}
As counterpart of the sequence $f_n$  in \eqref{eq:fn} we define,   for some  $q^{\prime} \geq 2$ and $2\leq q\leq s^*/2$,  the quantity
\begin{align}
      f_n^* & ={s_n} (nN^2)^{1/q} ((n\tau)^{-1} + N^{-1}) + s_n \tau^{-3/q^{\prime}} \sqrt{n\tau}\left\{b^2 + (nb)^{-1}\right\} + \tau^{-3/q^{\prime}}\sqrt{s_n \tau/b},
\end{align}
and as counterpart for $\vartheta_n$ \eqref{eq:vtheta_n} we consider, for some $q \geq 4$, 
\begin{align}
    \vartheta_n^* = \frac{s^4_n(\log n)^2}{L} + \frac{s_n^3L\log n}{n\tau}+ \sqrt{\frac{L}{n\tau}}s_n^2 (nN^2)^{4/q}.
\end{align}
Note that the effect of $(u, v)$ instead of $u$ is reflected in the factor in the sequences $N^2$ in $f_n^*$ and  $\vartheta_n^*$. This factor also appears in conditions \eqref{eq:bootstrap_ass_extend1} and \eqref{eq:bootstrap_ass_extend} of our final result, which are the analogs of \eqref{eq:bootstrap_ass} and \eqref{eq:bootstrap_ass1}, respectively.  

\begin{theorem}
\label{thm:extend}
Under the conditions of \cref{thm:alg} and \cref{ass:diff_extend}, and additionally assuming $f_n^* \to 0$,   
the following assertions are true for \cref{alg:boot3}.
\\
(i) Under the null hypothesis \eqref{eq:H0extend},  if there exists a sequence $\eta_n \to \infty$, $q \geq 4$, such that 
\begin{align}
  \sqrt{L s_n \log(nN^2) }\{l_n\eta_n(n^2N^2)^{1/q} +s_n/n\}+(\vartheta_n^*)^{1/3} \left\{1 \vee \log (nN^2/\vartheta^*_n) \right\}^{2/3} \to 0, \label{eq:bootstrap_ass_extend1}
\end{align} we have
\begin{align}
    P(Q_n/\sqrt{s_n} > \tilde r_{\bt}|\mathcal F_n) \to \alpha
\end{align}
in probability, where $\FF_n$ denotes the $\sigma$-field generated by $\mf X_{1,n} ,  \ldots , \mf X_{n,n}$.
\\
(ii) Under the alternative hypothesis \eqref{eq:Haextend}, if \begin{align}
\sqrt{L  s_n\log(nN^2) }  l_n\eta_n (n^2N^2)^{1/q} + (L\log (n  N^2))^{1/2}\{ L/(n\tau) + Ls_n/n\} = o(s_n \log(nN^2)), \label{eq:bootstrap_ass_extend}
\end{align} we have
\begin{align}
   P(Q_n/\sqrt{s_n} > \tilde r_{\bt}|\mathcal F_n) \to 1
\end{align}
in probability, where $ \tilde r_{\bt}$ is as defined in \cref{alg:boot3}.
\end{theorem}

\HDB{
\begin{remark} \label{rem51}
    {\rm For the case $p=1$  hypotheses of the form  \eqref{eq:H0extend} have been considered  by \cite{Axel2023}  for a finite number of lags $k=1, \ldots , H$ and centered locally  stationary time series. An ad-hoc alternative test can therefore constructed by applying the test  proposed in \cite{Axel2023} to the residuals in \eqref{det10}. However, in a simulation study  we observed that this test  exceeds its nominal  level under the null hypothesis substantially (these results are not depicted for the sake of brevity).
    }
\end{remark}}

\HDB{
In Table \ref{tb:3}, we compare the rejection rates of the our extended test in this section with the results of the original test under Model 2 in Section \ref{sec5} with  the same procedure for selecting tuning parameters. We found that for both tests the rejection rates are close to their nominal levels under the null hypothesis. For the alternative hypothesis,  the concurrent test  (Algorithm \ref{alg:boot2}) and the extended test  (Algorithm \ref{alg:boot3})  have comparable power performance. In particular, for the tests with the significance level $10\%$, the extended test  is slightly more powerful. 
\begin{table}[ht]
\centering
\begin{tabular}{rrrrr| rrrr}
  \hline 
  & \multicolumn{4}{c|}{$H_0$}  & \multicolumn{4}{c}{$H_A$} \\  
  & \multicolumn{2}{c}{Test \ref{alg:boot2}}  & \multicolumn{2}{c| }{Test \ref{alg:boot3}} & \multicolumn{2}{c}{Test \ref{alg:boot2}}  & \multicolumn{2}{c}{Test \ref{alg:boot3}}\\
   \hline
 & 5\% & 10\% & 5\% & 10\% & 5\% & 10\% & 5\% & 10\% \\ 
  \hline
200 & 5.0 & 10.1 & 4.4 & 8.8 & 10.7 & 19.2 & 9.9 & 16.8 \\ 
  400 & 5.7 & 11.9 & 5.4 & 10.0 & 21.0 & 31.9 & 20.6 & 34.4 \\ 
  600 & 4.8 & 9.0 & 4.9 & 10.3 & 45.3 & 61.2 & 47.4 & 64.1 \\ 
  800 & 5.3 & 10.2 & 4.7 & 9.8 & 77.0 & 86.3 & 74.3 & 87.6 \\ 
  1000 & 5.0 & 10.6 & 5.8 & 10.1 & 87.7 & 93.9 & 89.0 & 95.7 \\ 
   \hline
\end{tabular}
\label{tb:3}
\vskip .1cm 
\caption{Comparison of the rejection rates of  the  proposed tests 
for the hypothesis \eqref{eq:H0}
(Algorithm \ref{alg:boot2})  and 
for the hypothesis \eqref{eq:H0extend} (Algorithm \ref{alg:boot3}) in  the locally stationary Model 2 introcuced in Section \ref{sec5} ($p=5$) under the null hypothesis and the alternative hypothesis.}
\end{table}
}
\HDB{
\section{Data analysis} 
We analyze a three-dimensional time series of hourly energy consumption, measured in megawatts (MW), for American Electric Power, The Dayton Power and Light Company, and Northern Illinois Hub. The data is available in \href{https://www.kaggle.com/code/robikscube/time-series-forecasting-with-machine-learning-yt/input?select=pjm_hourly_est.csv}{kaggle}. Each day’s consumption profile is treated as a continuous function over 24 hours.
Given that 2006 ranked as the sixth warmest year globally, we focus on the period from February 16, 2006, to July 4, 2007. This interval comprises 501 days with complete 24-hour observations, and the data for $6$ days is displayed in Figure \ref{fig:placeholder}.
\begin{figure}
    \centering
\includegraphics[width=1\linewidth]{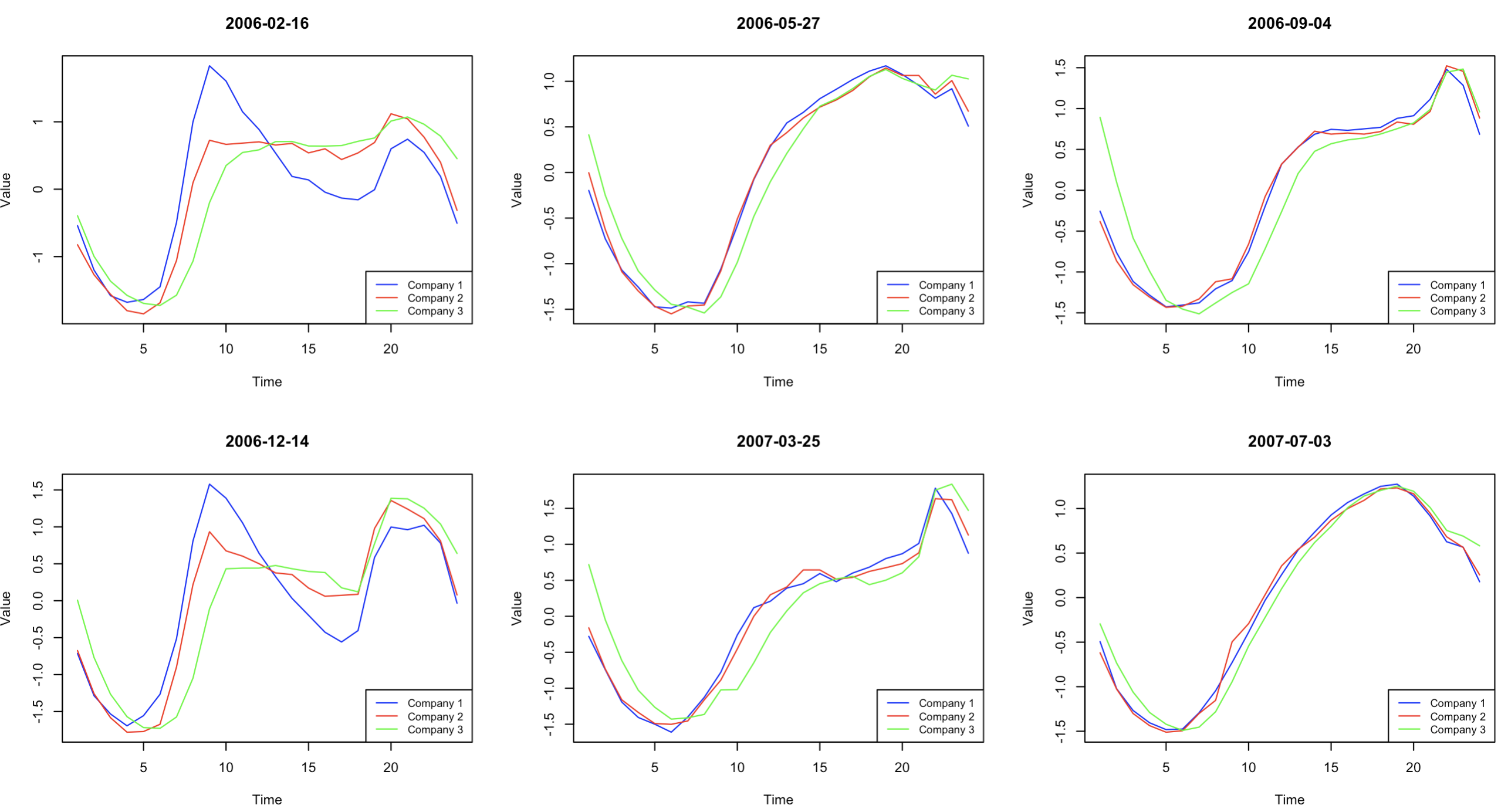}
    \caption{Normalized consumption curves from 6 days from 2006-02-16 to 2007-07-04 from  American Electric Power(blue), The Dayton Power and Light Company(red), Northern Illinois Hub (green). }
    \label{fig:placeholder}
\end{figure}
We conduct our functional portmanteau test using selecting block sizes from $\lfloor n^{1/5} \rfloor, \lfloor  n^{1/5} \rfloor + 1, \ldots, \lfloor 3 n^{1/5} \rfloor$ as in the simulation, yielding p-value 0.066, while the extended functional portmanteau test yields 0.089.  Therefore, we reject  the null hypothesis of white noise of the function time series at the significance level of $0.1$. There can be correlations among the daily energy consumption curves of different companies from 2006-02-16 to 2007-07-04. }
\section*{Appendix}

\begin{appendix}
This appendix contains the detailed proofs of the results as well as the necessary auxiliary results. Some of our main technical contributions are of independent interest, in particular: 
\begin{itemize}
    \item[(1)] 
We provide verifiable dependence conditions to guarantee bounds for  the joint cumulants for locally stationary functional time series.
         \item[(2)]  We provide  a high-dimensional Gaussian approximation and bootstrap consistency  for the estimates of 
    $\mf G_k(t,u) = \bs \Gamma_k(t,u) \bs \Gamma^{\top}_k(t,u) $
    with increasing dependence between 
    the two factors appearing in the representation \eqref{eq:hatgamma}.
    The  limit of the variance of the proposed statistic is derived to ensure non-degeneracy.   
    \item[(3)]  We uniformly control the dependence between the  terms 
$\hat{\mf  U}_i(t, u)$ in the representation 
\eqref{hol123} of the estimate $\sum_{k=1}^{s_n} \hat {\mf G}_k(t,u)$ for an increasing $s_n$.
\end{itemize}

In \cref{sec7}, we introduce some preliminaries and notation.  The proof of Theorem \ref{thm:ga} is given in \cref{sec73}. \cref{sec6} provides the proofs of Theorems \ref{thm:alt}  and \ref{thm:bs}. The proof of \cref{thm:extend} is discussed in  \cref{sec:exproof}.

\section{Preliminaries}
 \def\theequation{A.\arabic{equation}}	
  \setcounter{equation}{0}
\label{sec7} 

In this section we provide and prove a  technical result, which will be essential for the proofs of the main statements in Section \ref{sec3}.

For this purpose we introduce some notation, which used in the remaining part of the paper.
For a vector $\bs v \in \R^p$, let $v^{s}$ denote its $s$-th element. For example, let $\epsilon^i_{j,n}(u)$ denote the $i$-th element of $\bs \epsilon_{j,n}(u)$. Let $(\cdot)_+$ denote $\max\{\cdot, 0\}$, $a \vee b$ denote $\max\{a, b\}$. In the proof,  we use  $C$ as a  sufficiently large constant depending on the context.  Let $\FF_{i}^j = (\eta_i,\cdots, \eta_j)$. 
 Define the projection operators $\proj^j(\cdot) = \E(\cdot|\FF_{-\infty}^j) -\E(\cdot|\FF_{-\infty}^{j-1})$.
 Then,  $(\proj^j(\cdot))_{j \in \mathbb Z}$ and $(\proj_i(\cdot))_{i \in \mathbb Z}$ are martingale differences with respect to the filtrations $(\FF_{-\infty}^j)_{j\in \mathbb Z}$ and $(\FF_{-i}^{\infty})_{i\in \mathbb Z}$. Write $\|X\|_{F,p}:= \||X|_F\|_p$ if $X$ is a matrix. For simplicity, we omit the index $n$  in $M_n$ in the definition of the set  $\mathbb L_i$ in  \eqref{eq:hatgamma}.

\begin{proposition}
If \cref{ass:expo} is satisfied and  for some  $q \geq 2$,  $$\delta_n= s_n^2(n\tau) \{b^2 + 1/(nb)\}\tau^{-2/q} +s_n^{3/2} \sqrt{n/b} \tau^{1-2/q} = o(s_n \sqrt{n\tau}),$$ and   $\tau \to 0$,  $n\tau \to \infty$, $(\log n)^4/(nb) \to 0$, then we have  
\begin{align}
    \left\| \sup_{t \in [\tau, 1-\tau], u \in [0,1]}\left|  \sum_{i=1}^{n}\mathrm{tr} \left\{\hat{\mf U}_i(t,u)- \mf U_i(t,u)\right\}K_{\tau}(i/n - t) \right| \right\|_q &=  O(\delta_n) = o(s_n \sqrt{n\tau}).
\end{align}
\label{prop:esti}
\end{proposition}

In the following, we establish the Gaussian approximation scheme for the summation of the Frobenius norm of of high-dimensional functional covariance operator with respect to a diverging number of lags. 
We recall the definition of the vectors 
${\mf V}_i$
defined in \eqref{det27}.

 \subsection{Proof of Proposition \ref{prop:esti}}

\subsubsection{Auxiliary results} Recall that  $\epsilon^i_{j,n}(u)$ denote the $i$-th element of the vector $\bs \epsilon_{j,n}(u)$, $j=1,\cdots, n$.
The following lemma gives the order of moments and physical dependence measure of $\epsilon^i_{j,n}(u)$, which will be frequently used in the all proofs.
\begin{lemma}\label{lm:delta}
Under \cref{ass:expo}, for $q \geq 2$ and $c_0$ as defined in \cref{ass:expo}, we have $1 \leq i \leq p$, $1 \leq j \leq n$, and $u \in [0,1]$, $\mathbb E(|\epsilon^i_{j,n}(u)|^q) \leq c_0 t_0^{-q} q^q$, and $\|\epsilon^i_{j,n}(u) - \epsilon^{i,(j-l)}_{j,n}(u)\|_{q} \leq 2 c_0^{1/4} t_0^{-1/2} q \chi^{l/(2q)}$.
\end{lemma}
\begin{proof}

For $t_0 > 0$ as defined in \cref{ass:expo}, 
elementary calculation gives that $x > t_0 \log (x/t_0)$ for $x > 0$. 
    By \cref{ass:expo}, we have 
\begin{align}
    \max_{1 \leq i \leq p, 1 \leq j \leq n, u \in [0,1]}\mathbb E |\epsilon^i_{j,n}(u)|^q  <\max_{1 \leq i \leq p, 1 \leq j \leq n, u \in [0,1]} \left(\frac{q}{t_0} \right)^q \mathbb E \exp(t_0|\epsilon^i_{j,n}(u)|) \leq c_0 t_0^{-q} q^q. \label{eq:moment}
\end{align}
 By Hölder inequality, we have
\begin{align}
    (\mathbb E|\epsilon_{j,n}^i(u) - \epsilon_{j,n}^{i,(j-l)}(u)|^q)^2 & =  (\mathbb E|\epsilon_{j,n}^i(u) - \epsilon_{j,n}^{i,(j-l)}(u)|^{q-1/2}|\epsilon_{j,n}^i(u) - \epsilon_{j,n}^{i,(j-l)}(u)|^{1/2})^2\\
    & \leq \mathbb E|\epsilon_{j,n}^i(u) - \epsilon_{j,n}^{i,(j-l)}(u)|^{2q-1} \mathbb  E|\epsilon_{j,n}^i(u) - \epsilon_{j,n}^{i,(j-l)}(u)| \\ 
    & \leq \|\epsilon_{j,n}^i(u) - \epsilon_{j,n}^{i,(j-l)}(u)\|^{2q-1}_{2q-1}\delta_{l,i} \leq c_0 (2q-1)^{2q-1} t_0^{-(2q-1)} \delta_{l,j},\label{eq:delta}
\end{align}
Therefore, we have $\|\epsilon_{j,n}^i(u) - \epsilon_{j,n}^{i,(j-l)}(u)\|_q \leq  c_0^{1/(2q)} (2q-1)^{1-1/(2q)} t_0^{-(1-1/(2q))} \delta_{l,j}^{1/(2q)} \leq C q \chi^{l/(2q)}$, where $C = 2 c_0^{1/4} t_0^{-1/2}$.
\end{proof}

\begin{proposition}\label{prop:b2}
     Let $\{\Upsilon_n(x_1,x_2)\}_{x_1, x_2 \in [0,1]}$ be a sequence of stochastic processes with differentiable paths. Assume that for some $p \geq 1$, and any $x_1, x_2  \in [0,1]$, $\| \Upsilon_n(x_1, x_2)\|_p = O(m_n)$, $ \max_{i=1,2}\|\frac{\partial}{\partial x_i} \Upsilon_n(x_1, x_2)\|_p = O(l_{1,n})$, $ \|\frac{\partial^2}{\partial x_1 \partial x_2} \Upsilon_n(x_1, x_2)\|_p = O(l_{2,n})$, where $m_n$, $l_{i,n}$ are sequences of real numbers, $i=1,2$. If $\tau \to 0$, then, 
    \begin{align}
        \left\| \sup_{x \in [\tau,1-\tau], y\in [0,1]}|\Upsilon_n(x,y)| \right\|_p = O(m_n c_n^{-2/p}),
    \end{align}
    where $c_n = \min\left\{ \left(\frac{m_{n}}{l_{2,n}}\right)^{1/2}, \frac{m_{n}}{l_{1,n}}\right\}$.
\end{proposition}
\begin{proof}
    The proof follows from similar but simpler arguments of the proof of the following \cref{prop:b3}.  $c_n$ is chosen in our case to make sure $c_n^2 l_{2,n} + c_n l_{1,n}  = O(m_n)$.
\end{proof}

\begin{proposition}\label{prop:b3}
    Let $\{\Upsilon_n(x_1,x_2,x_3)\}_{x_1, x_2, x_3 \in [0,1]}$ be a sequence of stochastic processes with differentiable paths. Assume that for some $p \geq 1$, and any $x_1, x_2, x_3 \in [0,1]$, $\| \Upsilon_n(x_1, x_2, x_3)\|_p = O(m_n)$, $\max_{i = 1,2,3} \|\frac{\partial}{\partial x_i} \Upsilon_n(x_1, x_2, x_3)\|_p = O(l_{1,n})$, $\max_{i< j, i,j = 1,2,3} \|\frac{\partial^2}{\partial x_i x_i} \Upsilon_n(x_1, x_2, x_3)\|_p = O(l_{2,n})$, $ \|\frac{\partial^3}{\partial x_1 \partial x_2 \partial x_3} \Upsilon_n(x_1, x_2, x_3)\|_p = O(l_{3,n})$, where $m_n$, $\{l_{i,n}\}_{i=1}^3$ are sequences of real numbers. If $\tau \to 0$, then, 
    \begin{align}
        \left\| \sup_{x \in [\tau,1-\tau], y,z \in [0,1]}|\Upsilon_n(x,y,z)| \right\|_p = O(m_n c_n^{-3/p}),
    \end{align}
    where $c_n = \min \left\{\left(\frac{m_{n}}{l_{3,n}}\right)^{1/3}, \left(\frac{m_{n}}{l_{2,n}}\right)^{1/2}, \frac{m_n}{l_{1,n}}\right\}$.
\end{proposition}

\begin{proof}
For sequences of real numbers $c_n \to 0$, write $\tilde c_n = \lfloor c_n^{-1} \rfloor$. Let $t_i = i \tilde c_n $, $i = 0,1,\ldots, \tilde c_n, t_{\tilde c_n+1} = 1$ and we have
\begin{align}
    \sup_{x, y, z \in [0,1]}|\Upsilon_{n}(x, y, z)| &\leq Y_1 + Y_2, \label{eq:Yn}
\end{align}
where $Y_1 = \max_{0 \leq i,j,k \leq \tilde c_n + 1}|\Upsilon_{n}(t_i, t_j, t_k)|$, $Y_2 = \max_{0 \leq i,j,k \leq \tilde c_n + 1} \tilde Z_{i,j,k}$ where $$\tilde Z_{i,j,k} =  \sup_{t_i - c_n \leq x_1 < t_i, t_j - c_n \leq x_2 < t_j, t_k - c_n \leq x_3 < t_k}|\Upsilon_{n}(t_i, t_j, t_k) - \Upsilon_{n}(x_1,x_2,x_3)|.$$
For $Y_1$, we have 
\begin{align}
    \| Y_1\| = \|\max_{0 \leq i,j,k \leq \tilde c_n + 1}|\Upsilon_{n}(t_i, t_j, t_k)| \|_p  \leq \left\{\sum_{i,j,k=0}^{ \tilde c_n + 1}\E |\Upsilon_{n}(t_i, t_j, t_k)|^p \right\}^{1/p} = O(m_n c_n^{-3/p}). \label{eq:Y1}
\end{align}
For convenience, write $l_1 = i$, $l_2 =j$, $l_3 = k$.
For $\tilde Z_{i,j,k}$, by elementary calculation, we have 
\begin{align}
    \|\tilde Z_{i,j,k} \|_p &=  \sum_{s=1}^3\int_{t_{l_s} - c_n}^{t_{l_s}} \left\|\frac{\partial}{\partial x_s} \Upsilon_n(x_1, x_2, x_3) \right\|_p dx  \\
   & + \sum_{s < q, s,q=1,2,3}\int_{t_{l_s} - c_n}^{t_{l_s}}\int_{t_{l_q} - c_n}^{t_{l_q}} \left\|\frac{\partial^2}{\partial x_s \partial x_q} \Upsilon_n(x_1, x_2, x_3) \right\|_p dx  \\
   &+ \int_{t_i - c_n}^{t_i} \int_{t_j - c_n}^{t_j}\int_{t_k- c_n}^{t_k}\left\|\frac{\partial^3}{\partial x_1 \partial x_2 \partial x_3} \Upsilon_n(x_1, x_2, x_3) \right\|_p dx\\ 
   & = O(c_n l_{1,n} + c_n^2 l_{2,n} + c_n^3 l_{3,n}).
\end{align}
Following the similar arguments for \eqref{eq:Y1}, we have 
\begin{align}
     \| Y_2\|_p = \|\max_{0 \leq i,j,k \leq \tilde c_n + 1} \tilde Z_{i,j,k} \|_p= O( c_n^{1-1/p} l_{1,n} + c_n^{2-2/p} l_{2,n} + c_n^{3-3/p} l_{3,n}). \label{eq:Y2}
\end{align}
Combining \eqref{eq:Y1}, \eqref{eq:Y2} and \eqref{eq:Yn}, by triangle inequality, we have 
\begin{align}
   &\left\| \sup_{x \in [\tau,1-\tau], y,z \in [0,1]}|\Upsilon_n(x,y,z)| \right\|_p \\ & \leq \|  \sup_{x, y, z \in [0,1]}|\Upsilon_{n}(x, y, z)|\|_p \\&= O( m_n c_n^{-3/p} +  c_n^{1-1/p} l_{1,n} + c_n^{2-2/p} l_{2,n} + c_n^{3-3/p} l_{3,n})\\ 
 &= O(m_n c_n^{-3/p}).
\end{align}
\end{proof}

\subsubsection{Proof of Proposition \ref{prop:esti}}
Note that 
\begin{align}
    \hat{\bs \phi}_{i,k}(u) - {\bs \phi}_{i,k}(u) &= \hat{\bs \epsilon}_{i-k}(u)\hat{\bs \epsilon}_{i}^{\top}(u) - {\bs \epsilon}_{i-k}(u){\bs \epsilon}_{i}^{\top}(u)\\ 
    &= {\bs \epsilon}_{i-k}(u) \left\{ \mf m(i/n,u) -\hat {\mf m}(i/n, u) \right\}^{\top}\\& +  \left\{ \mf m((i-k)/n,u) -\hat {\mf m}((i-k)/n, u) \right\} {\bs \epsilon}_{i}^{\top}(u)\\ 
    &+ \{ \mf m((i-k)/n,u) -\hat {\mf m}((i-k)/n, u)\} \{ \mf m(i/n,u) -\hat {\mf m}(i/n, u)\}^{\top}.\label{eq:phi}
\end{align}
Let $I_n = [b, 1-b]$, and $I_n^c = [0,b) \cup (1-b, 1]$. 
By similar arguments in Lemma B.1 and B.2 of \cite{dette2019}, we have 
\begin{align}
    \sup_{i/n \in I_n, u \in [0,1]}\left|\mf m(i/n, u) - \hat{\mf m}(i/n, u) -  \frac{1}{nb}\sum_{j=1}^n \bs \epsilon_i(u)K_b(\frac{i-j}{n})\right|_F = O(b^2 + (nb)^{-1}) \label{eq:m1} 
\end{align}
and 
\begin{align}
    \sup_{i/n \in I^c_n, u \in [0,1]}\left|\mf m(i/n, u) - \hat{\mf m}(i/n, u) - \frac{1}{nb} \sum_{j=1}^n c_j(i/n) \bs \epsilon_i(u)K_b(\frac{i-j}{n})\right|_F = O(b^2 + (nb)^{-1}),\label{eq:m2}   
\end{align}
where $c_j(t) = \left\{\nu_{2,b}(t) - \nu_{1,b}(t)(j/n-t)/b\right\}/(\nu_{0,b}(t)\nu_{2,b}(t) - \nu^2_{1,b}(t))$, and $$\nu_{l,b}(t) = \int_{-t/b}^{(1-t)/b} x^l K(x) dx.$$ Next, for $i/n \in I_n$, we define
\begin{align}
    \mf a_{i,k}(u) = \frac{1}{nb} \bs \epsilon_{i-k}(u) \sum_{j=1}^n \bs \epsilon_j^{\top}(u) K_b(\frac{i-j}{n}) + \frac{1}{nb}  \sum_{j=1}^n \bs 
 \epsilon_j(u) K_b(\frac{i-k-j}{n})\bs \epsilon^{\top}_{i}(u),
\end{align}
 for $i/n \in I_n^c$, let
\begin{align}
    \mf a_{i,k}(u)& = \frac{1}{nb}  \bs \epsilon_{i-k}(u) \sum_{j=1}^n c_j(i/n) \bs \epsilon_j^{\top}(u) K_b(\frac{i-j}{n}) \\ &+  \frac{1}{nb}\sum_{j=1}^n c_j((i-k)/n)\bs 
 \epsilon_j(u) K_b(\frac{i-k-j}{n}) \bs \epsilon^{\top}_{i}(u) , \label{eq:aikIc}
\end{align}
and $\mf a_{i,k}(u) = 0$ if $i < k$.
Since $\proj^l \left\{ \bs \epsilon_j(u)\right\}$ are martingale differences, by Burkholder inequality and  \cref{lm:delta}, for fixed $q$, we have 
\begin{align}
    &\max_{i/n \in I_n} \|\mf a_{i,k}(u) \|_{F,q} \\&\leq  \frac{1}{nb}   \max_{i/n \in I_n} \left\{\| \bs \epsilon_{i-k}(u)\|_{F, 2q} \| \sum_{j=1}^n \bs \epsilon_j^{\top}(u) K_b(\frac{i-j}{n})\|_{F, 2q} \right.\\
  & \left. +  \| \bs \epsilon_{i}(u)\|_{F, 2q} \|  \sum_{j=1}^n \bs 
 \epsilon_j^{\top}(u) K_b(\frac{i-k-j}{n})\|_{F, 2q}\right\}\\ 
 &\leq  \frac{2}{nb}  \max_{1 \leq i \leq n} \| \bs \epsilon_{i}(u)\|_{F, 2q}  \max_{1 \leq i \leq n} \sum_{l=1}^{\infty} \| \sum_{j=1}^n \proj^{i-l}\left\{\bs \epsilon_j(u)\right\} K_b(\frac{i-j}{n})\|_{F, 2q}\\
 &\leq \frac{2}{nb}  \max_{1 \leq i \leq n} \| \bs \epsilon_{i}(u)\|_{F, 2q}  \max_{1 \leq i \leq n} \sum_{l=1}^{\infty} \left\{ \sum_{j=1}^n \|\proj^{i-l}\left\{\bs \epsilon_j(u)\right\} K_b(\frac{i-j}{n})\|^2_{F, 2q}\right\}^{1/2}\\
  &\leq \frac{2}{nb}  \max_{1 \leq i \leq n} \| \bs \epsilon_{i}(u)\|_{F, 2q}  \max_{1 \leq i \leq n} \sum_{l=1}^{\infty} \left\{ \sum_{j=1}^n \|\sum_{r=1}^p( \epsilon^{r}_j(u)- \epsilon^{r, (j-l)}_j(u)) K_b(\frac{i-j}{n})\|^2_{2q}\right\}^{1/2}\\
 & = O(1/\sqrt{nb}).\label{eq:aikin}
\end{align}
 Since $c_j(t)$ are uniformly bounded, following similar arguments as  \eqref{eq:aikin}, we have 
\begin{align}
     \max_{i/n \in I^c_n} \|\mf a_{i,k}(u) \|_{F,q}= O(1/\sqrt{nb}).\label{eq:aikbd}
\end{align}
Combining \eqref{eq:m1},  \eqref{eq:m2}, \eqref{eq:aikin} and \eqref{eq:aikbd}, we have 
\begin{align}
   \max_{1 \leq i \leq n} \|\mf m(i/n, u) - \hat{\mf m}(i/n, u) \|_{F,q} = O(b^2 + 1/\sqrt{nb}). \label{eq:mhatm} 
\end{align}
Together with \cref{lm:delta}, \eqref{eq:phi}, \eqref{eq:m1},  \eqref{eq:m2},  by Cauchy inequality, we have 
\begin{align}
 \max_{1 \leq i \leq n}   \|\hat{\bs \phi}_{i,k}(u) - {\bs \phi}_{i,k}(u) - \mf a_{i,k}(u) \|_{F,q} = O(b^2 + 1/(nb)).\label{eq:phiai}
\end{align}
We shall proceed by quantifying the estimation error using $\mf a_{i,k}(u)$ and showing the terms involving $\mf a_{i,k}(u)$ are negligible. Recall that $\mathbb L_i = \{j: M < |j-i| \leq s_n\}$. Observe that $\{i \in \mathbb L_l\}$ and $\{l \in \mathbb L_i\}$ are symmetric in the following summation. By triangle inequality and elementary calculation, we have
\begin{align}
    &\sup_{t \in [\tau, 1-\tau], u \in [0,1]}\left|\sum_{i=1}^n K_{\tau}(i/n - t) \mathrm{tr}(\hat{\mf U}_i(t, u ) - {\mf U}_i(t, u ))  \right|/2 \\ 
    &\leq 2 \sup_{t \in [\tau, 1-\tau], u \in [0,1]} \left| \mathrm{tr}\left( \sum_{i=1}^n\sum_{k=1}^{s_n} K_{\tau}(i/n - t) \bs \phi_{i,k}(u) \sum_{l = 1}^{n} K_{\tau} (l/n - t)(\hat {\bs \phi}_{l,k}(u) - \bs \phi_{l,k}(u) )^{\top} \mathbf 1(l \in \mathbb L_i)\right) \right| \\ 
    &+ \sup_{t \in [\tau, 1-\tau], u \in [0,1]} \left| \mathrm{tr}\left( \sum_{i=1}^n\sum_{k=1}^{s_n} K_{\tau}(i/n - t) (\hat{\bs \phi}_{i,k}(u) -    \bs \phi_{i,k}(u)) \sum_{l = 1}^{n} K_{\tau} (l/n - t)\right.\right. \\ &\times \left.\left.(\hat {\bs \phi}_{l,k}(u) - \bs \phi_{l,k}(u) )^{\top} \mathbf 1(l \in \mathbb L_i)\right) \right| \\ 
    & \leq 2 \sup_{t \in [\tau, 1-\tau], u \in [0,1]} \left| \mathrm{tr}\left( \sum_{i=1}^n\sum_{k=1}^{s_n} K_{\tau}(i/n - t) \bs \phi_{i,k}(u) \sum_{l=1}^n  \mathbf 1(l \in \mathbb L_i) K_{\tau} (l/n - t)\right.\right. \\ &\times \left.\left.(\hat {\bs \phi}_{l,k}(u) - \bs \phi_{l,k}(u) 
 - \mf a_{i,k}(u))^{\top}\right) \right|\\ 
 &+ 2 \sup_{t \in [\tau, 1-\tau], u \in [0,1]} \left| \mathrm{tr}\left( \sum_{i=1}^n\sum_{k=1}^{s_n} K_{\tau}(i/n - t) \bs \phi_{i,k}(u) \sum_{l=1}^n  \mathbf 1(l \in \mathbb L_i) K_{\tau} (l/n - t)\mf a^{\top}_{l,k}(u)\right) \right| \\ 
  & + \sup_{t \in [\tau, 1-\tau], u \in [0,1]} \left| \mathrm{tr}\left( \sum_{i=1}^n\sum_{k=1}^{s_n} K_{\tau}(i/n - t) (\hat{\bs \phi}_{i,k}(u) -    \bs \phi_{i,k}(u) - \mf a_{i,k}(u)) \right.\right. \\ &\times \left.\left.\sum_{l=1}^n \mathbf 1(l \in \mathbb L_i) K_{\tau} (l/n - t)(\hat {\bs \phi}_{l,k}(u) - \bs \phi_{l,k}(u) -\mf a_{l,k}(u))^{\top}\right) \right|\\ 
  &+ 2\sup_{t \in [\tau, 1-\tau], u \in [0,1]} \left| \mathrm{tr}\left( \sum_{i=1}^n\sum_{k=1}^{s_n} K_{\tau}(i/n - t)  \mf a_{i,k}(u) \sum_{l =   1}^{n} K_{\tau} (l/n - t)\right.\right. \\ &\times \left.\left.(\hat {\bs \phi}_{l,k}(u) - \bs \phi_{l,k}(u) -\mf a_{l,k}(u))^{\top} \mathbf 1(l \in \mathbb L_i)\right) \right|\\
  &+ \sup_{t \in [\tau, 1-\tau], u \in [0,1]} \left| \mathrm{tr}\left( \sum_{i=1}^n\sum_{k=1}^{s_n} K_{\tau}(i/n - t)  \mf a_{i,k}(u)\sum_{l=1}^n \mathbf 1(l \in \mathbb L_i) K_{\tau} (l/n - t) \mf a^{\top}_{l,k}(u)\right) \right|\\
  &:=A_1 + B_1 + A_2 + A_3 + B_2, \label{eq:defAB}
\end{align}
where $A_1, A_2, A_3, B_1, B_2$ are defined in the obvious way.\par
\textbf{Approximation the estimation error by $\hat {\bs \phi}_{i,k}(u) - \bs \phi_{l,k}(u) 
 - \mf a_{i,k}(u)$, namely $A_1$, $A_2$, $A_3$.}
The upper bound for $A_1, A_2$ and $A_3$ can be derived by triangle inequality. For example, by \eqref{eq:phiai} and \cref{lm:delta}, we have
\begin{align}
     &\sup_{t \in [\tau, 1-\tau], u \in [0,1]}  \left\| \mathrm{tr}\left( \sum_{i=1}^n\sum_{k=1}^{s_n} K_{\tau}(i/n - t) \bs \phi_{i,k}(u) \sum_{l=1}^n  \mathbf 1(l \in \mathbb L_i) K_{\tau} (l/n - t)\right.\right. \\ &\times \left.\left.(\hat {\bs \phi}_{i,k}(u) - \bs \phi_{l,k}(u) 
 - \mf a_{i,k}(u))^{\top}\right) \right\|_{F,q}\\ 
 &  \leq  \sup_{t \in [\tau, 1-\tau], u \in [0,1]} \sum_{i=1}^n\sum_{k=1}^{s_n} |K_{\tau}(i/n - t)| \|\bs \phi_{i,k}(u)\|_{F,2q} \\ &\times  \sum_{l=1}^n  \mathbf 1(l \in \mathbb L_i) |K_{\tau} (l/n - t)|\| \hat {\bs \phi}_{i,k}(u) - \bs \phi_{l,k}(u) 
 - \mf a_{i,k}(u)\|_{F,2q}\\
 &  = O[s_n^2 (n\tau) \left\{ b^2 + 1/(nb)\right\}].
\end{align}
By \cref{prop:b2}, we have 
\begin{align}
    \| A_1 \|_q =  O[s_n^2 (n\tau) \left\{ b^2 + 1/(nb)\right\} \tau^{-2/q}].\label{eq:bdA1}
\end{align}
Similarly, by \eqref{eq:aikin} and \eqref{eq:aikbd}, we can obtain 
\begin{align}
    \| A_2\|_q =  O[s_n^2 (n\tau) \left\{ b^4 + 1/(nb)^2\right\} \tau^{-2/q}], \quad  \| A_3\|_q =  O[s_n^2 (n\tau)(nb)^{-1/2} \left\{ b^2 + 1/(nb)\right\} \tau^{-2/q}].\label{eq:bdA2}
\end{align}
\textbf{Terms involving $\mf a_{i,k}(u)$, namely $B_1$ and $B_2$.}
In order to bound $B_1$ and $B_2$, we will first approximate the terms by $m$-dependent processes and then use blocking arguments. We show the calculation for $B_2$ and the calculation for $B_1$ follows similarly but with simpler arguments. \par Let $\check{\bs \epsilon}_{i}(u) :=  \check{\bs G}(i/n, u, \F_i) = \E(\bs \epsilon_{i}(u)| \F_{i-m}^i)$, where $\F_{i-m}^i= (\eta_{i-m}, \cdots, \eta_i)$.
Define $\mf a_{i,k}^{(m)}(u)$ as the counterpart of $\mf a_{i,k}(u)$ where $\bs \epsilon_i(u)$ are replaced by $\check{\bs \epsilon}_i(u)$. 
Notice that 
\begin{align}
    \bs \epsilon_j(u) - \check{\bs \epsilon}_j(u) = \sum_{r = 0}^{\infty} \proj^{j-r} \{\bs \epsilon_j(u) - \check{\bs \epsilon}_j(u)\},
\end{align}
where by Jensen's inequality 
\begin{align}\label{checkepsilon}
   &\|\proj^{j-r} \{\bs \epsilon_j(u) - \check{\bs \epsilon}_j(u)\}\|_{F,2q}\\&\leq \min\{\proj^{j-r} \{\|\{\bs \epsilon_j(u) - \check{\bs \epsilon}_j(u)\}\|_{F,2q}\},  \\&\|\proj^{j-r}\{\bs \epsilon_j(u)\}\|_{F,2q} + \|\proj^{j-r}\{\bs{\check \epsilon}_j(u)\}\|_{F,2q}\}\\&\leq \min\{2\|\{\bs \epsilon_j(u) - \check{\bs \epsilon}_j(u)\}\|_{F,2q}, \\& \|\proj^{j-r}\{\bs \epsilon_j(u)\}\|_{F,2q} +  \E (\|\proj^{j-r}\{\bs{ \epsilon}_j(u)\}\|_{F,2q}|\F_{i-m}^i)\}
   \\&\leq 2\min\{\|\{\bs \epsilon_j(u) - \check{\bs \epsilon}_j(u)\}\|_{F,2q},  \|\proj^{j-r}\{\bs \epsilon_j(u)\}\|_{F,2q} \}\\
   &= O(\min\{\chi^m, \chi^r\}). 
\end{align}
Since $\proj^{j-r} \{\bs \epsilon_j(u) - \check{\bs \epsilon}_j(u)\}$ and $\proj^{j-r} \{\bs \epsilon_j(u)\}$ are martingale differences with respect to $j$, by Cauchy inequality and Burkholder's inequality, \eqref{checkepsilon} leads to
\begin{align}
    &\max_{l/n \in I_n} \|\mf a_{l,k}(u) - \mf a_{l,k}^{(m)}(u)\|_{F,q}\\ &\leq  \frac{2}{nb} \max_{1 \leq i \leq n} \|  \bs \epsilon_{i}(u) -\check  {\bs \epsilon}_{i}(u) \|_{F,2q}  \max_{1 \leq i \leq n}  \| \sum_{j=1}^n \bs \epsilon_j^{\top}(u) K_b(\frac{i-j}{n})\|_{F,q} \\ &  + \frac{2}{nb} \max_{1 \leq i \leq n} \|  \check  {\bs \epsilon}_{i}(u) \|_{F,2q}  \max_{1 \leq i \leq n}  \| \sum_{j=1}^n (\bs \epsilon_j(u) - \check{\bs \epsilon}_j(u)) K_b(\frac{i-j}{n})\|_{F,2q} \\
    &\leq  \frac{2}{nb} \max_{1 \leq i \leq n} \|  \bs \epsilon_{i}(u) -\check  {\bs \epsilon}_{i}(u) \|_{F,2q} \sum_{r = 0}^{\infty} \max_{1 \leq i \leq n} \{ \sum_{j=1}^n  \|\proj^{j-r}\{\bs \epsilon_j^{\top}(u)\} K_b(\frac{i-j}{n})\|^2_{F,2q}\}^{1/2} \\ &  + \frac{2}{nb} \max_{1 \leq i \leq n} \|  \check  {\bs \epsilon}_{i}(u) \|_{F,2q}  \sum_{r = 0}^{\infty}\max_{1 \leq i \leq n}  \{ \sum_{j=1}^n \| \proj^{j-r}\{\bs \epsilon_j(u) - \check{\bs \epsilon}_j(u)\} K_b(\frac{i-j}{n})\|^2_{F,2q} \}^{1/2}\\
    & = O(m\chi^m/\sqrt{nb}).\label{aikmin}
\end{align}

Since $c_j(t)$ are uniformly bounded, following similar arguments as  \eqref{aikmin}, we have 
\begin{align}
    \max_{l/n \in I^c_n} \|\mf a_{l,k}(u) - \mf a_{l,k}^{(m)}(u)\|_{F,q}= O(m\chi^m/\sqrt{nb}).\label{aikmbd}
\end{align}
By \eqref{eq:aikin}, \eqref{eq:aikbd}, \eqref{aikmin}, \eqref{aikmbd} and Cauchy inequality, we have
\begin{align}
    &\left\| \mathrm{tr}\left( \sum_{i=1}^n\sum_{k=1}^{s_n} K_{\tau}(i/n - t)  \mf a_{i,k}(u)\sum_{l=1}^n \mathbf 1(l \in \mathbb L_i) K_{\tau} (l/n - t) \mf a^{\top}_{l,k}(u)\right) \right. \\& - \left. \mathrm{tr}\left( \sum_{i=1}^n\sum_{k=1}^{s_n} K_{\tau}(i/n - t)  \mf a_{i,k}^{(m)}(u)\sum_{l=1}^n \mathbf 1(l \in \mathbb L_i) K_{\tau} (l/n - t) \mf a^{(m),\top}_{l,k}(u)\right) \right\|_{F, q} \\
    & \leq  \sum_{i=1}^n\sum_{k=1}^{s_n} |K_{\tau}(i/n - t)|  \| \mf a_{i,k}(u) - \mf a^{(m)}_{i,k}(u) \|_{F, 2q}\sum_{l=1}^n \mathbf 1(l \in \mathbb L_i) |K_{\tau} (l/n - t)| \|  \mf a^{\top}_{l,k}(u)\|_{F, 2q}  \\ 
    &+ \sum_{i=1}^n\sum_{k=1}^{s_n} |K_{\tau}(i/n - t)|  \|\mf a^{(m)}_{i,k}(u) \|_{F, 2q}\sum_{l=1}^n \mathbf 1(l \in \mathbb L_i) |K_{\tau} (l/n - t)| \|  \mf a_{l,k}(u) - \mf a_{l,k}^{(m)}(u)\|_{F, 2q} \\
    & = O(s_n^2 \tau m\chi^m /b). \label{eq:aikaikmaprox}
\end{align}
By triangle inequality and Cauchy inequality, we have 
\begin{align}
   & \left\| \mathrm{tr}\left( \sum_{i=1}^n\sum_{k=1}^{s_n} K_{\tau}(i/n - t) \mf a_{i,k}^{(m)}(u) \sum_{l=1}^n \mathbf 1(l \in \mathbb L_i) K_{\tau} (l/n - t)  \mf a^{(m),\top}_{l,k}(u) \right) \right\|_q \\ 
    & \leq C n \tau s_n \max_{1 \leq i \leq n, 1 \leq k \leq s_n} \|\mf a_{i,k}^{(m)}(u) \|_{F,2q} \max_{1 \leq i \leq n, 1 \leq k \leq s_n} \left \|\sum_{l=1}^n \mathbf 1(l \in \mathbb L_i) K_{\tau} (l/n - t)  \mf a^{(m)}_{l,k}(u)\right\|_{F,2q}, \label{eq:aikaik}
\end{align}
where $C$ is a sufficiently large constant. At the end of the proof, we shall show that for any fixed $q \geq 2$
\begin{align}
    \max_{1 \leq i \leq n, 1 \leq k \leq s_n} \left \|\sum_{l=1}^n \mathbf 1(l \in \mathbb L_i) K_{\tau} (l/n - t)  \mf a^{(m)}_{l,k}(u)\right\|_{F,2q} = O(\sqrt{s_n/(nb)}) .\label{eq:trD}
\end{align} 
Combining \eqref{eq:aikaikmaprox},  \eqref{eq:aikaik} and \eqref{eq:trD}, and similar arguments in calculating \eqref{eq:aikin} and \eqref{eq:aikbd}, we have 
\begin{align}
   \left\| \mathrm{tr}\left( \sum_{i=1}^n\sum_{k=1}^{s_n} K_{\tau}(i/n - t)  \mf a_{i,k}(u)\sum_{l=1}^n \mathbf 1(l \in \mathbb L_i) K_{\tau} (l/n - t) \mf a^{\top}_{l,k}(u)\right)\right\|_q = O(s_n^{3/2} \tau/ b + s_n^2 \tau m\chi^m /b).
\end{align}
 Taking $m = \lfloor C_0 \log n \rfloor$ for some sufficiently large $C_0$, since $(\log n)\sqrt{s_n}/n = o(1)$, by \cref{prop:b2}, we have 
\begin{align}
    \| B_2 \|_q = O(s_n^{3/2} \tau^{1-2/q} /b) = o(s_n \sqrt{n\tau}).\label{eq:bdB2}
\end{align}
And by similar arguments as calculating $B_2$, we have 
\begin{align}
     \| B_1 \|_q = O(s_n^{3/2} \sqrt{n/b} \tau^{1-2/q})= o(s_n \sqrt{n\tau}).\label{eq:bdB1}
\end{align}
Combining \eqref{eq:defAB}, \eqref{eq:bdA1}, \eqref{eq:bdA2}, \eqref{eq:bdB2} and \eqref{eq:bdB1}, 
the Proposition is proved. \par 
\textbf{Proof of \eqref{eq:trD}.}
In the following, we show that 
\begin{align}
    \max_{1 \leq i \leq n, 1 \leq k \leq s_n} \left \|\sum_{l=1}^n \mathbf 1(l \in \mathbb L_i) K_{\tau} (l/n - t)  \mf a^{(m)}_{l,k}(u)\right\|_{F,2q} = O(\sqrt{s_n/(nb)}). \label{eq:sumalk}
\end{align}
Observe that uniformly for $i/n \in I_n$, similar to Lemma 5 of \cite{zhou2010simultaneous}, we have
\begin{align}
   &\left| \sum_{l=1}^n \mathbf 1(l \in \mathbb L_i) K_{\tau} (l/n - t) \E\left\{ \mf a^{(m)}_{l,k}(u)\right\} \right|_F \\ &\leq \left|  
 \frac{1}{nb}  \sum_{j=1}^n\sum_{l=1}^n \mathbf 1(l \in \mathbb L_i) \E \left\{ \check{\bs \epsilon}_{l-k}(u)\check{\bs \epsilon}_j^{\top}(u) \right\} K_b(\frac{l-j}{n})K_{\tau} (l/n - t) \right|_F  \\ &+ \left| \frac{1}{nb}  \sum_{j=1}^n\sum_{l=1}^n  \mathbf 1(l \in \mathbb L_i)  \E \left\{\check{\bs \epsilon}_{l}(u)\check{\bs \epsilon}_j^{\top}(u) \right\}K_b(\frac{l-k-j}{n}) K_{\tau} (l/n - t)\right|_F\\
 & = O\left(\frac{1}{nb}  \sum_{j=1}^n\sum_{l=1}^n  \mathbf 1(l \in \mathbb L_i) \chi^{|l-k-j|}\right) = O(s_n/(nb)), \label{eq:Ealkin}
\end{align}
where the last equality follows from the fact that there are $O(s_n)$ terms in $\mathbb L_i$. 
Define $\check{\bs \epsilon}_{l, (h)}(u)$ as the random variable from substituting $\eta_h$ in $\check{\bs \epsilon}_{l, (h)}(u)$ by its $i.i.d.$ copy $\eta_h^*$. If $h < l - m$, $\check{\bs \epsilon}_{l, (h)}(u) = \check{\bs \epsilon}_{l}(u)$.
Observe that uniformly for $1 \leq l \leq n$, 
\begin{align}
  &\left\|\sum_{j=1}^n \proj^{l-h} \left\{ \check{\bs \epsilon}_{l}(u)\check{\bs \epsilon}_j^{\top}(u)\right\}K_b(\frac{l-j}{n}) \right\|_{F,2q} \\ &\leq  \left\|\sum_{j=1}^n   \check{\bs \epsilon}_j(u) K_b(\frac{l-j}{n})\right\|_{F,4q}  \| \check{\bs \epsilon}_l(u) -\check{\bs \epsilon}_{l, (l-h)}(u) \|_{F,4q} \\
    &+  \left\|\sum_{j=1}^n  \left\{\check{\bs \epsilon}_{j}(u) - \check{\bs \epsilon}_{j, (l-h)}(u)\right\}K_b(\frac{l-j}{n})\right\|_{F,4q}  \| \check{\bs \epsilon}_{l, (l-h)}(u) \|_{F,4q} \\
    & = O( \sqrt{nb} \chi^{h}\mf 1(h \geq 0) + m).\label{eq:pee}
\end{align}
As a consequence, by \eqref{eq:pee}, triangle inequality and Burkholder inequality, uniformly for $i/n \in I_n$, we have
\begin{align}\allowdisplaybreaks
    & \left\| \sum_{l=1}^n \mathbf 1(l \in \mathbb L_i) K_{\tau} (l/n - t)\left[\mf a^{(m)}_{l,k}(u) - \E\left\{ \mf a^{(m)}_{l,k}(u) \right\}\right] \right\|_{2q}\\ 
    & \leq \sum_{h=-2m}^{2m}\left\|  
 \frac{1}{nb}  \sum_{j=1}^n\sum_{l=1}^n \mathbf 1(l \in \mathbb L_i) \proj^{l-k-h} \left\{ \check{\bs \epsilon}_{l-k}(u)\check{\bs \epsilon}_j^{\top}(u) \right\} K_b(\frac{l-j}{n}) K_{\tau} (l/n - t)\right\|_{F,2q} \\
 &+\sum_{h=-2m}^{2m}\left\|  
 \frac{1}{nb}  \sum_{j=1}^n\sum_{l=1}^n \mathbf 1(l \in \mathbb L_i) \proj^{l-h} \left\{ \check{\bs \epsilon}_{l}(u)\check{\bs \epsilon}_j^{\top}(u) \right\} K_b(\frac{l-k-j}{n}) K_{\tau} (l/n - t)\right\|_{F,2q}\\ 
 & \leq \sum_{h=-2m}^{2m} 
 \frac{1}{nb} \left\{\sum_{l=1}^n \mathbf 1(l \in \mathbb L_i)  \left\|\sum_{j=1}^n  \proj^{l-h-k} \left\{ \check{\bs \epsilon}_{l-k}(u)\check{\bs \epsilon}_j^{\top}(u) \right\} K_b(\frac{l-j}{n}) K_{\tau} (l/n - t)\right\|_{F,2q}^2\right\}^{1/2} \\
 &+\sum_{h=-2m}^{2m}
 \frac{1}{nb}  \left\{\sum_{l=1}^n \mathbf 1(l \in \mathbb L_i)  \left\| \sum_{j=1}^n  \proj^{l-h} \left\{ \check{\bs \epsilon}_{l}(u)\check{\bs \epsilon}_j^{\top}(u) \right\} K_b(\frac{l-k-j}{n}) K_{\tau} (l/n - t)\right\|_{F,2q}^2 \right\}^{1/2} \\ 
  & = O(\sqrt{s_n /(nb)} + m^2 \sqrt{s_n} / (nb)), \label{eq:Varalkin}
 \end{align} 
 where the last equality follows from the fact that there are $O(s_n)$ terms in $\mathbb L_i$. 
Combining  \eqref{eq:Ealkin} and \eqref{eq:Varalkin}, taking $m = \lfloor C_0 \log n \rfloor$, for a sufficiently large constant $C_0$,  we have 
\begin{align}
    &\max_{l/n \in I_n}  \left\| \sum_{l=1}^n \mathbf 1(l \in \mathbb L_i) K_{\tau} (l/n - t) \mf a^{(m)}_{l,k}(u)  \right\|_{F,2q} \\ &= O( s_n/(nb) + (\log n)^2 \sqrt{s_n}/(nb)+\sqrt{s_n/(nb)}) = O( \sqrt{s_n/(nb)}) .\label{eq:sumalkin}
\end{align}

Similarly using the fact that $c_{j}(t)$ in \eqref{eq:aikIc} are uniformly bounded, we also have 
\begin{align}
      \max_{l/n \in I_n^c}  \left\| \sum_{l=1}^n \mathbf 1(l \in \mathbb L_i) K_{\tau} (l/n - t) \mf a^{(m)}_{l,k}(u)  \right\|_{2q} = O(\sqrt{s_n/(nb)}).\label{eq:sumalkbd}
\end{align}
Combining \eqref{eq:sumalkin} and \eqref{eq:sumalkbd}, \eqref{eq:sumalk} holds and the proof is completed.

\section{Proof of Theorem \ref{thm:ga}}\label{sec73}
\def\theequation{B.\arabic{equation}}	
  \setcounter{equation}{0}
\subsection{Joint cumulants of locally stationary functional time series}
The definition of $k$th-order joint cumulant of random variables $X_1, \ldots, X_k$ is 
\begin{align}
 \operatorname{Cum}\left(X_1, \ldots, X_k\right)=\sum_{\bs \nu}(-1)^{p-1}(p-1) ! \prod_{j=1}^p\left(\mathbb{E} \prod_{i \in \nu_j} X_i\right),\label{def:cum}
\end{align}
where $\bs \nu= \{\nu_1,\ldots, \nu_p \}$ is a non-overlap partition of $\left\{1,\ldots, k\right\}$ and the summation is over all possible partitions, $p = \#\bs \nu$. For the moments, we have 
\begin{align}
    \E(X_1, \cdots, X_k) = \sum_{\bs \nu} D_{\nu_1}\ldots D_{\nu_p},
\end{align}
where $D_{\nu_s} = \operatorname{Cum}\left(X_{\alpha_1}, \ldots, X_{\alpha_m}\right)$, $\alpha_i$ is the $i$th element in $\nu_s$, and the summation is over all the possible partitions.
\begin{proof}[Proof of \cref{lm:cumk}]
    For the sake of brevity, in the proof we use $X_i$ short for $H_{r_i}(t_i, u, \FF_{m_{i}})$. Define $X_i(j) = \E (H_{r_i}(t_i, u, \FF_{m_{i}})|\FF_{m_j+1}^{\infty})$.  For $1 \leq l \leq K-1$, let $Y_0 = (X_0, \cdots, X_{l-1})$, then we have 
    \begin{align}
        &\mathrm{Cum}(H_{r_0}(t_0, u , \FF_{m_0}) , \ldots, H_{r_{K-1}}(t_{K-1}, u , \FF_{m_{K-1}}))\\ 
        & =  \mathrm{Cum}(Y_0, X_l-X_l(l-1), X_{l+1}, \ldots, X_{K-1}) + \mathrm{Cum}(Y_0, X_l(l-1), X_{l+1}, \ldots, X_{K-1}) \\ 
        & = \mathrm{Cum}(Y_0, X_l-X_l(l-1), X_{l+1}, \ldots, X_{K-1})\\ &+ \sum_{j=1}^{K-l-1}\mathrm{Cum}(Y_0, X_l(l-1),\ldots, X_{l+j-1}(l-1), X_{l+j}-X_{l+j}(l-1), X_{l+j+1}, \ldots, X_{K-1}) \\& + \mathrm{Cum}(Y_0, X_l(l-1), \ldots, X_{K-1}(l-1))\\ 
        &:= A_0(l) + \sum_{j=1}^{K-l-1} A_j(l) + B(l). \label{eq:decompAB}
    \end{align}
    By the independence of sigma field, $Y_0$ is independent of $(X_l(l-1), \ldots, X_{K-1}(l-1))$. Therefore, we have $B = 0$. 
    
    We proceed to calculate $A_0$. According to the definition in \eqref{def:cum}, it suffices to investigate for any $V \subset \{0, \ldots, K-1\}, l \notin V$,
    $$
    \E\left\{(X_l - X_l(l-1))\prod_{j \in V} X_j\right\},
    $$
    since the calculation of $\E\left\{\prod_{j \in V} X_j\right\}$ follows similarly.
 By Hölder inequality, we have 
    \begin{align}
        &\E\left\{(X_l - X_l(l-1))\prod_{j \in V} X_j\right\} \\ &\leq \|X_l - X_l(l-1)\|_{1+|V|} \left( \E \prod_{j \in V} |X_j|^{(|V|+1)/|V|} \right)^{\frac{|V|}{1+|V|}}\\ 
        & \leq \|X_l - X_l(l-1)\|_{K} \left(\sup_{1 \leq j \leq p, t,u \in[0,1]} \E |H_j(t,u, \FF_0) |^{(|V|+1)} \right)^{\frac{|V|}{1+|V|}}\\
        &  \leq \|X_l - X_l(l-1)\|_{K} \sum_{i=0}^{K-1}\left(\sup_{1 \leq j \leq p, t,u \in[0,1]} \E |H_j(t,u, \FF_0) |^{i+1} \right)^{\frac{i}{1+i}}.\label{eq:momentA0}
    \end{align}
    By \cref{lm:delta} and \eqref{eq:momentA0}, considering all possible $V$'s, we have
    \begin{align}
        \sup_{u \in [0,1]}\sup_{0 \leq t_0,\ldots, t_{k-1} \leq 1} |A_0(l)| \leq 2c_0^{1/4}t_0^{-1/2} K \chi^{(m_l-m_{l-1}) / (2K)} \left[\sum_{i=1}^{K-1} \left\{c_0 t_0^{-(i+1)}(i+1)^{i+1}\right\}^{i/(i+1)} \right]. 
    \end{align}
    Write $C =2(K-l)c_0^{1/4}t_0^{-1/2}K \left[\sum_{i=1}^{K-1} \left\{c_0 t_0^{-(i+1)}(i+1)^{i+1}\right\}^{i/(i+1)} \right]$. Note that $C$ is independent of $m_0, \ldots, m_K$.
    
    Similarly, we have $\sup_{u \in [0,1]}\sup_{0 \leq t_0,\ldots, t_{K-1} \leq 1}|A_j(l)| \leq C \chi^{(m_l-m_{l-1}) / (2K)}/(K-l)$, $1 \leq j \leq K-l-1$. 
    
    Note that \eqref{eq:decompAB} holds for all $1\leq l \leq K-1$, and $K$ is fixed. Therefore, we have 
    \begin{align}
       & \sup_{u \in [0,1]}\sup_{0 \leq t_0,\ldots, t_{K-1} \leq 1}| \mathrm{Cum}(H_{r_0}(t_0, u , \FF_{m_0}) , \ldots, H_{r_{K-1}}(t_{K-1}, u , \FF_{m_{K-1}}))|\\ &\leq C \min_{1 \leq l \leq K-1} \chi^{(m_l-m_{l-1})/(2K)}\\ 
       &= C  \chi^{\max_{1 \leq l \leq K-1}\{(m_l-m_{l-1})\}/(2K)}\\
       & \leq  C  \chi^{(m_{K-1}-m_{0})/\{2K(K-1)\}}.
    \end{align}
    For the second result, note that 
    \begin{align}
         &\sum_{s_1,\cdots, s_{K-1} \in \mathbb Z}  \sup_{u \in [0,1]}\sup_{0 \leq t_0,\ldots, t_{K-1} \leq 1} |\mathrm{Cum}(H_{r_0}(t_0, u , \FF_{m_0}) ,H_{r_1}(t_1, u , \FF_{m_0+s_1}), \\ &\ldots, H_{r_{K-1}}(t_{K-1}, u , \FF_{m_0+s_K}))| \\ 
         &\leq 2 \sum_{s=0} ^{\infty} \sum_{\substack{\max\left\{|s_1|,\cdots, |s_{K-1}| \right\} = s\\ s_1,\cdots, s_{K-1} \in \mathbb Z}}\sup_{u \in [0,1]}\sup_{0 \leq t_0,\ldots, t_{K-1} \leq 1} |\mathrm{Cum}(H_{r_0}(t_0, u , \FF_{m_0}), \\&\ldots, H_{r_{K-1}}(t_{K-1}, u , \FF_{m_0+s_K}))| \\
         & = O\left( \sum_{s=0} ^{\infty} s^{K-1}\chi^{s/(2K(K-1))}\right) = O(1).
    \end{align}
\end{proof}
The following lemma studies the summability of joint cumulants of second-order product of locally stationary functional time series
\begin{lemma}\label{lm:cumsum}
Under the conditions of \cref{lm:cumk}
, we have for $k,l > 0$,$i \in \mathbb Z$, $t_i$ in $[\tau, 1-\tau]$, $1 \leq r_i \leq p$, $i = 1,\ldots, 8$, $1\leq i_1, i_2 , j_1,j_2 \leq n$:
\begin{align}\sum_{i_2, j_1, j_2\in \mathbb Z} &\left|\mathrm{Cum}(H_{r_{1}}(t_1, u, \FF_{i_1} )H_{r_{2}}(t_2, u, \FF_{i_{1}-k} ) ,H_{r_{3}}(t_3, u, \FF_{i_2})  H_{r_{4}}(t_4, u, \FF_{i_2-k}), \right.  \\&\left. H_{r_{5}} (t_5, u, \FF_{j_1}) H_{r_{6}}(t_6, u, \FF_{j_1-l}),H_{r_{7}}(t_7, u, \FF_{j_2}) H_{r_{8}}(t_8, u, \FF_{j_2-l}) )\right| =O(1).
\end{align} 
\end{lemma}

\begin{proof}
Write $X_j = H_{r_{j}}(j/n, u, \FF_{i_j} )$ for simplicity. 
By Theorem II.2 of \cite{rosenblatt2012stationary} , we have 
\begin{align}
    &\left|\mathrm{Cum}(H_{r_{1}}(t_1, u, \FF_{i_1} )H_{r_{2}}(t_2, u, \FF_{i_2} ) ,H_{r_{3}}(t_3, u, \FF_{i_3})  H_{r_{4}}(t_4, u, \FF_{i_4}), \right.  \\&\left. H_{r_{5}} (t_5, u, \FF_{i_5}) H_{r_{6}}(t_6, u, \FF_{i_6}),H_{r_{7}}(t_7, u, \FF_{i_7}) H_{r_{8}}(t_8, u, \FF_{i_8}) )\right| = \sum_{\nu} \prod_{q=1}^w |\mathrm{Cum}(X_v, i \in \nu_q) |, \label{eq:prodcum}
\end{align}
where $\{\nu_1, \ldots, \nu_w\}$ is the indecomposable partition of the two way table
\begin{align}
    \begin{pmatrix}
    1& 2 \\ 
    3& 4 \\
    5& 6\\
    7& 8\\
\end{pmatrix},
\end{align}
i.e., there are no rows $R_1, \ldots, R_w$ such that $\nu_1 \cup \cdots \cup \nu_w = R_1\cup \cdots \cup R_w$. 

    The proof follows from  \eqref{eq:prodcum}, Proposition S.3 of \cite{xiao2014portmanteau} and the summability of joint cumulants up to 8-th order implied by \cref{cor:cumk}.
\end{proof}

\subsection{Some auxiliary results}
Let
\begin{align}
\mf W(t,u, i/n , \FF_{i}) = s_n^{-1} \sum_{k=1}^{s_n} (\bs \phi_{i,k}(u)\bs \psi^{\top}_{i,k}(t,u) &+ \bs \psi_{i,k}(t, u)\bs \phi^{\top}_{i,k}(u)),\label{eq:defineW}
\end{align}
and $W_{r,j} (t,u, i/n, \FF_{i}) $ as the element of $\mf W(t,u, i/n , \FF_{i})$ on $r$th row and $j$th column. Define the $s$-order dependence measure of $\mf W(t,u,i/n, \FF_i)$ as $$\theta_{l,s} = \sup_{t , u\in[0,1], 1 \leq i \leq n}\left\|\left|  \sum_{r=1}^p \left\{ W_{r,r} (t,u, i/n, \FF_{i}) -   W_{r,r}(t,u, i/n, \FF_{i}^{(i-l)}) \right\}\right| \right\|_s,$$
where 
$\FF^{(j)}_i$ denotes changing $\eta_j$ in $\FF_i$ by its $i.i.d.$ copy $\eta_j^{\prime}$. 
The following lemma investigates the order of tail sums of $\theta_{l,s}$ and $l\theta_{l,s}$,i.e.,
\begin{align}
    \Theta_{m,s} := \sum_{l=m}^{\infty} \theta_{l,s},\quad \Xi_{m,s}:= \sum_{l=m}^{\infty} l\theta_{l,s}. \label{eq:dependence}
\end{align}

\begin{lemma}\label{lm:delta1}\label{lm:3rdmoment}
    Under the condition of \cref{thm:ga}, for $s = 2,3$, \par
    (i) When $m >  4s_n$,
    there exists a constant $C > 0$,  $\chi_1 \in (0,1)$, such that for $l \geq 0$, 
    \begin{align}
 \Theta_{m,s} = \sum_{l=m}^{\infty} \theta_{l,s} < C   \chi_1^{m}, \quad  \Xi_{m,s} = \sum_{l=m}^{\infty} l\theta_{l,s} < C m   \chi_1^{m}.
    \end{align}
        (ii) When $m\leq  4s_n$, we have 
    \begin{align}
       \Theta_{m,s} =O\left(s_n\right), \quad \Xi_{m,s} = O\left(s_n^2\right).
    \end{align}
\end{lemma}
\begin{proof}
For convenience, we use $(x)_+$ short for $x \vee 0$.
    Let $\bs \phi_{i,k}^{(l)}(u)$ denote $\bs \epsilon^{(l)}_{i-k}(u) (\bs \epsilon^{(l)}_{i}(u))^{\top}$, where $\bs \epsilon^{(l)}_{i-k}(u) = \mf H(t_{i-k}, u, \FF^{(l)}_{i-k})$, $\bs \epsilon^{(l)}_{i}(u) = \mf H(t_{i}, u, \FF^{(l)}_{i})$, and $$\bs \psi_{i,k}^{(l)}(t,u) = \sum_{j=(i-s_n)\vee 1}^{i -M -1}\bs \phi_{j,k}^{(l)}(u)K_{\tau}(j/n - t).$$ Define $\phi^{r,h}_{i,k}(u)$, $\psi^{r,h}_{i,k}(t,u)$  as the $(r,h)$ element of $\bs \phi_{i,k}(u)$ and $\bs \psi_{i,k}(t,u)$, $\phi^{r,h,(l)}_{i,k}(u)$, $\psi^{r,h,(l)}_{i,k}(t,u)$  as the $(r,h)$ element of $\bs \phi_{i,k}^{(l)}(u)$ and $\bs \psi_{i,k}^{(l)}(t,u)$. Then, we obtain for $t, u \in [0,1]$, $k = 1,\cdots, s_n$, $1 \leq r, h \leq p$, $1 \leq i \leq n$,
    \begin{align}
        & \left\|\sum_{h=1}^p \left( \phi^{r,h}_{i,k}(u)\psi^{r,h}_{i,k}(t,u) - \phi^{r,h, (i-l)}_{i,k}(u)\psi^{r,h,(i-l)}_{i,k}(t,u) \right)\right\|_s \\ 
        & \leq  p\max_{1 \leq h \leq p} \left\|(\phi^{r,h}_{i,k}(u)-\phi^{r,h,(i-l)}_{i,k}(u))\psi^{r,h,(i-l)}_{i,k}(t,u)\right\|_s\\ &  + p \max_{1 \leq h \leq p}\left\|\phi^{r,h}_{i,k}(u)(\psi^{r,h}_{i,k}(t,u) -\psi^{r,h,(i-l)}_{i,k}(t,u))\right\|_s\\ 
        & \leq  p   \max_{1 \leq h \leq p} \left\|\phi^{r,h}_{i,k}(u)-\phi^{r,h,(i-l)}_{i,k}(u)\right\|_{2s} \left\|\psi^{r,h,(i-l)}_{i,k}(t,u)\right\|_{2s} \\ & + p \max_{1 \leq h \leq p} \left\|  \phi^{r,h}_{i,k}(u)(\psi^{r,h}_{i,k}(t,u) - \psi^{r,h,(i-l)}_{i,k}(t,u))\right\|_s\\ 
        & := F_{k,1} + F_{k,2}.\label{eq:Fdecom}
    \end{align}
    By Cauchy-Schwarz inequality and 
\cref{lm:delta}, there exists  a $\chi_1 \in (0,1)$, for $l>0$, we have 
    \begin{align}
       & \sup_{t, u \in [0,1]}\max_{1 \leq r, h \leq p} \left\| \phi^{r,h}_{i,k}(u)- \phi^{r,h,(i-l)}_{i,k}(u)\right\|_{2s}\\ &\leq \sup_{t, u \in [0,1]}\max_{1 \leq r, h \leq p}  \{\|\epsilon^{r,(i-l)}_{i-k}(u) - \epsilon^{r}_{i-k}(u)\|_{4s} \| \epsilon^{h,(i-l)}_{i}(u))  \|_{4s}\\&  + \|\epsilon^{r}_{i-k}(u)\|_{4s} \|\epsilon^{h,(i-l)}_{i}(u) - \epsilon^{h}_{i}(u)  \|_{4s}\}\\
        &\leq C\max_{1 \leq r, h \leq p}  \left\{4s(\delta_{4s, l-k, r}+ \delta_{4s, l, h})\right\}\\ 
        & = O( \chi_1^{l-k}\mf 1(l-k > 0) + \chi_1^{l}\mf 1(l > 0) ).
         \label{eq:phidelta}
    \end{align}
Uniformly for $t, u \in [0,1]$, $1 \leq r, h \leq p$,  by Burkholder inequality, for a large positive constant $C_1$, we have
    \begin{align}
        &\max_{1 \leq i \leq n, 1 \leq k \leq s_n} \left\|\psi^{r,h}_{i,k}(t,u)\right\|_{2s} \\ &\leq \max_{1 \leq i \leq n, 1 \leq k \leq s_n}  \sum_{l \in  \mathbb Z} \left\| \sum_{j=(i-s_n)\vee 1}^{i -M -1} \proj^{j-l} \phi^{r,h}_{j,k}(u)K_{\tau}(j/n -t) \right\|_{2s}\\
        &\leq C_1 \max_{1 \leq i \leq n, 1 \leq k \leq s_n}  \sum_{l \in  \mathbb Z} \left\{ \sum_{j=(i-s_n)\vee 1}^{i -M -1} \left\| \proj^{j-l} \phi^{r,h}_{j,k}(u)K_{\tau}(j/n -t) \right\|^2_{2s} \right\}^{1/2}\\
        &\leq C_1 \max_{1 \leq i \leq n, 1 \leq k \leq s_n}  \sum_{l \in  \mathbb Z} \left\{ \sum_{j=(i-s_n)\vee 1}^{i -M -1} \left\||\phi^{r,h}_{j,k}(u) -  \phi^{r,h, (j-l)}_{j,k}\right\|^2_{2s}|K_{\tau}(j/n -t)|^2 \right\}^{1/2}\\
        & = O(\sqrt{s_n}).\label{eq:psimoment}
    \end{align}
    Notice that $\delta_{4s, l, h} = 0$ when $l < 0$.
    Therefore, for a sufficiently large constant $C$, we have 
    \begin{align}
\sum_{l=m}^{\infty}\sum_{k=1}^{s_n} F_{k,1} 
&\leq C\sqrt{s_n}   \sum_{l=m}^{\infty}\sum_{k=1}^{s_n} (\chi_1^{l-k}\mf  1(l > k) + \chi_1^{l}\mf 1(l > 0)).
    \end{align}
When $ m > 4 s_n$, we have $m-s_n > m/2$. In this case, 
since $\sqrt{s_n} \chi^{s_n} \to 0$, it follows immediately that 
\begin{align}
\sum_{l=m}^{\infty}\sum_{k=1}^{s_n} F_{k,1}  \leq C  \sqrt{s_n} (\chi_1^{m - s_n} + s_n\chi_1^{m})= O( s_n \chi_1^{m/2}), \label{eq:F11}
\end{align}
and
\begin{align}
\sum_{l=m}^{\infty} l\sum_{k=1}^{s_n} F_{k,1}  &\leq C\sqrt{s_n} (m\chi_1^{(m - s_n)} + s_n m\chi_1^{m})= O( s_n m\chi_1^{m/2}). \label{eq:F11+}
\end{align}

When $0 \leq m \leq 4 s_n$, we have 
\begin{align}
\sum_{l=m}^{\infty}\sum_{k=1}^{s_n} F_{k,1} &\leq    C \sqrt{s_n}\sum_{l=m}^{\infty}\left(\sum_{k=1}^{s_n}( \chi_1^{l-k}\mf 1(l-k > 0) + \chi_1^{l}  )\right) \\ 
     & = O\Big[\sqrt{s_n}\left\{(s_n-m)
     \mf 1 (m \leq s_n) + \chi_1^{(m-s_n)}\mf 1(m > s_n)+s_n\chi^m\right\}\Big] = O(s_n^{3/2}),
     \label{eq:F12}
\end{align}
where $C$ is a sufficiently large constant, and \begin{align}
\sum_{l=m}^{\infty}\sum_{k=1}^{s_n} l F_{k,1} &\leq    C \sqrt{s_n}\sum_{l=m}^{\infty}l\left(\sum_{k=1}^{s_n} (\chi_1^{l-k}\mf 1(l-k > 0) + \chi_1^{l} \mf 1(l> 0) )\right)
      = O(s_n^{5/2}).
     \label{eq:F12+}
\end{align}

Similarly,
for $F_{k, 2}$, when $l < k$, we have 
\begin{align}
    & \sup_{t, u \in [0,1]}\max_{1 \leq r, h \leq p} \left\| \phi^{r,h}_{i,k}(u)(\psi^{r,h}_{i,k}(t,u) - \psi^{r,h,(i-l)}_{i,k}(t,u))\right\|_{s}\\ 
    & \leq \sup_{t, u \in [0,1]}\max_{1 \leq r, h \leq p}  \sum_{j=(i-s_n) \vee 1}^{i} \|\phi_{i,k}^{r,h}(u) \|_{2s}  \|\epsilon^{r}_{j-k}(u)\|_{4s} \|\epsilon^{h,(i-l)}_{j}(u) - \epsilon^{h}_{j}(u) \|_{4s} K_{\tau}(j/n - t)\\
    & \leq  C  \sup_{t , u \in [0,1]}\max_{1 \leq r, h \leq p}   \sum_{j=(i-s_n) \vee 1}^{i}  \chi_1^{j-i+l}\mf 1(j-i+l>0)\\ 
    & \leq C \chi_1^{(l-s_n)_+},
\end{align}
where the second inequality follows similar arguments for \eqref{eq:phidelta}.

When $l \geq k$, we have 
\begin{align}
    & \sup_{t , u \in [0,1]}\max_{1 \leq r, h \leq p} \left\| \phi^{r,h}_{i,k}(u)(\psi^{r,h}_{i,k}(t,u)- \psi^{r,h,(i-l)}_{i,k}(t,u))\right\|_{2s}\\ 
    & \leq \sup_{t ,  u \in [0,1]}\max_{1 \leq r, h \leq p} \sum_{j=(i-s_n)\vee 1}^{i -M -1}  \|\phi_{i,k}^{r,h}(u) \|_{2s} ( \|\epsilon^{r}_{j-k}(u)\|_{4s} \|\epsilon^{h,(i-l)}_{j}(u) - \epsilon^{h}_{j}(u)  \|_{4s}\\ & + \|\epsilon^{h, (i-l)}_{j}(u)\|_{4s} \|\epsilon^{r,(i-l)}_{j-k}(u) - \epsilon^{r}_{j-k}(u)  \|_{4s}) K_{\tau}(j/n - t) \\
    & \leq C \chi_1^{(l-s_n - k)_+}.\label{eq:deltapsi}
\end{align}
When $m \geq 4 s_n$, for the second term in \eqref{eq:Fdecom}, since $m-2s_n > m/2$, $\sqrt{s_n} \chi_1^{s_n/2} \to 0$, we have
\begin{align}
\sum_{l=m}^{\infty}\sum_{k=1}^{s_n} F_{k,2} \leq C    \sum_{k=1}^{s_n} \sum_{l=m}^{\infty} \chi^{(l-s_n-k)_+} = O(\sqrt{s_n }\chi_1^{m/2}),\label{eq:F21}
\end{align}
and
\begin{align}
\sum_{l=m}^{\infty} l \sum_{k=1}^{s_n} F_{k,2} \leq C  \sum_{k=1}^{s_n} \sum_{l=m}^{\infty} l \chi^{(l-s_n-k)_+} = O(\sqrt{s_n} m\chi_1^{m/2}).\label{eq:F21+}
\end{align}

When $ m < 4s_n$, we have 
\begin{align}
\sum_{l=m}^{\infty}\sum_{k=1}^{s_n} F_{k,2} &\leq C  \sum_{l=m}^{\infty} \sum_{k=1}^{s_n}\left\{\chi_1^{(l-s_n - k)_+} + \chi_1^{(l-s_n)_+}\right\}\\
& = O\left( \sum_{l=m}^{2s_n } \sum_{k=1}^{s_n}\chi_1^{(l-s_n - k)_+} + \sum_{l=2s_n+1}^{\infty}\chi_1^{l-2s_n} +  s_n \sum_{l=m}^{2s_n} \chi_1^{(l-s_n)_+}+  s_n \sum_{l=2s_n+1}^{\infty} \chi_1^{l-s_n}\right)\\ & = O(s_n^2), \label{eq:F22}
\end{align}
where in the first equality we use the convention for a sequence $x_i$, and positive integers $a$ and $b$, $\sum_{i=a}^b x_i = 0$, if $b < a$.  Using similar arguments in \eqref{eq:F22}, we have 
\begin{align}
\sum_{l=m}^{\infty}l \sum_{k=1}^{s_n} F_{k,2} &\leq C  \sum_{l=m}^{\infty} l\sum_{k=1}^{s_n}\left\{\chi_1^{(l-s_n - k)_+} + \chi_1^{(l-s_n)_+}\right\}\\
&= C  \left\{\sum_{l=m}^{2s_n}l s_n  + \sum_{l=2s_n+1}^{\infty} l \chi_1^{(l-2s_n)}+  \sum_{l=m}^{2s_n} s_nl +  \sum_{l=2s_n+1}^{\infty} s_n l\chi_1^{(l-s_n)} \right\}\\
&   = O(s_n^3), \label{eq:F22+}
\end{align}
where $C$ is a sufficiently large constant.

Therefore, by \eqref{eq:Fdecom}, \eqref{eq:F11}, and \eqref{eq:F21}, for some $\chi_1 \in (0,1)$, when $m > 4s_n$, we have 
\begin{align}
   s_n^{-1} \sum_{l=m}^{\infty}\left\|\sum_{k=1}^{s_n}\sum_{h=1}^p \left( \phi^{r,h}_{i,k}(u)\psi^{r,h}_{i,k}(t,u) - \phi^{r,h, (i-l)}_{i,k}(u)\psi^{r,h,(i-l)}_{i,k}(t,u) \right)\right\|_s\leq C   \chi_1^{m/2},
\end{align}
and by \eqref{eq:Fdecom}, \eqref{eq:F11+}, and \eqref{eq:F21+},
\begin{align}
   s_n^{-1} \sum_{l=m}^{\infty}l\left\|\sum_{k=1}^{s_n}\sum_{h=1}^p \left( \phi^{r,h}_{i,k}(u)\psi^{r,h}_{i,k}(t,u) - \phi^{r,h, (i-l)}_{i,k}(u)\psi^{r,h,(i-l)}_{i,k}(t,u) \right)\right\|_s\leq C m  \chi_1^{m/2},
\end{align}
where $C$ is a sufficiently large constant.
The other terms in \eqref{eq:defineW} follow similarly. By \eqref{eq:defineW} and triangle inequality, we have for $m > 4s_n$,
\begin{align}
\Theta_{m,s} \leq C  \chi_1^{m/2}, \quad \Xi_{m,s} \leq C  m \chi_1^{m/2}.
\end{align}
When $m \leq 4s_n$,  by \eqref{eq:Fdecom}, \eqref{eq:F12}, and \eqref{eq:F22}, we have 
\begin{align}
     &s_n^{-1} \sum_{l=m}^{\infty}\left\|\sum_{k=1}^{s_n}\sum_{h=1}^p \left( \phi^{r,h}_{i,k}(u)\psi^{r,h}_{i,k}(t,u) - \phi^{r,h, (i-l)}_{i,k}(u)\psi^{r,h,(i-l)}_{i,k}(t,u) \right)\right\|_s= O( s_n ),
\end{align}
and by \eqref{eq:Fdecom}, \eqref{eq:F12+}, and \eqref{eq:F22+},
\begin{align}
     &s_n^{-1}\sum_{l=m}^{\infty} l \left\|\sum_{k=1}^{s_n}\sum_{h=1}^p \left( \phi^{r,h}_{i,k}(u)\psi^{r,h}_{i,k}(t,u) - \phi^{r,h, (i-l)}_{i,k}(u)\psi^{r,h,(i-l)}_{i,k}(t,u) \right)\right\|_s = O(s_n^2).
\end{align}
 By \eqref{eq:defineW} and triangle inequality, we have for $m \leq 4s_n$,
\begin{align}
\Theta_{m,s} =O\left(s_n\right), \quad \Xi_{m,s} = O\left(s_n^2\right).
\end{align}
\end{proof}

\begin{lemma}
 Under the condition of \cref{thm:ga},  under the null hypothesis, we have 
 \begin{align}
     &n\tau s_n \sup_{t \in [\tau, 1-\tau], u \in [0,1]}\left |\sum_{k=1}^{s_n} \E\left[\mathrm{tr}\{\mf G_k(t,u)\}\right]  \right|\\& =
     \sup_{t \in [\tau, 1-\tau], u \in [0,1]}\left |\sum_{i=1}^n \E\left[\mathrm{tr}\{\mf U_i(t,u)\}\right] K_{\tau}(i/n -t) \right|  = O(\tau s_n^{5/2}) = o(\sqrt{n\tau} s_n).
 \end{align}
 \label{lm:4th}
\end{lemma}
\begin{proof}
Let $m =\lfloor  c (\log n) /4\rfloor$. Recall that $\check{\bs \epsilon}_{i}(u) :=  \check{\bs G}(i/n, u, \F_i) = \E(\bs \epsilon_{i}(u)| \F_{i-m}^i)$. Define $\check{\bs \phi}_{i,k}(u) = \check{\bs \epsilon}_{i-k,n}(u)\check{\bs \epsilon}^{\top}_{i,n}(u)$, and 
\begin{align}
    \sum_{i=1}^n \check{\mf U}_i(t,u) K_{\tau}(i/n -t) = \sum_{k=1}^{s_n} \sum_{i=1}^{n} K_{\tau}(i/n-t) \check{\bs \phi}_{i,k} (u)\sum_{\substack{1 \leq l \leq n\\ M < |l-i|\leq s_n}}K_{\tau}(l/n-t)\check{\bs \phi}^{\top}_{l,k}(u). \label{eq:defcheckU}
\end{align}
The proof consists of two steps. 

First, we shall approximate $\sum_{i=1}^n \mathrm{tr}\{\mf U_i(t,u)\} K_{\tau}(i/n -t)$ by $\sum_{i=1}^n \mathrm{tr}\{\check{\mf U}_i(t,u)\} K_{\tau}(i/n -t)$. Then, calculate the expectation of $\sum_{i=1}^n \mathrm{tr}\{\check{\mf U}_i(t,u)\} K_{\tau}(i/n -t)$.\\
\textbf{Step 1}
By triangle inequality and Cauchy inequality, we have 
\begin{align}
   & \E \left|\sum_{i=1}^n \mathrm{tr}\{\mf U_i(t,u) - \check{\mf U}_i(t,u)\} K_{\tau}(i/n -t) \right| \\  &\leq  \sum_{k=1}^{s_n }\sum_{i=1}^{n} |K_{\tau}(i/n-t)| \|{\bs \phi}_{i,k} (u)- \check{\bs \phi}_{i,k} (u)\|_{F,2} \left\|\sum_{\substack{1 \leq l \leq n\\ M < |l-i|\leq s_n}}K_{\tau}(l/n-t)\check{\bs \phi}_{l,k}(u)\right\|_{F,2} \\  &+  \sum_{k=1}^{s_n} \sum_{i=1}^{n} |K_{\tau}(i/n-t)| \|{\bs \phi}_{i,k} (u)\|_{F,2} \left\| \sum_{\substack{1 \leq l \leq n\\ M < |l-i|\leq s_n}} K_{\tau}(l/n-t) ({\bs \phi}_{l,k}(u) - \check{\bs \phi}_{l,k}(u))\right\|_{F,2}.\label{eq:mapproxU}
\end{align}
Similar to the steps in proving \eqref{aikmin}, for a fixed $q$, for all $1 \leq k \leq s_n$, by \cref{lm:delta}, we can obtain 
\begin{align}\allowdisplaybreaks
&\sup_{t, u \in [0,1]} \left\| \sum_{\substack{1 \leq l \leq n\\ M < |l-i|\leq s_n}} K_{\tau}(l/n-t) ({\bs \phi}_{l,k}(u) - \check{\bs \phi}_{l,k}(u)) \right\|_{F,2}  \\ 
&\leq \sum_{h = 0}^{\infty} \sup_{t, u \in [0,1]}  \left\| \sum_{\substack{1 \leq l \leq n\\ M < |l-i|\leq s_n}} \proj^{l-h} (\check{\bs \epsilon}_{l-k}(u)  \check{\bs \epsilon}^{\T}_{l}(u) - {\bs \epsilon}_{l-k}(u)  {\bs \epsilon}^{\T}_{l}(u)) K_{\tau}(l/n - t)  \right\|_{F,2}  \\
& \leq  \sum_{h=0
}^{\infty} \sup_{t, u \in [0,1]}  C \left\{\sum_{\substack{1 \leq l \leq n\\ M < |l-i|\leq s_n}}\left\| \proj^{l-h} (\check{\bs \epsilon}_{l-k}(u)  \check{\bs \epsilon}^{\T}_{l}(u) - {\bs \epsilon}_{l-k}(u)  {\bs \epsilon}^{\T}_{l}(u)) K_{\tau}(l/n - t)   \right\|_{F,2}^2\right\}^{1/2}\\
& \leq \sum_{h=0
}^{\infty} \sup_{t, u \in [0,1]}  C \left\{\sum_{l=(i-s_n)\vee 1 }^{(i+s_n) \wedge n}\min\left\{\|\check {\bs \epsilon}_{l-k}(u) -  {\bs \epsilon}_{l-k}(u)\|_{F,4} \|  {\bs \epsilon}_{l}(u)\|_{F,4} \right.\right. \\&\left.\left.+ \| \check {\bs \epsilon}_{l}(u) - {\bs \epsilon}_{l}(u)\|_{F, 4} \|  \check {\bs \epsilon}_{l-k}(u)\|_{F, 4},  \right.\right. \\&\left. \left. \left\|{\bs \phi}_{l, k}(u) - {\bs \phi}^{(l-h) }_{l, k}(u)\right\|_{F,2} + \left\|\check {\bs \phi}_{l, k}(u) - \check{\bs \phi}^{(l-h) }_{l, k}(u)\right\|_{F,2} \right\}^2|K_{\tau}(l/n - t)|^2 \right\}^{1/2}\\
& \leq C\sqrt{s_n}\sum_{h = 0}^{\infty}   \min \{ \chi^{m/8}, \chi^{(h-k)/8} \mathbf 1(h > k) + \chi^{h/8} \mathbf 1(h > 0) \} \\
& \leq C\sqrt{s_n}(m + k)\chi^{m/8}+ C\sqrt{s_n}\sum_{h = m + k}^{\infty} \min \{ \chi^{m/8}, \chi^{(h-k)/8} \mathbf 1(h > k) + \chi^{h/8}\mathbf 1(h > 0)\}\\&=O(  \sqrt{s_n}(\log n+ s_n)/n),\label{eq:mapproxsumphi}
\end{align}
where  in the third inequality we use $\| \proj^{l-h}\{\cdot\}\|_{F,2} \leq \| \cdot\|_{F,2}$, which follows from Jensen inequality; in the last inequality we consider the summation of $h$ smaller than $m+k$ and the summation of $h$ equal or greater than $m+k$ separately.
Note that by  \cref{lm:delta}, Cauchy inequality and Jensen inequality $\|{\bs \phi}_{i,k} (t,u)\|_{F,2} = O(1)$ and $\|{\bs \phi}_{i,k} (t,u)- \check{\bs \phi}_{i,k} (t,u)\|_{F,2} = O(1)$. Using similar arguments in \eqref{eq:aikin}, we have 
\begin{align}
    \left\|\sum_{\substack{1 \leq l \leq n\\ M < |l-i|\leq s_n}}K_{\tau}(l/n-t)\check{\bs \phi}_{l,k}(u)\right\|_{F,2} = O(\sqrt{s_n}).
\end{align}
Together with \eqref{eq:mapproxU} and \eqref{eq:mapproxsumphi}, we have 
\begin{align}
   \sup_{t, u \in [0,1]} \left\|\sum_{i=1}^n \mathrm{tr}\{\mf U_i(t,u) - \check{\mf U}_i(t,u)\} K_{\tau}(i/n -t) \right\|_{F,2}  = O((n\tau)s_n^{3/2} (\log n+s_n)/n).
\end{align}
Therefore, we have 
\begin{align}
     \sup_{t, u \in [0,1]} \left|\sum_{i=1}^n \E[\mathrm{tr}\{\mf U_i(t,u) - \check{\mf U}_i(t,u)\}] K_{\tau}(i/n -t) \right|&  = O((n\tau)s_n^{3/2} (\log n+s_n)/n)\\& = O(\tau s_n^{5/2}).\label{eq:EUstep1}
\end{align}

\textbf{Step 2} We first consider $m < \lfloor k/3 \rfloor$. 
Notice that for $m < \lfloor k/3 \rfloor$, $\check{\bs \epsilon}_{i-k}^{(m)}(u)$ and $\check{\bs \epsilon}^{(m)}_{i}(u)$ are independent. 
Note that $i-k+m < i - m$. Without the loss of generality, we consider $j < i$.  Note that $M > 4m$.  For $|j - i| > M$, one of the followings must hold: (1) $j < i-m$ (2) $j > i-k+m $. Recall that  $\check \epsilon_{i}^h(u) \in  \F_{i-m}^i$.
If $j < i-m$, we have
\begin{align}
     \E\{\mathrm{tr}(\check{\bs \phi}_{i,k}(u) \check{\bs \phi}^{\top}_{j,k}(u))\} &= \sum_{r,h=1}^p \E\{\check \epsilon_{i-k}^r(u) \check \epsilon_{i}^h(u) \check  \epsilon_{j-k}^r(u) \check  \epsilon_{j}^h(u)\}\\
     & =  \sum_{r,h=1}^p \E\{\check \epsilon_{i-k}^r(u) \check \epsilon_{i}^h(u) \check  \epsilon_{j}^r(u)\} \E\{\check \epsilon_{j-k}^h(u)\} = 0.\label{eq:checkphiij1}
\end{align}
If $j > i-k+m$, note that $\check \epsilon_{i-k}^r(u) \check \epsilon_{j-k}^r(u)$ and $\check \epsilon_{i}^h(u)  \check \epsilon_{j}^h(u)\}$ are independent, then we have
\begin{align}
     \E\{\mathrm{tr}(\check{\bs \phi}_{i,k}(u) \check{\bs \phi}^{\top}_{j,k}(u))\} &= \sum_{r,h=1}^p \E\{\check \epsilon_{i-k}^r(u) \check \epsilon_{i}^h(u)  \check \epsilon_{j-k}^r(u) \check \epsilon_{j}^h(u)\}\\
     & =  \sum_{r,h=1}^p \E\{\check \epsilon_{i-k}^r(u) \check \epsilon_{j-k}^r(u)\} \E\{\check \epsilon_{i}^h(u)  \check \epsilon_{j}^h(u)\}\\
     & =  \sum_{r,h=1}^p \Gamma_{i-j}^{r,r}((i-k)/n,u) \Gamma_{i-j}^{r,r}(i/n,u) + O(\chi^m)\\ 
     &= O(\chi^m) ,\label{eq:checkphiij2}
\end{align}
where the last equality follows from $|j-i|>M$.
Combining 
\eqref{eq:checkphiij1} and \eqref{eq:checkphiij2}, for $m < \lfloor k/3 \rfloor$, we have 
\begin{align}
    &\sum_{i=1}^{n} K_{\tau}(i/n-t) \sum_{k=3m+1}^{s_n} \E \left[\mathrm{tr}\left\{\check{\bs \phi}_{i,k} (u)\sum_{\substack{1 \leq l \leq n\\ M < |l-i|\leq s_n}} K_{\tau}(l/n-t)\check{\bs \phi}^{\top}_{l,k}(t,u)\right\} \right]\\& = O(n\tau s_n^2 \chi^m) = O(\tau s_n).
    \label{eq:largeKE}
\end{align}
Combining \eqref{eq:defcheckU} and \eqref{eq:largeKE}, under the null hypothesis
we have 
\begin{align}
       &\sum_{i=1}^n \E[\mathrm{tr}\check{\mf U}_i(t,u)] K_{\tau}(i/n -t)\\ &= \sum_{k=1}^{3m}\sum_{i=1}^{n} K_{\tau}(i/n-t) \E \left[\mathrm{tr}\left\{\check{\bs \phi}_{i,k} (u)\sum_{M <  |l-i| \leq s_n} K_{\tau}(l/n-t)\check{\bs \phi}^{\top}_{l,k}(u)\right\} \right] + O(\tau s_n)  \\
       & =\sum_{k=1}^{3m}\sum_{i=1}^{n} K_{\tau}(i/n-t) \left[\mathrm{tr}\left\{\bs \Gamma_k(i/n, u)\sum_{M <  |l-i| \leq s_n} K_{\tau}(l/n-t)\bs \Gamma_k^{\top}(l/n, u)\right\} \right]+ O(n\tau m s_n \chi^m),\label{eq:smallKE}
\end{align}
where the last equation is due to the fact that $\check{\bs \phi}_{i,k} (t,u)$ and $\check{\bs \phi}_{l,k} (t,u)$ are independent, since $l<i-M< i-m-k$. Recall that $m = \lfloor c(\log n)/4 \rfloor$. 
Therefore, under the null hypothesis, 
\begin{align}
   \sum_{i=1}^n \E[\mathrm{tr}\check{\mf U}_i(t,u)] K_{\tau}(i/n -t) = O(\tau (\log n) s_n). \label{eq:EUstep2}
\end{align}
The result follows from \eqref{eq:EUstep1} and \eqref{eq:EUstep2}.

\end{proof}
Recall the definition of $\mf U_i(t,u)$ in \eqref{eq:defU}.
\begin{align}
     \mf U_i(t, u) =\sum_{k=1}^{s_n}(\bs \phi_{i,k}(u)\bs \psi^{\top}_{i,k}(t,u) + \bs \psi_{i,k}(t,u)\bs \phi^{\top}_{i,k}(u)) ,\label{eq:recallU}
 \end{align}
and  $\bs \phi_{i,k}(u) = \bs \epsilon_{i-k}(u) \bs \epsilon^{\top}_{i}(u)$, $\bs \psi_{i,k}(t,u) =\sum_{\ell = (i-s_n) \vee 1}^{i- M - 1}  \bs \phi_{\ell ,k}(u)K_{\tau}(\ell /n - t)$. 
Define \begin{align}
    \mf G_k(t,u) 
    &=(n\tau s_n)^{-1} \sum_{i=1}^n K_{\tau}(i/n - t)\bs (\bs \phi_{i,k}(u)\bs \psi^{\top}_{i,k}(t,u) + \bs \psi_{i,k}(t,u)\bs \phi^{\top}_{i,k}(u)) \label{eq:defGk0}\\ 
    &=(n\tau s_n)^{-1} \sum_{i=1}^n K_{\tau}(i/n - t)\bs \phi_{i,k}(u) \sum_{l=1}^n K_{\tau}(l/n - t) \bs \phi^{\top}_{l,k}(u)\mf 1( l \in \mathbb  L_i),\label{eq:defGk}
\end{align}
where $\mathbb  L_i = \{l| M < |l-i|\leq s_n\}$, and the second equality follows from a change of the order of summations of $i$ and $l$.
 
Let 
\begin{align}
    \tilde \sigma (t,u) = \frac{1}{s_n^2n\tau}\sum_{i,j=1}^n K_{\tau}(i/n-t) K_{\tau}(j/n-t) \mathrm{Cov}\left(\mathrm{tr} \{\mf U_i(t,u) \},\mathrm{tr} \{\mf U_j(t,u) \} \right).\label{eq:deftildesigma}
\end{align}
Note that 
\begin{align}
\sum_{k=1}^{s_n}  {\mf G}_k(t,u) =  \frac{1}{n\tau s_n} \sum_{i=1}^{n}  K_{\tau}(i/n - t)  {\mf  U}_i(t, u),
\end{align}
and hence
\begin{align}
    \tilde \sigma (t,u) = s_n^{-1}\sum_{k,l=1}^{s_n} \mathrm{Cov}\left(\mathrm{tr} \{\mf G_k(t,u) \},\mathrm{tr} \{\mf G_l(t,u) \} \right).
\end{align}
\begin{lemma}
Under the conditions of \cref{thm:ga} and $H_0$, we have uniformly for $t \in [\tau, 1-\tau]$, $u \in [0,1]$, 
\begin{align}
\left|\tilde \sigma(t,u) -  4\phi^2  \mathrm{tr}^2\{\bs \Gamma^2_0(t,u)\}-  4\phi^2  \mathrm{tr}\{\bs \Gamma^4_0(t,u)\} \right| = o(1),
\end{align}
where $\phi^2 = \int K^2(t) dt$.
\label{lm:cum}
\end{lemma}

\begin{proof}
Define $$\sigma_{r_1,q_1, r_2, q_2}(t,u) = \frac{1}{s_n^2n \tau} \sum_{i,j=1}^{n} K_{\tau}(i/n-t) K_{\tau}(j/n-t) \mathrm{Cov}(U_i^{r_1, q_1}(t,u), U_j^{r_2, q_2}(t,u)),$$
so that by definition in \eqref{eq:deftildesigma}
\begin{align}
    \tilde \sigma(t,u) =  \sum_{r_1, r_2 = 1}^p \sigma_{r_1,r_1, r_2, r_2}(t,u).\label{eq:tildeSigma}
\end{align}
Recall that in \eqref{eq:defGk}
\begin{align}
    \mf G_k(t,u) =(n\tau s_n)^{-1} \sum_{i=1}^n K_{\tau}(i/n - t)\bs \phi_{i,k}(u) \sum_{l=1}^n K_{\tau}(l/n - t) \bs \phi^{\top}_{l,k}(u)\mf 1( l \in \mathbb  L_i).
\end{align}
Let $G_k^{r_1, q_1}(t,u)$ denote the $(r_1, q_1)$th element of $\mf G_k(t,u)$.
Note that 
\begin{align}
\sigma_{r_1,q_1, r_2, q_2}(t,u) = (n\tau) \sum_{k,h=1}^{s_n}  \mathrm{Cov}(G_k^{r_1, q_1}(t,u), G_h^{r_2, q_2}(t,u)).
    \label{eq:sigmatu}
\end{align}
We shall consider small lags and large lags separately. Recall that $M > \lfloor c \log n \rfloor$. Take $m = \lfloor c \log n/4 \rfloor$ for some large constant $c$. Then, we have $\chi^m = O(1/\sqrt{n\tau})$ and $s_n /m \to \infty$. By \cref{lm:4th}, we have
\begin{align}
    &(n\tau s_n)^{2} \sum_{k,h=1}^{3m}\mathrm{Cov}(G_k^{r_1, r_1}(t,u), G_h^{r_2, r_2}(t,u)) \\ &=   (n\tau s_n)^{2}\sum_{k,h=1}^{3m} \left\{\mathbb{E}(G_k^{r_1, r_1}(t,u) G_h^{r_2, r_2}(t,u)) -  \mathbb{E}(G_k^{r_1, r_1}(t,u))\mathbb{E}( G_h^{r_2, r_2}(t,u))  \right\}\\
    &=   (n\tau s_n)^{2}\sum_{k,h=1}^{3m} \mathbb{E}(G_k^{r_1, r_1}(t,u) G_h^{r_2, r_2}(t,u)) + o(n\tau s_n^2)\\
    & =  (n\tau s_n)^{2} \sum_{k,h=1}^{3m}  \mathrm{Cum}(G_k^{r_1, r_1}(t,u), G_h^{r_2, r_2}(t,u)) + o(n\tau s_n^2)\\
    & = \sum_{q_1, q_2 = 1}^p \sum_{k,h=1}^{3m} \sum_{i,j=1}^n K_{\tau}(i/n-t) K_{\tau}(j/n-t)   \\ &\times  \sum_{l \in \mathbb L_i, l^{\prime} \in \mathbb L_j} \left\{\mathrm{Cum}(\phi_{i,k}^{r_1, q_1}(t,u), \phi_{l,k}^{r_1, q_1}(t,u), \phi_{j,h}^{r_2, q_2}(t,u), \phi_{l^{\prime},h}^{r_2, q_2}(t,u)) \right.\\ & \left.+ \E(\phi_{i,k}^{r_1, q_1}(t,u)\phi_{j,h}^{r_2, q_2}(t,u))\E( \phi_{l,k}^{r_1, q_1}(t,u)\phi_{l^{\prime},h}^{r_2, q_2}(t,u) )\right.\\ & \left.+ \E(\phi_{i,k}^{r_1, q_1}(t,u)\phi_{l^{\prime},h}^{r_2, q_2}(t,u))\E( \phi_{l,k}^{r_1, q_1}(t,u)\phi_{j,h}^{r_2, q_2}(t,u) ) \right\} 
    + o(n\tau s_n^2) \\ 
    & = O(n\tau m^2) + o(n\tau s_n^2) = o(n\tau s_n^2),\label{eq:mm}
\end{align}
where the fourth equality follows from \cref{lm:cumsum} and similar arguments in  \cref{lm:4th}.  
Similarly, we have 
\begin{align}
&(n\tau s_n)^{2} \sum_{k=1}^{3m}\sum_{h=3m+1}^{s_n}\mathrm{Cov}(G_k^{r_1, r_1}(t,u), G_h^{r_2, r_2}(t,u))\\ &= \sum_{q_1, q_2 = 1}^p \sum_{k=1}^{3m}\sum_{h=3m+1}^{s_n} \sum_{i,j=1}^n K_{\tau}(i/n-t) K_{\tau}(j/n-t) \\&\times \sum_{l \in \mathbb L_i, l^{\prime} \in \mathbb L_j} \left\{\mathrm{Cum}(\phi_{i,k}^{r_1, q_1}(t,u), \phi_{l,k}^{r_1, q_1}(t,u), \phi_{j,h}^{r_2, q_2}(t,u), \phi_{l^{\prime},h}^{r_2, q_2}(t,u)) \right.\\ & \left.+ \E(\phi_{i,k}^{r_1, q_1}(t,u)\phi_{j,h}^{r_2, q_2}(t,u)\E( \phi_{l,k}^{r_1, q_1}(t,u)\phi_{l^{\prime},h}^{r_2, q_2}(t,u) )\right.\\ & \left.+ \E(\phi_{i,k}^{r_1, q_1}(t,u)\phi_{l^{\prime},h}^{r_2, q_2}(t,u)\E( \phi_{l,k}^{r_1, q_1}(t,u)\phi_{j,h}^{r_2, q_2}(t,u) ) \right\} 
    +o(n\tau s_n^2) \\ 
    & = O(n\tau m s_n)+o(n\tau s_n^2) = o(n\tau s_n^2).\label{eq:Gmlarge}
\end{align}
Finally, for $k, h > 3m$, by \cref{lm:4th}, we have 
\begin{align}
&(n\tau s_n)^{2}  \sum_{k, h=3m+1}^{s_n}\mathrm{Cov}(G_k^{r_1, r_1}(t,u), G_h^{r_2, r_2}(t,u)) \\ 
& =  (n\tau s_n)^{2} \sum_{k,h=3m+1}^{s_n}  \left\{\mathbb{E}(G_k^{r_1, r_1}(t,u) G_h^{r_2, r_2}(t,u))\right\} + o(n\tau s_n^2).\label{eq:approxG}
\end{align}
At the end of the proof, we shall show that 
\begin{align}
    & (n\tau s_n)^{2} 
 \sum_{k,h=3m+1}^{s_n}  \left\{\mathbb{E}(G_k^{r_1, r_1}(t,u)G_h^{r_2, r_2}(t,u))\right\}\\
 &=4s_n^2 n\tau \phi^2 \sum_{q_1, q_2 = 1}^p [\{\Gamma_0^{r_1, q_2} (t, u) \Gamma_0^{r_2, q_1}(t, u)\}\times \{\Gamma_0^{q_1, q_2} (i/n, u) \Gamma_0^{r_1, r_2}(i/n, u)\}\\
 &+\{\Gamma_0^{r_1, r_2} (t, u) \Gamma_0^{q_1, q_2}(t, u)\}^2] + o(s^2_n n\tau).\label{eq:Glargelarge}
\end{align}
Let $\Gamma_0^{2,r_1, q_1} (t, u)$ denote the $(r_1, q_1)$-th element of $\bs \Gamma_0^{2} (t, u)$. Combining \eqref{eq:tildeSigma}, \eqref{eq:sigmatu}, \eqref{eq:mm}, \eqref{eq:Gmlarge}, \eqref{eq:Glargelarge}, we have 

\begin{align}
\tilde \sigma(t,u) &=  \sum_{r_1, r_2 = 1}^p \sigma_{r_1,r_1, r_2, r_2}(t,u)\\ 
&= (n\tau) \sum_{r_1, r_2 = 1}^p\sum_{k,h=1}^{s_n}  \mathrm{Cov}(G_k^{r_1, r_1}(t,u), G_h^{r_2, r_2}(t,u))\\
& = 4 \phi^2 \sum_{r_1, r_2, q_1, q_2 = 1}^p [\{\Gamma_0^{r_1, q_2} (t, u) \Gamma_0^{r_2, q_1}(t, u)\}\times \{\Gamma_0^{q_1, q_2} (i/n, u) \Gamma_0^{r_1, r_2}(i/n, u)\}\\
&+\{\Gamma_0^{r_1, r_2} (t, u) \Gamma_0^{q_1, q_2}(t, u)\}^2]  + o(1)\\
& = 4 \phi^2 \sum_{r_1, q_1= 1}^p [\Gamma_0^{2,r_1, q_1} (t, u) \times \Gamma_0^{2,r_1, q_1} (t, u)] + 4 \phi^2  \mathrm{tr}^2\{\bs \Gamma^2_0(t,u)\}    + o(1)\\
&=  4 \phi^2  \mathrm{tr}\{\bs \Gamma^4_0(t,u)\} + 4 \phi^2  \mathrm{tr}^2\{\bs \Gamma^2_0(t,u)\}   +  o(1).
\end{align}
\par

\textbf{Proof of \eqref{eq:Glargelarge}. }
Define 
\begin{align}
   \check{\mf G}_k(t,u) &=(n\tau s_n)^{-1} \sum_{i=1}^n K_{\tau}(i/n - t)\bs (\check{\bs \phi}_{i,k}(u)\check{\bs \psi}^{\top}_{i,k}(t,u) + \check{\bs \psi}_{i,k}(t,u)\check{\bs \phi}^{\top}_{i,k}(u))\label{eq:defcheckG0} \\ &= (n\tau s_n)^{-1} \sum_{i,l =1}^n K_{\tau}(i/n - t)  K_{\tau}(l/n - t) \check{\bs \phi}_{i,k}(u) \check{\bs \phi}^{\top}_{l,k}(u)\mf 1( l \in \mathbb  L_i),\label{eq:defcheckG}
\end{align}
where $\check{\bs \psi}_{i,k}(t,u) = \sum_{\ell = (i-s_n) \vee 1}^{i- M - 1}  K_{\tau}( \ell/n - t) \check{\bs \phi}_{\ell, k}(u),$ and the second equality follows from a change of the order of summations of $i$ and $l$.
Using $m$-approximation, uniformly for $1 \leq r_1, r_2 \leq p$, we have
\begin{align}
   & \left|  \sum_{k,h=3m+1}^{s_n}  \left\{\mathbb{E}(G_k^{r_1, r_1}(t,u) G_h^{r_2, r_2}(t,u))\right\} - \sum_{k,h=3m+1}^{s_n}  \left\{\mathbb{E}(\check G_k^{r_1, r_1}(t,u)\check G_h^{r_2, r_2}(t,u))\right\} \right| \\ 
   & \leq \sum_{k,h=3m+1}^{s_n}  \left| \left\{\mathbb{E}(G_k^{r_1, r_1}(t,u), G_h^{r_2, r_2}(t,u))\right\} - \left\{\mathbb{E}(\check G_k^{r_1, r_1}(t,u) \check G_h^{r_2, r_2}(t,u))\right\} \right| \\ 
   & \leq \sum_{k,h=3m+1}^{s_n} \| G_k^{r_1, r_1}(t,u) - \check G_k^{r_1, r_1}(t,u)\|\|G_h^{r_2, r_2}(t,u) \|\\ 
   &   + \|\check G_k^{r_1, r_1}(t,u)\|\|G_h^{r_2, r_2}(t,u) -\check G_h^{r_2, r_2}(t,u) \| .\label{eq:mapproxEGG}
\end{align}
Similar to the calculation in \eqref{eq:mapproxU} and \eqref{eq:mapproxsumphi}, for all $3m+1 \leq k \leq s_n$,  we have 
\begin{align}
    \| G_k^{r_1, r_1}(t,u) - \check G_k^{r_1, r_1}(t,u)\| = O( (\log n+s_n) /(n\sqrt{s_n}) ), \label{eq:EGGm1}
\end{align}
and by similar arguments in \eqref{eq:mm}, for all $3m+1 \leq k \leq s_n$, we have
\begin{align}
    &\|G_k^{r_1, r_1}(t,u)\|^2 
=O((n\tau s_n^2)^{-1}),\label{eq:EGGm2}
\end{align}
and similar results hold for $\|\check G_k^{r_1, r_1}(t,u)\|^2$.
Therefore, since $s_n/\sqrt{n\tau} \to 0$ and $s_n/ (\log n) \to \infty$, by \eqref{eq:mapproxEGG}, \eqref{eq:EGGm1} and \eqref{eq:EGGm2}, we have 
\begin{align}
     &(n\tau s_n)^2\left|  \sum_{k,h=3m+1}^{s_n}  \left\{\mathbb{E}(G_k^{r_1, r_1}(t,u) G_h^{r_2, r_2}(t,u))\right\} - \sum_{k,h=3m+1}^{s_n}  \left\{\mathbb{E}(\check G_k^{r_1, r_1}(t,u) \check G_h^{r_2, r_2}(t,u))\right\} \right|\\&  = o(n\tau s_n^2). \label{eq:mapproxG}
\end{align}
Recall the definition of  $\check{\mf G}_k(t,u)$ in \eqref{eq:defcheckG0}. We can further write the element on the $r_1$-th column and $r_1$-th row of $\check{\mf G}_k(t,u)$ as $\check G_k^{r_1,r_1}(t,u) = \sum_{i=1}^n K(i/n -t) \varrho_{i,k}^{r_1}(t,u)$, where
\begin{align}
    \varrho_{i,k}^{r_1}(t,u) = 2\sum_{q_1 = 1}^p \check{ \phi}^{r_1, q_1}_{i,k}(u) \sum_{l=(i-s_n)
    \vee 1}^{i-M-1} K_{\tau}(l/n - t) \check{\phi}^{r_1, q_1}_{l,k}(u), 
\end{align}
Therefore, we can write 
\begin{align}
    &(n\tau s_n)^{2}  \sum_{k,h=3m+1}^{s_n}  \left\{\mathbb{E}(\check G_k^{r_1, r_1}(t,u) \check G_h^{r_2, r_2}(t,u))\right\} \\
    &= \sum_{k,h=3m+1}^{s_n} \sum_{i,j=1}^n K(i/n -t)K(j/n -t)  \mathbb{E}(\varrho_{i,k}^{r_1}(t,u),\varrho_{j,h}^{r_2}(t,u)).\label{eq:Glarge}
\end{align}
In the following, we proceed to discuss $\mathbb{E}(\varrho_{i,k}^{r_1}(t,u),\varrho_{j, h}^{r_2}(t,u))$, where $k,h >3m$.
Note that $\check{\bs \epsilon}_i(u) \in \FF_{i-m}^i$, $\check{\bs \epsilon}_{i-k}(u) \in \FF_{i-k-m}^{i-k}$ and $\sum_{l=(i-s_n)
    \vee 1}^{i-M-1} K_{\tau}(l/n - t) \check{\bs \phi}^{\top}_{l,k}(u) \in \FF_{i-k-s_n-m}^{i-M-1}$, $M > 4m$. Therefore, for $k, h > 3m$,  $j\leq i$, if $j < i-m$, for $1 \leq r_1, r_2 \leq p$, we have
\begin{align}
   \mathbb{E}(\varrho_{i,k}^{r_1}(t,u)\varrho_{j,h}^{r_2}(t,u)) = 0.\label{eq:case(0)}
\end{align}
For $k, h > 3m$, if $i-m \leq j \leq i$, $\{\check{\bs \epsilon}_{i}(u), \check{\bs \epsilon}_{j}(u)\} $ and $$\{\check{\bs \epsilon}_{i-k}(u), \sum_{l=(i-s_n)
    \vee 1}^{i-M-1} K_{\tau}(l/n - t) \check{\bs \phi}^{\top}_{l,k}(u), \ \check{\bs \epsilon}_{j-h}(u), \sum_{l=(j-s_n)
    \vee 1}^{j-M-1} K_{\tau}(l/n - t) \check{\bs \phi}^{\top}_{l,h}(u)\}$$ are independent.\par 
Under null hypothesis, using the arguments similar to \eqref{eq:checkphiij2} but simpler, $\E (\check \epsilon^{q_1}_i(u) \check 
 \epsilon^{q_2}_j(u)) = O(\chi^m)$ if $i \neq j$. Therefore, for $i \neq j$, we have 
\begin{align}
   & \mathbb{E}(\varrho_{i,k}^{r_1}(t,u)\varrho_{j,h}^{r_2}(t,u))\\ 
    & = 4\sum_{q_1, q_2 = 1}^p \E (\check \epsilon^{q_1}_i(u) \check \epsilon^{q_2}_j(u))\\& \times \E(
    \check \epsilon^{r_1}_{i-k}(u)
    \sum_{l=(i-s_n)
    \vee 1}^{i-M-1} K_{\tau}(l/n - t) \check{\phi}^{r_1, q_1}_{l,k}(u)  \check \epsilon^{r_2}_{j-h}(u) \sum_{l=(j-s_n)
    \vee 1}^{j-M-1} K_{\tau}(l/n - t) \check{ \phi}^{r_2, q_2}_{l,h}(u)) \\ &= O(s_n^2 \chi^m) = o(s_n^2/\sqrt{n\tau}).\label{eq:caseij}
\end{align}
It remains to consider the case where $i = j$. Namely, the quantity
\begin{align}
     &\mathbb{E}(\varrho_{i,k}^{r_1}(t,u)\varrho_{i,h}^{r_2}(t,u))\\ 
     & = O(s_n^2 \chi^m) +4 \sum_{q_1, q_2 = 1}^p \Gamma_0^{q_1, q_2}(i/n, u)\\ & \times \E(
    \epsilon^{r_1}_{i-k}(u)
    \sum_{l=(i-s_n)
    \vee 1}^{i-M-1} K_{\tau}(l/n - t) \check{\phi}^{r_1, q_1}_{l,k}(u)  \epsilon^{r_2}_{i-h}(u) \sum_{l^{\prime}=(i-s_n)
    \vee 1}^{i-M-1} K_{\tau}(l^{\prime}/n - t) \check{ \phi}^{r_2, q_2}_{l^{\prime},h}(u))\label{eq:rhoii}
\end{align}Without loss of generality, consider $h \leq k$. Note that $l-k < i-k-m$, $l^{\prime} -h < i-k-m$.
Hence, $\check \epsilon_{l-k}^{r_1}(u)$ is independent of $\{\check \epsilon^{r_1}_{i-k}(u),  \check \epsilon^{r_2}_{i-h}(u), \check \epsilon_{l}^{r_1}(u) \}$, since $k,h > 3m$. In the following, we will discuss three cases:
(1) $k=h$; (2) $k > h+2m$; (3) $h<k \leq h+2m$ . \par 
For (1), we  have 
\begin{align}
    &\mathbb{E}(\varrho_{i,k}^{r_1}(t,u)\varrho_{i,h}^{r_2}(t,u))\\ 
    & =O(s_n^2 \chi^m) +4\sum_{q_1, q_2 = 1}^p \Gamma_0^{q_1, q_2}(i/n, u)\\ & \times \E(
    \epsilon^{r_1}_{i-k}(u)
    \sum_{l=(i-s_n)
    \vee 1}^{i-M-1} K_{\tau}(l/n - t) \check{\phi}^{r_1, q_1}_{l,k}(u) \epsilon^{r_2}_{i-k}(u) \sum_{l^{\prime}=(i-s_n)
    \vee 1}^{i-M-1} K_{\tau}(l^{\prime}/n - t) \check{ \phi}^{r_2, q_2}_{l^{\prime}, k}(u))\\ & = 4 \sum_{q_1, q_2 = 1}^p \Gamma_0^{q_1, q_2}(i/n, u)\sum_{l, l^{\prime}=(i-s_n)
    \vee 1}^{i-M-1} K_{\tau}(l/n - t)K_{\tau}(l^{\prime}/n - t) \E (  \check  \epsilon_{l-k}^{r_1}(u) \check \epsilon_{l^{\prime} -k}^{r_2}(u)) \\
&\times \E ( \check  \epsilon_{i-k}^{r_1}(u) \check \epsilon_{i -k}^{r_2}(u)  \check 
 \epsilon_{l}^{q_1}(u) \check \epsilon_{l^{\prime}}^{q_2}(u))+O(s_n^2 \chi^m) \\
  & =4 \sum_{q_1, q_2 = 1}^p\Gamma_0^{q_1, q_2}(i/n, u)\sum_{l=(i-s_n)
    \vee 1}^{i-M-1} K^2_{\tau}(l/n - t) \Gamma_0^{r_1, r_2}((l-k)/n, u) \\ &\times \E ( \check  \epsilon_{i-k}^{r_1}(u) \check \epsilon_{i -k}^{r_2}(u)  \check 
 \epsilon_{l}^{q_1}(u) \check \epsilon_{l}^{q_2}(u)) + O(s_n^2\chi^m)\\
 & = 4\sum_{q_1, q_2 = 1}^p\Gamma_0^{q_1, q_2}(i/n, u) \sum_{l=(i-s_n)
    \vee 1}^{i-M-1} K^2_{\tau}(l/n - t) \Gamma_0^{r_1, r_2}((l-k)/n, u) \\&  \times \left\{\Gamma_0^{q_1, q_2}(l/n, u)\Gamma_0^{r_1, r_2}((i-k)/n, u) + O(\chi^m)+\mathrm{Cum}(  \check  \epsilon_{i-k}^{r_1}(u) \check \epsilon_{i -k}^{r_2}(u),  \check 
 \epsilon_{l}^{q_1}(u) \check \epsilon_{l}^{q_2}(u))\right\} \\& + O(s_n^2\chi^m)\\
 & = 4 K^2_{\tau}(i/n - t)s_n\sum_{q_1, q_2 = 1}^p \{\Gamma_0^{q_1, q_2} (i/n, u) \Gamma_0^{r_1, r_2}(i/n, u)\}^2  + O(M), \label{eq:case(1)}
\end{align}
where in the fourth equality $O(\chi^m)$ takes into consideration  the approximation error of $\Gamma_0^{q_1, q_2}(l/n, u)\Gamma_0^{r_1, r_2}((i-k)/n, u)$ to  $\E (\check  \epsilon_{i-k}^{r_1}(u) \check \epsilon_{i -k}^{r_2}(u))\E(   \check 
 \epsilon_{l}^{q_1}(u) \check \epsilon_{l}^{q_2}(u))$, and the last equality similar arguments as in \cref{lm:cumk} are used. \par 
For (2), $k-h > 2m$. In the subsequent analysis, we consider $\{(i-s_n)
    \vee 1 \}\leq l, l^{\prime} \leq i-M-1$, as used in \eqref{eq:rhoii}. Note that if $l-k < l^{\prime} -h -m$, $\check \epsilon_{l-k}^{r_1}(u)$ is independent of $$\{\check \epsilon_{i-k}^{q_1}(u),\check \epsilon_{i-h}^{q_2}(u), \check \epsilon_{l^{\prime}-h}^{r_2}(u), \check \epsilon_{l}^{q_1}(u), \check \epsilon_{l^{\prime}}^{q_2}(u)\}.$$
Therefore, we have 
\begin{align}
   & \E (  \check  \epsilon_{l-k}^{r_1}(u) \check \epsilon_{l^{\prime} -h}^{r_2}(u) \check  \epsilon_{i-k}^{r_1}(u) \check \epsilon_{i -h}^{r_2}(u)  \check 
 \epsilon_{l}^{q_1}(u) \check \epsilon_{l^{\prime}}^{q_2}(u))\\ &= \E (  \check  \epsilon_{l-k}^{r_1}(u))  \E ( \check \epsilon_{l^{\prime} -h}^{r_2}(u) \check  \epsilon_{i-k}^{r_1}(u) \check \epsilon_{i -h}^{r_2}(u)  \check 
 \epsilon_{l}^{q_1}(u) \check \epsilon_{l^{\prime}}^{q_2}(u))  = 0.
\end{align}
Similarly, if $l^{\prime} -h< l-k  -m$,
\begin{align}
  &  \E (  \check  \epsilon_{l-k}^{r_1}(u) \check \epsilon_{l^{\prime} -h}^{r_2}(u) \check  \epsilon_{i-k}^{r_1}(u) \check \epsilon_{i -h}^{r_2}(u)  \check 
 \epsilon_{l}^{q_1}(u) \check \epsilon_{l^{\prime}}^{q_2}(u))\\ & =   \E(\check \epsilon_{l^{\prime} -h}^{r_2}(u))  \E (  \check  \epsilon_{l-k}^{r_1}(u)  \check  \epsilon_{i-k}^{r_1}(u) \check \epsilon_{i -h}^{r_2}(u)  \check 
 \epsilon_{l}^{q_1}(u) \check \epsilon_{l^{\prime}}^{q_2}(u)) = 0.
\end{align}
We proceed to consider $k-h-m \leq l -l^{\prime}\leq k-h+m $.
If $l^{\prime} - h \geq i-k-m$, then we have $l  \geq l^{\prime}+k-h-m \geq i - 2m$, which is contradictory, since $l \leq i-M$, $M > 4m$. Therefore, we always have $l^{\prime} - h < i-k-m$. Hence, $\{\check \epsilon_{l-k}^{r_1}(u), \check \epsilon_{l^{\prime}-h}^{r_2}(u)\} $ is independent of $\{\check \epsilon^{r_1}_{i-k}(u),  \check \epsilon^{r_2}_{i-h}(u) \}$. 
Note that $l^{\prime} -m \geq l-k+h-2m> l-k$.  Together with $l^{\prime} - h\leq l - k + m\leq l - 2 m$, it follows that $\{\check  \epsilon_{l-k}^{r_1}(u), \check \epsilon_{l^{\prime} -h}^{r_2}(u)\}$ is independent of $\{\check  \epsilon_{i-k}^{r_1}(u), \check \epsilon_{i -h}^{r_2}(u),  \check 
 \epsilon_{l}^{q_1}(u), \check \epsilon_{l^{\prime}}^{q_2}(u)\}$, which leads to 
\begin{align}
    &\E (  \check  \epsilon_{l-k}^{r_1}(u) \check \epsilon_{l^{\prime} -h}^{r_2}(u) \check  \epsilon_{i-k}^{r_1}(u) \check \epsilon_{i -h}^{r_2}(u)  \check 
 \epsilon_{l}^{q_1}(u) \check \epsilon_{l^{\prime}}^{q_2}(u)) \\ &= \E (  \check  \epsilon_{l-k}^{r_1}(u) \check \epsilon_{l^{\prime} -h}^{r_2}(u))\E( \check  \epsilon_{i-k}^{r_1}(u) \check \epsilon_{i -h}^{r_2}(u)  \check 
 \epsilon_{l}^{q_1}(u) \check \epsilon_{l^{\prime}}^{q_2}(u)).
\end{align}
It follows that
\begin{align}
    &\mathbb{E}(\varrho_{i,k}^{r_1}(t,u)\varrho_{i,h}^{r_2}(t,u))\\ 
    & = 4\sum_{q_1, q_2 = 1}^p \Gamma_0^{q_1, q_2}(i/n, u)\sum_{l, l^{\prime}=(i-s_n)
    \vee 1}^{i-M-1}   K_{\tau}(l/n - t)K_{\tau}(l^{\prime}/n - t) \E (  \check  \epsilon_{l-k}^{r_1}(u) \check \epsilon_{l^{\prime} -h}^{r_2}(u))\\ &\times \E ( \check  \epsilon_{i-k}^{r_1}(u) \check \epsilon_{i -h}^{r_2}(u)  \check 
 \epsilon_{l}^{q_1}(u) \check \epsilon_{l^{\prime}}^{q_2}(u)) + O(s_n^2\chi^m)\\ 
 & = 4 \sum_{q_1, q_2 = 1}^p \Gamma_0^{q_1, q_2}(i/n, u)\sum_{l=(i-s_n)
    \vee 1}^{i-M-1}  K_{\tau}(l/n - t)K_{\tau}((l-k+h) /n - t) \E (  \check  \epsilon_{l-k}^{r_1}(u) \check \epsilon_{l-k}^{r_2}(u)) \\ &\times \E ( \check  \epsilon_{i-k}^{r_1}(u) \check \epsilon_{i -h}^{r_2}(u)  \check 
 \epsilon_{l}^{q_1}(u) \check \epsilon_{l-k+h}^{q_2}(u)) + O(s^2_n\chi^m)\\
  & = 4 \sum_{q_1, q_2 = 1}^p\Gamma_0^{q_1, q_2}(i/n, u)\sum_{l=(i-s_n)
    \vee 1}^{i-M-1}  K_{\tau}(l/n - t)K_{\tau}((l-k+h) /n - t) \Gamma_0^{r_1,r_2} ((l-k)/n, u) \\ &\times \E ( \check  \epsilon_{i-k}^{r_1}(u) \check \epsilon_{i -h}^{r_2}(u)  \check 
 \epsilon_{l}^{q_1}(u) \check \epsilon_{l-k+h}^{q_2}(u)) \left\{\mf 1( |l - i+h|  > m)+\mf 1(|l - i+h|  \leq m)\right\} +O(s^2_n\chi^m) \\ 
 & = 4 \sum_{q_1, q_2 = 1}^p\Gamma_0^{q_1, q_2}(i/n, u)\sum_{l=(i-s_n)
    \vee 1}^{i-M-1} K_{\tau}(l/n - t)K_{\tau}((l-k+h) /n - t) \Gamma_0^{r_1,r_2} ((l-k)/n, u) \\ &\times \E ( \check  \epsilon_{i-k}^{r_1}(u)  \check \epsilon_{l-k+h}^{q_2}(u)) \E(\check \epsilon_{i -h}^{r_2}(u)  \check 
 \epsilon_{l}^{q_1}(u)) \left\{\mf 1(|l - i+h|  \leq m)\right\} +O(s^2_n\chi^m) \\ 
 & = 4 \sum_{q_1, q_2 = 1}^p \Gamma_0^{q_1, q_2}(i/n, u)\sum_{l=(i-s_n)
    \vee 1}^{i-M-1}\mf 1 (l=i-h)  K_{\tau}(l/n - t)K_{\tau}((l-k+h) /n - t) \\ &\times \Gamma_0^{r_1,r_2} ((l-k)/n, u) 
 \Gamma_0^{r_1, q_2}((i-k)/n, u)\Gamma_0^{r_2, q_1}((i-h)/n, u) +O(s^2_n\chi^m)\\
 &=4  K^2_{\tau}(i/n - t)\mf 1 (h > M)\sum_{q_1, q_2 = 1}^p \{\Gamma_0^{r_1, q_2} (i/n, u) \Gamma_0^{r_2, q_1}(i/n, u)\}\times \{\Gamma_0^{q_1, q_2} (i/n, u) \Gamma_0^{r_1, r_2}(i/n, u)\} \\ & + O(s^2_n\chi^m + s_n/(n\tau)),\label{eq:case(2)}
\end{align}
 where in the fourth equality when $ |l - i+h|  > m$,  $\E ( \check  \epsilon_{i-k}^{r_1}(u) \check \epsilon_{i -h}^{r_2}(u)  \check 
 \epsilon_{l}^{q_1}(u) \check \epsilon_{l-k+h}^{q_2}(u))$ can be further written into $\E ( \check  \epsilon_{i-k}^{r_1}(u) \check \epsilon_{i -h}^{r_2}(u) \check \epsilon_{l-k+h}^{q_2}(u))\E( \check 
 \epsilon_{l}^{q_1}(u) )$ or $\E ( \check  \epsilon_{i-k}^{r_1}(u) \check 
 \epsilon_{l}^{q_1}(u) \check \epsilon_{l-k+h}^{q_2}(u))\E( \check \epsilon_{i -h}^{r_2}(u) )$, both of which are zero; when $ |l - i+h| \leq m$, we have $l-(i-k) > i-h-m - (i-k) > m$, $i-h - (l-k+h) = i-l + k - 2h >k-h-m > m$, since $k-h > 2m$.

For (3), we have $0<  k-h \leq 2m$. Note that
\begin{align}
    &\mathbb{E}(\varrho_{i,k}^{r_1}(t,u)\varrho_{i,h}^{r_2}(t,u))\\ 
     & =   O(s_n^2\chi^m) + 4 \sum_{q_1, q_2 = 1}^p \Gamma_0^{q_1, q_2}(i/n, u)\\ & \times \E(
    \epsilon^{r_1}_{i-k}(u)
    \sum_{l=(i-s_n)
    \vee 1}^{i-M-1} K_{\tau}(l/n - t) \check{\phi}^{r_1, q_1}_{l,k}(u) \epsilon^{r_2}_{i-h}(u) \sum_{l=(i-s_n)
    \vee 1}^{i-M-1} K_{\tau}(l/n - t) \check{ \phi}^{r_2, q_2}_{l,h}(u))\\
    &=  O(s_n^2\chi^m) + 4\sum_{q_1, q_2 = 1}^p \Gamma_0^{q_1, q_2}(i/n, u)\\ & \times 
    \sum_{l, l^{\prime}=(i-s_n)
    \vee 1}^{i-M-1} K_{\tau}(l/n - t) K_{\tau}(l^{\prime}/n - t) (\mathrm{Cum}(
    \epsilon^{r_1}_{i-k}(u)\epsilon^{r_2}_{i-h}(u), \check{\phi}^{r_1, q_1}_{l,k}(u),     \check{ \phi}^{r_2, q_2}_{ l^{\prime},h}(u))+O(\chi^m))\\
    &= O(1),
\end{align}
 where in the second equality $O(\chi^m)$ accounts for $\E(\epsilon^{r_1}_{i-k}(u)\epsilon^{r_2}_{i-h}(u)\check{\phi}^{r_1, q_1}_{l,k}(u))\E(\check{ \phi}^{r_2, q_2}_{ l^{\prime},h}(u))$,  the term $\E(\epsilon^{r_1}_{i-k}(u)\epsilon^{r_2}_{i-h}(u))
\E(\check{\phi}^{r_1, q_1}_{l,k}(u)\check{ \phi}^{r_2, q_2}_{ l^{\prime},h}(u))$, $\E(\epsilon^{r_1}_{i-k}(u)\epsilon^{r_2}_{i-h}(u)\check{ \phi}^{r_2, q_2}_{l,h}(u))\E(\check{\phi}^{r_1, q_1}_{ l^{\prime},k}(u))$ and  the term $\E(\epsilon^{r_1}_{i-k}(u)\epsilon^{r_2}_{i-h}(u))\E(\check{\phi}^{r_1, q_1}_{l,k}(u))\E(\check{ \phi}^{r_2, q_2}_{ l^{\prime},h}(u))$, and  similar arguments as in \cref{lm:cumsum} are used; in the last equality we use the fact that $m = \lfloor c \log n/4 \rfloor$ for some large constant $c$. 
Therefore, combining \eqref{eq:Glarge}, \eqref{eq:case(0)}, \eqref{eq:caseij}, \eqref{eq:case(1)} and \eqref{eq:case(2)}, we have 
\begin{align}
     &(n\tau s_n)^{2} \sum_{k,h=3m+1}^{s_n}  \left\{\mathbb{E}(\check G_k^{r_1, r_1}(t,u)\check G_h^{r_2, r_2}(t,u))\right\}\\& = \sum_{k,h=3m+1}^{s_n} \sum_{i,j=1}^n K_{\tau}(i/n -t)K_{\tau}(j/n -t)  \mathbb{E}(\varrho_{i,k}^{r_1}(t,u)\varrho_{j,l}^{r_2}(t,u))\\
      & = 4s_n^2 \sum_{i=1}^n K^2_{\tau}(i/n - t)\sum_{q_1, q_2 = 1}^p [\{\Gamma_0^{r_1, q_2} (i/n, u) \Gamma_0^{r_2, q_1}(i/n, u)\}\times \{\Gamma_0^{q_1, q_2} (i/n, u) \Gamma_0^{r_1, r_2}(i/n, u)\} \\
      &+\{\Gamma_0^{r_1, r_2} (i/n, u) \Gamma_0^{q_1, q_2}(i/n, u)\}^2]+ O(n\tau M s_n +  (n\tau) s_n^4\chi^m  + s^3_n)\\ 
     & = 4s_n^2 n\tau \phi^2 \sum_{q_1, q_2 = 1}^p [\{\Gamma_0^{r_1, q_2} (t, u) \Gamma_0^{r_2, q_1}(t, u)\}\times \{\Gamma_0^{q_1, q_2} (i/n, u) \Gamma_0^{r_1, r_2}(i/n, u)\}\\
      &+\{\Gamma_0^{r_1, r_2} (t, u) \Gamma_0^{q_1, q_2}(t, u)\}^2]+ o(n\tau s_n^2).
\end{align}

\end{proof}

\subsection{A Gaussian approximation with moment assumptions and increasing
dependence}
\label{eq:GAsn}
In this section, we extend Theorem 2.1 of \cite{zhang2018} to  allow non-stationarity and increasing dependence in the time series. 

Recall that $x_i=\mathcal{G}_i\left(\mathcal{F}_i\right)$, where $\mathcal{F}_i=\left(\ldots, \epsilon_{i-1}, \epsilon_i\right)$ and $\mathcal{G}_i=\left(\mathcal{G}_{i 1}, \ldots, \mathcal{G}_{i p}\right)^{\prime}:[0,1]\times \mathcal S^{\mathbb Z}$, and $y_i = (y_{i1},\cdots, y_{ip})^{\top}$ is a Gaussian process which is independent of $x_i$ and preserve the auto-covariance structure of $x_i$. Let $\left\{\epsilon_i^{\prime}\right\}$ be an i.i.d. copy of $\left\{\epsilon_i\right\}$. For $\left\|x_{i j}\right\|_q<\infty$, define the dependence measure as 
$$
\theta_{k, j, q}=\sup _{i \in \mathbb Z}\left\|\mathcal{G}_{i j}\left(\mathcal{F}_i\right)-\mathcal{G}_{i j}\left( \mathcal{F}_{i, i-k}\right)\right\|_q, \quad \Theta_{k, j, q}=\sum_{l=k}^{+\infty} \theta_{l, j, q},
$$
where $\mathcal{F}_{i, k}=\big (\ldots, \epsilon_{k-1}, \epsilon_k^{\prime}, \epsilon_{k+1}, \ldots, \epsilon_{i-1}, \epsilon_i\big )$ is a coupled version of $\mathcal{F}_i$. 
\begin{assumption}\label{ass:GA} We need the following assumptions:
\begin{enumerate}
    \item \label{A:moment} (Moment condition) For a large enough $s > 0$, there exists a $C_n$ such that $$
    \max_{1 \leq i \leq n} \max_{1 \leq j \leq p} \|x_{ij}\|_s \leq C_n.
$$ 
\item (Dependence condition) For a  large positive constant $c$, $\theta \in [0,1/11)$, 
\begin{align}
&\sum_{j=c\lfloor n^\theta \rfloor+1}^{\infty} \max_{1 \leq k \leq p} j\theta_{j,k,2}<c_{2,n}, \quad \sum_{j=1}^{\infty} \max_{1 \leq k \leq p} j\theta_{j,k,3} < c_{3,n},
\end{align}
where $c_{2,n} \to 0$, $c_{3,n} \to \infty$, $c_{3,n}/n^\theta \to 0$, as $n \to \infty$. \label{A:depend}
\item (Rate condition of dimensions) $(n p/\gamma)^{1/s}C_n \leq c_{0} n^{(3-9\theta)/8}$, where $\gamma \asymp n^{-(1-11\theta)/8}$, $c_0 > 0$. \label{A:rate}
\item (Variance condition) For $\sigma_{j, k}:=\operatorname{cov}\left(X_j, X_k\right)=\sum_{i, l=1}^n \operatorname{cov}\left(x_{i j}, x_{l k}\right) / n$, there exist positive constants $c_1$ and $c_2$, 
$$
c_1<\min _{1 \leq j \leq p} \sigma_{j, j} \leq \max _{1 \leq j \leq p} \sigma_{j, j}<c_2.
$$ \label{A:var}
\end{enumerate}
\end{assumption}
Condition \ref{A:moment} only requires finite moment condition instead of the sub-Gaussian condition in Assumption 2.1 of  \cite{zhang2018}. Condition \ref{A:depend} relaxes (10) of Assumption 2.3 of \cite{zhang2018},  which assumes there exists $c_3 > 0$ such that $\sum_{j=1}^{\infty} \max_{1 \leq k \leq p} j\theta_{j,k,3} < c_{3}$. Condition \ref{A:rate} gives the rate condition under finite moment condition, where the dimension can grow at a polynomial rate. Condition \ref{A:var} is the same as (9) of Assumption 2.3 of \cite{zhang2018}.
\begin{proposition}\label{prop:GAhigh}
    Under \cref{ass:GA}, we have 
    \begin{align}
   &\rho_n:= \sup_{t \in \mathbb R}\left|P\left(\max_{1 \leq j \leq p}\frac{1}{\sqrt{n}} \sum_{i=1}^n x_{ij} \leq t \right) - P\left(\max_{1 \leq j \leq p}\frac{1}{\sqrt{n}} \sum_{i=1}^n y_{ij} \leq t\right)\right| \\ & = O\left\{n^{-(1-11\theta)/8}+\left(n^{(1 - 7\theta)q/\left\{8(q+1)\right\}} \right)\left(\sum_{j=1}^p \Theta_{c\lfloor n^\theta \rfloor, j, q}^q\right)^{1 /(1+q)} +c_{2,n}^{1 / 3}\left(1 \vee \log \left(p / c_{2,n}\right)\right)^{2 / 3}\right\}. \label{eq:increase}
\end{align}
\end{proposition}


\begin{proof}
The key of the proof is to investigate the error in the approximation of $M_0$-dependent sequences under increasing dependence.  We follow the steps in the proof of Theorem 2.1 of \cite{zhang2018}.
\par 
    Step 1: Construct the $M_0$-dependent sequence as
$$
x_i^{(M_0)}=\mathbb{E}\left[\mathcal{G}_i\left( \ldots, \epsilon_{i-1}, \epsilon_i\right) \mid \epsilon_{i-M_0}, \epsilon_{i-M_0+1}, \ldots, \epsilon_i\right]. 
$$
Let $M_0=c\lfloor n^\theta \rfloor$, $\gamma = n^{-(1-11\theta)/8}$, $\theta \in [0, 1/11)$,
Define the $\psi, \beta, M_x, M_y, u_n$ as
\begin{align}
    &\psi \asymp n^{(1-7\theta)/8}, \quad u_n  = M_x=M_y \asymp n^{(3-9\theta)/8},\quad \beta \asymp n^{(1+\theta)/8}.\label{eq:phirate}
\end{align} 
Under \eqref{eq:phirate}, we first show that the conditions of Proposition A.1 of \cite{zhang2018} can be satisfied. Let $M_{x y} = \max\{M_x, M_y\}$. By construction, $2 \sqrt{5} \beta(6 M+1) M_{x y} / \sqrt{n} \leq 1$.  After a careful investigation of Lemma A.2 of \cite{zhang2018}, we have
\begin{align}
    \max _{1 \leq j \leq p} n^{-1} \sum_{i, k=1}^n\left|\mathbb{E} x_{i j}^{(M_0)} x_{k j}^{(M_0)}-\mathbb{E} x_{i j} x_{k j}\right| \rightarrow 0
\end{align}
since  $\sum_{j=c\lfloor n^\theta \rfloor+1}^{\infty} \max_{1 \leq k \leq p} j\theta_{j,k,2} \to 0$, which leads to the result that $c_1 / 2<\min _{1 \leq j \leq p} \sigma_{j, j}^{(M_0)} \leq \max _{1 \leq j \leq p} \sigma_{j, j}^{(M_0)}<2 c_2$ under Assumption \ref{ass:GA}, where $\sigma_{j, k}^{(M_0)}=\sum_{i,l=1}^n\operatorname{cov}\left(x_{ij}^{(M_0)}, x_{lk}^{(M_0)}\right)/n$.

Recall that $u_n  = M_x=M_y \asymp n^{(3-9\theta)/8}$. We proceed to show that,
$$
P\left(\max _{1 \leq i \leq n} \max _{1 \leq j \leq p}\left|x_{i j}^{(M_0)}\right| \leq u_n\right) \geq 1-\gamma, \quad P\left(\max _{1 \leq i \leq n} \max _{1 \leq j \leq p}\left|y_{i j}^{(M_0)}\right| \leq u_n\right) \geq 1-\gamma.
$$
For the first term, we have
\begin{align}
    P\left(\max _{1 \leq i \leq n} \max _{1 \leq j \leq p}\left|x_{i j}^{(M_0)}\right| > u_n\right) &\leq \sum_{i=1}^n \sum_{j=1}^p  P\left(\left|x_{i j}^{(M_0)}\right| > u_n\right) \\  &\leq (np)\max _{1 \leq i \leq n} \max _{1 \leq j \leq p} \E |x_{ij}^{(M_0)}|^s/ u_n^s \\ &\leq (np)\max _{1 \leq i \leq n} \max _{1 \leq j \leq p} \E |x_{ij}|^s/ u_n^s \\ 
    & \leq (np)C_n^s/ u_n^s \leq \gamma,
\end{align}
where the third inequality follows from the convexity of $|\cdot|^s$, $s \geq 2$.
For the second term, since $y_{ij}$ are Gaussian random variables, we have 
\begin{align}
    &P\left(\max _{1 \leq i \leq n} \max _{1 \leq j \leq p}\left|y_{i j}^{(M_0)}\right| >  u_n\right)\\ 
    & \leq (np) \max _{1 \leq i \leq n} \max _{1 \leq j \leq p} \E \exp(|y_{ij}|)/\exp(u_n) \leq \gamma,
\end{align}
where $u_n   \asymp n^{(3-9\theta)/8}$.  Denote $\tilde x_{ij}^{(M_0)}$ by $$(x^{(M_0)}_{ij} \wedge M_x) \vee (-M_x) - \E \left\{ (x_{ij}^{(M_0)} \wedge M_x) \vee (-M_x)\right\},$$ and $\tilde y^{(M_0)}_{ij}$ by $(y_{ij}^{(M_0)} \wedge M_y) \vee (-M_y)$. Write $\breve x_{ij}^{(M_0)} = x_{ij}^{(M_0)} - \tilde x^{(M_0)}_{ij}$. 
Let $\phi^{(M_0)}(M_x)$ be the upper bound of the truncation, i.e.,
\begin{align}
    \max_{1 \leq j, k \leq p}\frac{1}{n} \sum_{i=1}^n |\sum_{|l - i| \leq M_0}\left\{\E(x_{ij}^{(M_0)}x_{lk}^{(M_0)})-\E(\tilde x_{ij}^{(M_0)}\tilde x_{lk}^{(M_0)}) \right\}|  \leq \phi^{(M_0)}(M_x).\label{eq:defphi}
\end{align}
$\phi^{(M_0)}(M_y)$ can be defined analogously. Let $\phi^{(M_0)}\left(M_x, M_y\right)=\phi^{(M_0)}\left(M_x\right) +  \phi^{(M_0)}\left(M_y\right)$. Suppose $n = (N+M_0)r$, $N \geq M_0$, and $N, M_0, r \to +\infty$ as $n$ goes to infinity. Define the block sums as 
\begin{align}
    A_{i j}^{(M_0)}=\sum_{l=i N+(i-1) M_0-N+1}^{i N+(i-1) M_0} x_{l j}^{(M_0)}, \quad B^{(M_0)}_{i j}=\sum_{l=i(N+M_0)-M_0+1}^{i(N+M_0)} x^{(M_0)}_{l j},\label{eq:A}
\end{align}
\begin{align}
    \widetilde A^{(M_0)}_{i j}=\sum_{l=i N+(i-1) M_0-N+1}^{i N+(i-1) M_0} \tilde x^{(M_0)}_{l j}, \quad  \widetilde B^{(M_0)}_{i j}=\sum_{l=i(N+M_0)-M_0+1}^{i(N+M_0)} \tilde x^{(M_0)}_{l j},\label{eq:Atilde}
\end{align}
and the truncated versions as 
\begin{align}
    \breve A_{i j}=\sum_{l=i N+(i-1) M_0-N+1}^{i N+(i-1) M_0} (x_{l j}^{(M_0)} \wedge M_x) \vee (-M_x),\\ \breve B_{i j}=\sum_{l=i(N+M_0)-M_0+1}^{i(N+M_0)} (x_{l j}^{(M_0)} \wedge M_x) \vee (-M_x).\label{eq:Abreve}
\end{align}
Let $\varphi^{(M_0)}(M_x)$ be the smallest finite constant such that uniformly for $i$ and $j$ 
\begin{align}
   \max\{ \E (A_{ij}^{(M_0)} -   \breve A^{(M_0)}_{i j})^2 /N,   \E (B^{(M_0)}_{ij} -   \breve B^{(M_0)}_{i j})^2 /M_0\} \leq (\varphi^{(M_0)}(M_x))^2.\label{eq:defvarphi}
\end{align}
Similarly, $\varphi^{(M_0)}(M_y)$ can be defined and $ \varphi^{(M_0)}\left(M_x, M_y\right) =\varphi^{(M_0)}\left(M_x\right) +  \varphi^{(M_0)}\left(M_y\right)$.
Following (41) of the proof of Theorem 2.1  of \cite{zhang2018}, we have 
\begin{align}
\rho_n  &=O\left\{ \left(\psi^2+\psi \beta\right) \phi^{(M_0)}\left(M_x, M_y\right)+\left(\psi^3+\psi^2 \beta+\psi \beta^2\right) \frac{(2 M_0+1)^2}{\sqrt{n}}\left(\bar{m}_{x, 3}^3+\bar{m}_{y, 3}^3\right)\right. \\
& +\psi \varphi^{(M_0)}\left(M_x, M_y\right) \sqrt{\log (p / \gamma)}+\gamma+\left(\beta^{-1} \log (p)+\psi^{-1}\right) \sqrt{1 \vee \log (p \psi)} \\
& \left.+\psi^{q /(1+q)}\left(\sum_{j=1}^p \Theta_{M_0, j, q}^q\right)^{1 /(1+q)}\right\},\label{eq:rhon}
\end{align}
where $\bar m_{x,3}^3 = \max_{1 \leq j \leq p}(\sum_{i=1}^n \E |x_{ij}|^3/n)^{1/3}$, and $\bar m_{y,3}^3 = \max_{1 \leq j \leq p}(\sum_{i=1}^n \E|y_{ij}|^3/n)^{1/3}$.

Step 2: quantify $\phi^{(M_0)}(M_x)$ and $\psi^{(M_0)}(M_x)$ as defined in \eqref{eq:defphi} and \eqref{eq:defvarphi}. Following the calculation in Step 2 of  Theorem 2.1 of \cite{zhang2018}, we have
\begin{align}
& \max _{1 \leq j, k \leq p} \frac{1}{n} \sum_{i=1}^n\left|\sum_{l=(i-M_0) \vee 1}^{(i+M_0) \wedge n}\left(\mathbb{E} x_{i j}^{(M_0)} x_{l k}^{(M_0)}-\mathbb{E} \tilde{x}_{i j}^{(M_0)} \tilde{x}_{l k}^{(M_0)}\right)\right| \\
& \leq \max _{1 \leq j, k \leq p} \frac{1}{n} \sum_{i=1}^n\left|\sum_{l=(i-M_0) \vee 1}^{(i+M_0) \wedge n}\left(\mathbb{E} x_{i j}^{(M_0)} \breve{x}_{l k}^{(M_0)}+\mathbb{E} \tilde{x}_{i j}^{(M_0)} \tilde{x}_{l k}^{(M_0)}\right)\right| \\
& \leq C \max _{1 \leq k \leq p} \sum_{l=1}^{M_0} \sum_{j=l}^{M_0} \theta_{j, k, 3} / M_x \\ 
& = C c_{3,n} / M_x,
\end{align}
where $C$ is a constant that is large enough. Therefore, we can take $\phi^{(M_0)}(M_x)$ as $C c_{3,n} / M_x$.
Based on $(x_i^{(M_0)})$, we can define the quantities $A_{ij}^{(M_0)}$, $B_{ij}^{(M_0)}$, $\widetilde A_{ij}^{(M_0)}$,  $\widetilde B_{ij}^{(M_0)}$, $\breve A_{ij}^{(M_0)}$,  $\breve B_{ij}^{(M_0)}$ as in \eqref{eq:A}, \eqref{eq:Atilde} and \eqref{eq:Abreve}. Define 
$$
\chi_{ij}^{(M_0)} = (x_{ij}^{(M_0)} \wedge M_x) \vee (-M_x).
$$ After a careful investigation of (42) in the proof in Step 2 of  Theorem 2.1  of \cite{zhang2018}, we have 
\begin{align}
 \mathbb{E}\left(A_{i j}^{(M_0)}-\widetilde{A}_{i j}^{(M_0)}\right)^2 &= \sum_{l=(i-M_0) \vee 1}^{(i+M_0) \wedge n}\sum_{k=(i-M_0) \vee 1}^{(i+M_0) \wedge n}\mathbb{E}(\breve x_{l j}^{(M_0)}\breve x_{k j}^{(M_0)})\\
    & \leq C_0( \max _{1 \leq k \leq p} \sum_{l=1}^{N} \sum_{j=l}^{M_0} \theta_{j, k, 3} / M_x^{5/3} + 1/M_x).\label{eq:Atildebd}
\end{align}
For a large positive constant $C_1$, using \eqref{eq:Atildebd}, we have
\begin{align}
& \mathbb{E}\left(A_{i j}^{(M_0)}-\breve{A}_{i j}^{(M_0)}\right)^2 / N  \\
& =  \mathbb{E}\left(A_{i j}^{(M_0)}-\widetilde{A}_{i j}^{(M_0)}-\E\breve{A}_{i j}^{(M_0)}\right)^2 / N \\
& =  \mathbb{E}\left(A_{i j}^{(M_0)}-\widetilde{A}_{i j}^{(M_0)}\right)^2/ N-2 \E (A_{i j}^{(M_0)} - \widetilde{A}_{i j}^{(M_0)})\E\breve{A}_{i j}^{(M_0)} / N+(\E\breve{A}_{i j}^{(M_0)})^2  / N \\
& =  \mathbb{E}\left(A_{i j}^{(M_0)}-\widetilde{A}_{i j}^{(M_0)}\right)^2/ N-2 \E (A_{i j}^{(M_0)})\E\breve{A}_{i j}^{(M_0)} / N+(\E\breve{A}_{i j}^{(M_0)})^2  / N \\
&\leq \mathbb{E}\left(A_{i j}^{(M_0)}-\widetilde{A}_{i j}^{(M_0)}\right)^2 / N+\left\{\mathbb{E} (A_{ij} ^{(M_0)}-\breve{A}_{i j}^{(M_0)})\right\}^2 / N\\ 
&\leq C_1 c_{3,n} / M_x^{5 / 3}+\left\{\sum_{l=i N+(i-1) M-N+1}^{i N+(i-1) M} \mathbb{E}\left(\chi_{l j}^{(M_0)}-x_{l j}^{(M_0)}\right) \mathbf{I}\left\{\left|x_{l j}^{(M_0)}\right|>M_x\right\}\right\}^2 / N \\
 &\leq C_1 c_{3,n} / M_x^{5 / 3}+ C_1 N (\max_{ij} E|x_{ij}|^4/M_x^3)^2.
\end{align}
Similarly, for a large positive constant $C_2$,  we have $$\mathbb{E}\left(B_{i j}^{(M_0)}-\breve{B}_{i j}^{(M_0)}\right)^2 / M \leq C_2 c_{3,n} / M_x^{5 / 3} +C_2 M_0 (\max_{ij} E|x_{ij}|^4/M_x^3)^2.$$
Therefore, we can take 
\begin{align}
    \psi^{(M_0)}(M_x) = C^{\prime} c_{3,n}^{1/2} / M_x^{5 / 6} + C^{\prime} \sqrt{N}/ M_x^3,
\end{align}
where $C^{\prime}$ is a large enough constant.\par
Step 3: Quantify $\phi^{(M_0)}(M_y)$ and $\varphi^{(M_0)}(M_y)$.  Following the arguments in Step 2, we have 
\begin{align}
    \phi^{(M_0)}(M_y) \leq C_3 c_{3,n}/M_y^2,\quad  \varphi^{(M_0)}(M_y) \leq C_4 c_{3,n}/M_y^2.
\end{align}
\par Step 4:
Under \eqref{eq:phirate}, we have 
\begin{align}
    &\left(\psi^3+\psi^2 \beta+\psi \beta^2\right) \frac{(2 M_0+1)^2}{\sqrt{n}} =O (n^{-(1-11\theta) / 8}), \label{eq:47}\\ 
    & \left(\beta^{-1} \log (p)+\psi^{-1}\right) \sqrt{1 \vee \log (p \psi)} =O (n^{-(1-11\theta) / 8} ).\label{eq:49}
\end{align}
Under the conditions $c_{3,n}/M_0 \to 0$ and \eqref{eq:phirate},  we have 
\begin{align}
   & \left(\psi^2+\psi \beta\right) \phi^{(M_0)}\left(M_x, M_y\right) =O(c_{3,n}n^{(9\theta-3)/8} n^{(2-6\theta)/8}) =O( n^{-(1-11\theta) / 8}),\label{eq:46}\\
  & \psi \varphi^{(M_0)}\left(M_x, M_y\right)  \sqrt{\log (p / \gamma)} = O\left(\frac{\psi c_{3,n}^{1/2}}{u_n^{5 / 6}}+\frac{\psi \sqrt{N}}{u_n^3}\right) = O( n^{-(1-11\theta) / 8}).\label{eq:48}
\end{align}
Under the \cref{ass:GA}, we have $\bar m_{x,3}^3 + \bar m_{y,3}^3  < \infty$. Combining \eqref{eq:rhon}, \eqref{eq:46}, \eqref{eq:47}, \eqref{eq:48} and \eqref{eq:49}, for $q \geq 2$, we have
\begin{align}
& \sup_{t \in \mathbb R}\left|P\left(\max_{1 \leq j \leq p}\frac{1}{\sqrt{n}} \sum_{i=1}^n x_{ij} \leq t \right) - P\left(\max_{1 \leq j \leq p}\frac{1}{\sqrt{n}} \sum_{i=1}^n y_{ij}^{(M_0)} \leq t\right)\right| \\
& = O\left\{n^{-(1-11\theta)/8}+\left(n^{(1 - 7\theta)q/\left\{8(q+1)\right\}} \right)\left(\sum_{j=1}^p \Theta_{M_0, j, q}^q\right)^{1 /(1+q)}\right\}.
\end{align}
Step 5: 
By Step 5 of the proof of Theorem 2.1 of \cite{zhang2018} and Theorem 2 of \cite{Chernozhukov2013}, we have 
\begin{align}
&\sup _{t \in \mathbb{R}}\left|P\left(\max_{1 \leq j \leq p}\frac{1}{\sqrt{n}} \sum_{i=1}^n y_{ij}^{(M_0)} \leq t\right)-P\left(\max_{1 \leq j \leq p}\frac{1}{\sqrt{n}} \sum_{i=1}^n y_{ij} \leq t\right)\right| \\ &=O\left\{c_{2,n}^{1 / 3}\left(1 \vee \log \left(p / c_{2,n}\right)\right)^{2 / 3}\right\},
\end{align}
where $\max_{1 \leq j \leq p}\sum_{i=M_0+1}^{\infty} l \theta_{l,j,2}< c_{2,n}$.
\end{proof}
\subsection{Proof of  Theorem \ref{thm:ga}}

Recall that \begin{align}
  \mf U_i(t, u) &= \sum_{k=1}^{s_n}\bs \epsilon_{i-k}(u) \bs \epsilon^{\top}_i(u)\left(\sum_{j=1}^{i} \bs \epsilon_{j-k}(u)\bs \epsilon^{\top}_j(u)K_{\tau}(j/n - t)\mf 1(j \in \mathbb L_i) \right)^{\top}\\
     & + \sum_{k=1}^{s_n}  \left(\sum_{j=1}^{i} \bs \epsilon_{j-k}(u)\bs \epsilon^{\top}_j(u)K_{\tau}(j/n - t)\mf 1(j \in \mathbb L_i)\right)(\bs \epsilon_{i-k}(u) \bs \epsilon^{\top}_i(u))^{\top},
 \end{align}
and for $1 \leq i \leq 2\nt$, 
\begin{align}
      {\mf V}_{i,l}:= K_{\tau}(i/n-l/n) (\mathrm{tr}(\mf U_{i}(l/n, j/N)),  1 \leq j \leq N)^{\T}, \quad \mf V_i = ({\mf V}_{i,\nt} ,\cdots, {\mf V}_{i + n - 2\nt, n-\nt} )^{\top}.
\end{align}  

Note that 
\begin{align}
    \max_{\nt \leq l \leq n - \nt, 1 \leq j \leq N}\left|\sum_{i=1}^n K_{\tau}(i/n-t) \mathrm{tr}(\mf U_i(l/n, j/N))\right| =  \left|\sum_{i=1}^{2\nt} \mf V_i\right|_{\infty}.\label{eq:gasum}
\end{align}
Write
\begin{align}
    \mf W(t, u, i/n, \F_i) = \mf U_i(t, u)/s_n, \label{def:W}
\end{align}
where $\mf W(\cdot, \cdot, \cdot, \cdot)$ is a nonlinear filter from $[0,1]\times [0,1] \times [0,1] \times \mathbb R^{\mathbb Z}$ to $\mathbb R^{p\times p}$ such that $\mf U_i(\cdot, \cdot)$ is well-defined.  Define the element-wise dependence measure as $$\theta_{l,s} = \sup_{t , u\in[0,1], 1 \leq i \leq n}\left\|\left|  \sum_{r=1}^p \left\{ W_{r,r} (t,u, i/n, \FF_{i}) -   W_{r,r}(t,u, i/n, \FF_{i}^{(i-l)}) \right\}\right| \right\|_s, $$
where $\FF^{(j)}_i$ denotes changing $\eta_j$ in $\FF_i$ by its $i.i.d.$ copy $\eta_j^{\prime}$, $W_{r,r} (t,u, i/n, \FF_{i})$ is the element on the $r$-th column and $r$-th row of $\mf W(t, u, i/n, \F_i)$.   Recall that  in\eqref{eq:dependence}, 
\begin{align}
     \Theta_{m,s} := \sum_{l=m}^{\infty} \theta_{l,s},\quad \Xi_{m,s}:= \sum_{l=m}^{\infty} l\theta_{l,s}. 
\end{align}

By \cref{lm:delta1},
we have for fixed $h \geq 2$, $m > 4s_n$,
\begin{align}
     \Theta_{m,h} < C  \chi_1^{m},\quad \chi_1 \in (0,1), \quad \Xi_{m,3} < c_3s_n^2, \label{eq:tailtheta}
\end{align}
for a positive constant $C$ and some constant $c_3 > 0$. 
Denote $V_{i,l,j}$ as the $j$-th element of $\mf V_{i,l}$. By definition of \eqref{def:W}, we have $V_{i,l,j} = s_n K_{\tau}(i/n-l/n) \mathrm{tr}(\mf W(l/n, j/N, i/n, \FF_i))$. Define 
\begin{align}
    &V_{i,l,j}^* = s_n K_{\tau}(i/n-l/n) \mathrm{tr}(\mf W(l/n, j/N, i/n, \FF_i^{(0)})).
\end{align}
Since $K(\cdot)$ is bounded, for $h \geq 2$, $m > 4s_n$, we thus have 
\begin{align}
   &\max_{1 \leq l \leq n, 1 \leq j \leq N} \sum_{i=m}^{\infty}\|V_{i,l,j} - V_{i,l,j}^*\|_h/s_n \leq C \chi_1^{m}, \label{eq:ThetasumV}\\
    &\max_{1 \leq l \leq n, 1 \leq j \leq N} \sum_{i=1}^{\infty}i\|V_{i,l,j} - V_{i,l,j}^*\|_3/s_n  < c_3s_n^2, \label{eq:ass2.3-2}
    \\
    &\max_{1 \leq l \leq n, 1 \leq j \leq N} \sum_{i=m}^{\infty}i\|V_{i,l,j} - V_{i,l,j}^*\|_h/s_n  < Cm\chi_1^{m}. \label{eq:XisumV}
\end{align}

Therefore, combining \eqref{eq:ass2.3-2} and \eqref{eq:XisumV}, with sample size $2\nt$, the dependence condition of \cref{ass:GA} holds with 
\begin{align}
    c_{2,n} = C (n\tau)^{\theta} \chi_1^{(n\tau)^{\theta}} \to 0,\quad c_{3,n} = c_3 s_n^2,\label{eq:c2nc3n}
\end{align}
and $c_{3,n} \to \infty$,  $c_{3,n} = o((n\tau)^{\theta})$, $\theta \in [0,1/11)$.
Define
\begin{align}
    \tilde \sigma_{j,j} = \frac{1}{(n\tau) s^2_n} \sum_{k,l=1}^{2\nt}\mathrm{Cov} (V_{k,j}, V_{l,j}),
\end{align}
where $V_{k,j}$ is the $j$-th element of $\mf V_k$. We shall show that there exists positive constants $c_1$ and $c_2$ such that 
\begin{align}
    c_1 \leq \min_{1 \leq j \leq  N(n-2\nt+1)} \tilde \sigma_{j,j} \leq \max_{1 \leq j \leq N(n-2\nt+1)} \tilde \sigma_{j,j} \leq c_2.\label{eq:ass2.3-1}
\end{align}
Recall that $G_k^{r_1, q_1}(t,u)$ is the $(r_1, q_1)$th element of $\mf G_k(t,u)$. 
It suffices to compute uniform upper and lower bounds of 
\begin{align}
    \tilde \sigma (t,u) = (n\tau) \sum_{k,l=1}^{s_n}  \sum_{r_1, r_2 = 1}^p \mathrm{Cov}\left(G^{r_1,r_1}_{k}(t, u), G_{l}^{r_2,r_2}(t, u)\right).
\end{align}
By \cref{lm:cum}, we have 
\begin{align}
  \left| \tilde \sigma(t,u) -  \phi^2  \mathrm{tr}^2\{\bs \Gamma_0^2(t,u)\}-  \phi^2  \mathrm{tr}\{\bs \Gamma^4_0(t,u)\}\right|  = o(1), \label{eq:varlimt}
\end{align}
where $\phi^2 = \int K^2(t) dt$, $\mathrm{tr}^2\{\bs \Gamma_0^2(t,u)\}>0$, $\mathrm{tr}\{\bs \Gamma_0^4(t,u)\}>0$.
Therefore, \eqref{eq:ass2.3-1} is proved. The variance condition of \cref{ass:GA} is satisfied. \par 
Note that $\|\phi_{i,k}(j/N)\|_{2s}= O(1)$ by Cauchy inequality and \cref{lm:delta}. Then, along with \eqref{eq:psimoment}, we obtain
\begin{align}
  \max_{\substack{\nt \leq l \leq n - \nt\\1 \leq i \leq 2\nt, 1 \leq j \leq N}}  \|V_{i,l,j} \|_s/s_n \leq  &\max_{\substack{\nt \leq l \leq n - \nt\\1 \leq i \leq 2\nt, 1 \leq j \leq N}}s_n^{-1}( 2s_n\|\phi_{i,k}(j/N)\|_{2s}   \|\psi_{i,k}(l/n, j/N)\|_{2s} ) \\&= O(\sqrt{s_n}).\label{eq:Vmoment}
\end{align}
Therefore, the moment condition of \cref{ass:GA} holds with $C_n = C \sqrt{s_n}$, where $C$ is a positive constant that is large enough.

For $0\leq \theta  < 1/11$, 
by \cref{prop:GAhigh}, under the condition $s_n^{2} = o((n\tau)^{\theta})$, for a large $s$ such that $\left\{(n\tau)^{(9-11\theta)/8} nN\right\}^{1/s} \sqrt{s_n} = O((n\tau)^{(3-9\theta)/8})$, and  we have
\begin{align}
    &\sup_{x \in \mathbb R}\left| \mathbb{P}\left( \frac{1}{s_n\sqrt{n\tau}}\underset{\substack{\nt \leq i \leq n - \nt\\ 1 \leq j \leq N }}{\max}\left|\sum_{i=1}^n K_{\tau}(i/n-l/n) \mathrm{tr}[\mf U_i(l/n, j/N) - \E\{\mf U_i(l/n, j/N)\}] \right|\leq x\right) \right.\\&\left. - \mathbb{P}\left(\frac{1}{s_n \sqrt{n \tau} }\left|\sum_{i=1}^{2\nt}\mf Z_i \right|_{\infty} \leq x\right) \right| \\ 
    & =\sup_{x \in \mathbb R}\left| \mathbb{P}\left(\underset{\substack{\nt \leq i \leq n - \nt\\ 1 \leq j \leq N  }}{\max}\frac{1}{s_n\sqrt{n \tau}} \left|\sum_{i=1}^{2\nt} (\mf V_i - \E \mf V_i)\right|_{\infty} \leq x\right)  - \mathbb{P}\left(\frac{1}{s_n \sqrt{n\tau} }\left|\sum_{i=1}^{2\nt}\mf Z_i \right|_{\infty} \leq x\right) \right| \\
    &=O\left((n\tau)^{-(1-11\theta)/8}+\left((n\tau)^{(1 - 7\theta)h/\left\{8(h+1)\right\}} \right)\left(\sum_{j=1}^p \Theta_{c\lfloor (n\tau)^\theta \rfloor, j, h}^h\right)^{1 /(1+h)} \right. \\& \left.+c_{2,n}^{1 / 3}\left(1 \vee \log \left(nN / c_{2,n}\right)\right)^{2 / 3}\right)=o(1), \label{eq:GA}
\end{align}
where the last equality follows from  \eqref{eq:c2nc3n}.
By \cref{lm:4th} and triangle inequality, we have 
\begin{align}
   &\left| \underset{\substack{\nt \leq i \leq n - \nt\\ 1 \leq j \leq N }}{\max}\left|\sum_{i=1}^n K_{\tau}(i/n-l/n) \mathrm{tr}[\mf U_i(l/n, j/N) - \E\{\mf U_i(l/n, j/N)\}] \right| \right.\\ & \left.- \underset{\substack{\nt \leq i \leq n - \nt\\ 1 \leq j \leq N }}{\max}\left|\sum_{i=1}^n K_{\tau}(i/n-l/n) \mathrm{tr}\{\mf U_i(l/n, j/N)\} \right|\right| \\ &\leq \underset{\substack{\nt \leq i \leq n - \nt\\ 1 \leq j \leq N }}{\max}\left|\sum_{i=1}^n K_{\tau}(i/n-l/n) \mathrm{tr}[\E\{\mf U_i(l/n, j/N)\}] \right| = O(\tau s^{5/2}_n) \label{eq:Etru}
\end{align}
Then, combining \eqref{eq:Etru} and \eqref{eq:GA}, by Lemma S1 of \cite{dette2021confidence}, we have
\begin{align}
    &\sup_{x \in \mathbb R}\left| \mathbb{P}\left( \frac{1}{s_n\sqrt{n\tau}}\underset{\substack{\nt \leq i \leq n - \nt\\ 1 \leq j \leq N }}{\max}\left|\sum_{i=1}^n K_{\tau}(i/n-l/n) \mathrm{tr}\{\mf U_i(l/n, j/N)\} \right|\leq x\right)  \right.\\& \left.- \mathbb{P}\left(\frac{1}{s_n \sqrt{n \tau} }\left|\sum_{i=1}^{2\nt}\mf Z_i \right|_{\infty} \leq x\right) \right| \\ 
    &=\sup_{x \in \mathbb R}\left| \mathbb{P}\left( \frac{1}{s_n\sqrt{n\tau}}\underset{\substack{\nt \leq i \leq n - \nt\\ 1 \leq j \leq N }}{\max}\left|\sum_{i=1}^n K_{\tau}(i/n-l/n) \mathrm{tr}[\mf U_i(l/n, j/N) - \E\{\mf U_i(l/n, j/N)\}] \right|\leq x\right) \right.\\&\left. - \mathbb{P}\left(\frac{1}{s_n \sqrt{n \tau} }\left|\sum_{i=1}^{2\nt}\mf Z_i \right|_{\infty} \leq x\right) \right|\\& + P\left(\frac{1}{s_n\sqrt{n\tau}}\underset{\substack{\nt \leq i \leq n - \nt\\ 1 \leq j \leq N }}{\max}\left|\sum_{i=1}^n K_{\tau}(i/n-l/n) \mathrm{tr}[\E\{\mf U_i(l/n, j/N)\}] \right| > \delta\right) \\ &+ O(\delta\sqrt{1\vee \log(nN/\delta)})\\
    &= o(1),\label{eq:Udiscrete}
\end{align}
where in the last equality we can take $\delta = s_n^{3/2} \eta_n\sqrt{\tau/n}$, where $\eta_n$ goes to infinity arbitrarily slow. 
Recall that $K_{\tau}(x) = K(x/\tau)$. Let  $g(t, u) =\frac{1}{s_n\sqrt{n\tau}} \sum_{i=1}^n K_{\tau}(i/n-t) \mathrm{tr}\{\mf U_i(t, u)\}$. Note that 
\begin{align}
    \frac{\partial}{\partial t} g(t, u) &= \frac{1}{s_n\sqrt{n\tau}}\sum_{i=1}^n \tau^{-1} K^{\prime}_{\tau}(i/n-t) \mathrm{tr}\{\mf U_i(t, u)\} \\&+ \frac{1}{s_n\sqrt{n\tau}}\sum_{i=1}^n \sum_{k=1}^{s_n}\tau^{-1} K_{\tau}(i/n-t) 2\mathrm{tr}\left\{\bs \phi_{i,k}(u) \sum_{j=(i-s_n)\vee 1}^{i-M-1}\bs \phi_{j,k}(u) K^{\prime}_{\tau}(j/n-t)\right\}.\label{eq:dgt}
\end{align}
Recall that in \cref{ass:expo}, $L_j(u,t, \F_i) = \frac{\partial H_j(u,t,\F_i)}{\partial u}$. Therefore, for $1 \leq r_1 \leq r_2 \leq p$, we obtain $$\frac{\partial}{\partial u}\phi_{i,k}^{r_1,r_2}(u) = L_{r_1}(u,(i-k)/n, \F_{i-k})H_{r_2}(u, i/n, \F_i) + L_{r_2}(u,i/n, \F_{i})H_{r_1}(u, (i-k)/n, \F_{i-k}).$$
Then,
\begin{align}
    \frac{\partial}{\partial u} g(t, u) &= \frac{1}{s_n\sqrt{n\tau}}\sum_{i=1}^n \sum_{k=1}^{s_n} K_{\tau}(i/n-t) \sum_{r_1,r_2=1}^{p} \left\{\frac{\partial}{\partial u}\phi_{i,k}^{r_1,r_2}(u)\psi_{i,k}^{r_1, r_2}(t,u)\right. \\& \left.+ \phi_{i,k}^{r_1,r_2}(u)\sum_{j=(i-s_n)\vee 1}^{i-M-1}\frac{\partial}{\partial u}\phi_{j,k}^{r_1, r_2}(u)K_{\tau}(j/n-t)\right\}.\label{eq:dgu}
\end{align}

Using \eqref{eq:dgt} and \eqref{eq:dgu}, under the \cref{ass:expo}, similar arguments in \cref{lm:delta1}, and Burkholder's inequality and Cauchy inequality, for $2q \leq s^*$, we have
\begin{align}
&\sup_{t \in [\tau, 1-\tau], u \in [0,1]}  \left\|\frac{\partial}{\partial t} g(t, u) \right\|_q \leq M {s_n}/\tau, \quad \sup_{t \in [\tau, 1-\tau], u \in [0,1]} \left\|\frac{\partial}{\partial u} g(t, u) \right\|_q \leq M {s_n}. \label{eq:dgbound}
\end{align}
Following similar arguments, we get
\begin{align}
    \sup_{t \in [\tau, 1-\tau], u \in [0,1]} \left\| \frac{\partial^2}{\partial t \partial u} g(t, u)\right\|_q \leq M {s_n}/\tau. \label{eq:dgbound2}
\end{align}
By by the definitions of $\sup$ and $\max$, we have
\begin{align}
 &\left| \underset{t \in [\tau, 1-\tau], u \in [0,1]}{\sup} g(t, u) -\underset{ \substack{1 \leq i \leq n\\ 1 \leq j \leq N  }}{\max} g(i/n,  j/N) \right| \\& \leq \left| \underset{t \in [\tau, 1-\tau], u \in [0,1]}{\sup}\underset{ \substack{1 \leq i \leq n\\ 1 \leq j \leq N  }}{\max} |g(t, u) - g(i/n,  j/N)| \right|  := \tilde Z_n.
\end{align}
Similar to (S6.3) of \cite{dette2021confidence}, combining \eqref{eq:dgbound} and \eqref{eq:dgbound2}, by triangle inequality, we have  
\begin{align}
\|\tilde Z_n\|_{q}  = O({s_n} (nN)^{1/q} ((n\tau)^{-1} + N^{-1})).\label{eq:maxg}
\end{align} 
Therefore,  by Lemma S1 of \cite{dette2021confidence}, for any small $\epsilon > 0$, there exists large $n_1 > 0$, $N_1 > 0$, such that for any $n > n_1$, $N > N_1$
\begin{align}
    &\sup_{x \in \mathbb R} \left| \mathbb{P}\left( \frac{1}{s_n\sqrt{n\tau}}\underset{t \in [\tau, 1-\tau], u \in [0,1]}{\sup} \left| \sum_{i=1}^n K_{\tau}(i/n-l/n) \mathrm{tr}\{\mf U_i(l/n, j/N)\} \right| \leq x\right)\right.\\& \left. - \mathbb{P}\left(\frac{1}{n\tau \sqrt{s_n} }\left|\sum_{i=1}^{2\nt}\mf Z_i \right|_{\infty} \leq x\right) \right|\\& \leq \epsilon/4 + C \delta\sqrt{1\vee \log(nN/\delta)}+ C P(\tilde Z_n > \delta)\\
     & \leq \epsilon/4 + C \delta\sqrt{1\vee \log(nN/\delta)}+ C [{s_n} (nN)^{1/q} \{(n\tau)^{-1} + N^{-1}\}]/ \delta^q.
    \label{eq:GA_conti}
\end{align}
Take $\delta = ({s_n}(nN)^{1/q}((n\tau)^{-1} + 1/N) )^{q/(q+1)} = o(1)$, we have by \eqref{eq:maxg} that $$P(\tilde Z_n > \delta) = O(\delta\sqrt{1 \vee \log (nN /\delta)} ).$$ Therefore, there exists large $n_2 \geq n_1$, $N_2 \geq N_1$, such that for any $n > n_2$, $N > N_2$, $\delta\sqrt{1 \vee \log (nN /\delta)}  < \epsilon/4$. \par
 By \cref{prop:esti}, we have
 \begin{align}
      \left\|\sup_{t\in[\tau, 1-\tau], u \in [0,1]}\left|  \sum_{i=1}^{n} \left\{\hat{\mf U}_i(t,u)- \mf U_i(t,u)\right\}K_{\tau}(i/n - t) \right| \right\|_q&=  O(\delta_n),\label{eq:UhatU}
 \end{align}
 where $\delta_n =  s_n^2(n\tau) \{b^2 + 1/(nb)\}\tau^{-2/q^{\prime}} +s_n^{3/2} \sqrt{n/b} \tau^{1-2/q^{\prime}} = o(s_n \sqrt{n\tau})$, for some $q^{\prime} \geq 2$.
 Combining \eqref{eq:GA_conti} and  \eqref{eq:UhatU},  and Lemma S1 of \cite{dette2021confidence} and \cref{prop:esti}, we have for any $\epsilon > 0$, there exists large $n_3 > 0$, $N_3>0$, such that for any $n > n_3$, $N > N_3$
  \begin{align}
  &\sup_{x \in \mathbb R} \left| \mathbb{P}\left( \frac{1}{s_n\sqrt{n\tau}}\underset{t \in [\tau, 1-\tau], u \in [0,1]}{\sup} \left| \sum_{i=1}^n K_{\tau}(i/n-l/n) \mathrm{tr}\{\hat{\mf U}_i(l/n, j/N)\} \right| \leq x\right)\right.\\& \left. - \mathbb{P}\left(\frac{1}{n\tau \sqrt{s_n} }\left|\sum_{i=1}^{2\nt}\mf Z_i \right|_{\infty} \leq x\right) \right|\\& \leq \epsilon/4+ C \delta\sqrt{1 \vee \log (nN /\delta)} + C (s_n \sqrt{n\tau})^{-1}\delta_n \sqrt{1 \vee \log [nN /\{\delta_n/(s_n \sqrt{n\tau})\}]}  < \epsilon,
\end{align}
where $C$ is a positive constant that is large enough.

\section{Proofs of Theorems \ref{thm:alt}  and \ref{thm:bs}} \label{sec6}
  \def\theequation{C.\arabic{equation}}	
  \setcounter{equation}{0}



\subsection{Proof of Theorem \ref{thm:alt}}
Recall that 
\begin{align}
    Q_n &= \max_{\nt \leq i \leq n - \nt, 1 \leq j \leq N}\left|\sqrt{n\tau s_n}\sum_{k=1}^{s_n} \mathrm{tr}\left\{\widehat{\mf G}_k(i/n, j/N) \right\} \right|\\ & = 
    1/\sqrt{n\tau s_n} \max_{\nt \leq i \leq n - \nt, 1 \leq j \leq N} \left|\sum_{l=1}^n K_{\tau}(l/n - i/n) \mathrm{tr}\left\{\hat{\mf U}_l(i/n, j/N) \right\} \right|.
\end{align}
By the proof of \cref{prop:esti} we have,
\begin{align}
\max_{\nt \leq i \leq n - \nt, 1 \leq j \leq N} \left|\sum_{l=1}^n K_{\tau}(l/n - i/n) \mathrm{tr}\left\{\hat{\mf U}_l(i/n, j/N) -{\mf U}_l(i/n, j/N) \right\} \right| = \op(s_n\sqrt{n\tau})
\end{align}
Therefore, we have with probability approaching one that 
\begin{align}
    Q_n/\sqrt{s_n} \geq 
    1/\sqrt{2 n\tau s^2_n} \max_{\nt \leq i \leq n - \nt, 1 \leq j \leq N} \left|\sum_{l=1}^n K_{\tau}(l/n - i/n) \mathrm{tr}\left\{{\mf U}_l(i/n, j/N) \right\} \right|.\label{eq:Qnlb}
\end{align}

Note that since $K(0) = 1$, $\int K(x) dx = 1$, following the steps in \eqref{eq:largeKE} and  \eqref{eq:smallKE} from \cref{lm:4th}, we have
\begin{align}
    &(n\tau s_n)^{-1} \sum_{l=1}^n K_{\tau}(l/n - i/n) \mathrm{tr}\left\{\E {\mf U}_l(i/n, j/N) \right\}\\ 
    & = (n\tau s_n)^{-1}  \sum_{l=1}^n K_{\tau}(l/n - i/n) \sum_{k=1}^{3m} 2 \mathrm{tr}\left[\bs \Gamma_k(l/n, j/N)\right.\\ & \times \left. \left\{\sum_{h=1}^n\boldsymbol\Gamma_k^{\top}(h/n, j/N) K_{\tau}(h/n - i/n) \mf 1(h \in \mathbb L_l)\right\} \right]\\
    &+ (n\tau s_n)^{-1}  \sum_{l=1}^n K_{\tau}(l/n - i/n) \sum_{k=3m+1}^{s_n} 2\mathrm{tr} \left\{\sum_{h=1}^n \bs \Gamma_{l-h}(l/n, j/N)  \boldsymbol\Gamma^{\top}_{l-h}((l-k)/n, j/N)\right.\\ & \times \left. K_{\tau}(h/n - i/n) \mf 1(h \in \mathbb L_l)\right\} + O(s_n^{1/2} (\log n+s_n)/n + ms_n\chi^m)\\
    & = 2\sum_{k=1}^{s_n} |\bs \Gamma_k(i/n, j/N)|_F^2 + O(n^{-1}s_n^{2}+s_n \tau+ms_n\chi^m) \geq 3/2 \HDB{\sum_{k=1}^{s_n}} |\bs \Gamma_k(i/n, j/N)|_F^2, \label{eq:QnE0}
\end{align}
where in the last equality we use \cref{ass:diff}, and  $\sup_{t,u \in [0,1]}|\bs \Gamma_k(t, u)|_F=O(\chi^k)$, following from similar arguments in Lemma 5 of \cite{zhou2010simultaneous}.  Finally, by the uniform Lipchitz continuity with respect to $t$ and $u$, along with\eqref{eq:QnE0}, we have 
\begin{align}
   & \max_{\nt \leq i \leq n - \nt, 1 \leq j \leq N}|(n\tau s_n)^{-1} \sum_{l=1}^n K_{\tau}(l/n - i/n) \mathrm{tr}\left\{\E {\mf U}_l(i/n, j/N) \right\}| \\ &\geq 3/2  \max_{\nt \leq i \leq n - \nt, 1 \leq j \leq N} \HDB{\sum_{k=1}^{s_n}} |\bs \Gamma_k(i/n, j/N)|_F^2 \\ &> \sup_{t \in [\tau,1 - \tau], u \in [0,1]}\HDB{\sum_{k=1}^{s_n}} |\bs \Gamma_k(t, u)|_F^2 >c_2 > 0.\label{eq:QnE}
\end{align}
By \cref{lm:delta1}, for a fixed $s \geq 2$, by Burkholder's inequality, we have 
\begin{align}
   &\left\|\sum_{l=1}^n K_{\tau}(l/n - i/n) \mathrm{tr}\left\{{\mf U}_l(i/n, j/N) -\E {\mf U}_l(i/n, j/N) \right\} \right\|_s/s_n \\ 
   & \leq\sum_{h=0}^{+\infty}\left\|\sum_{l=1}^n K_{\tau}(l/n - i/n) \mathrm{tr}\left\{\proj^{l-h}{\mf U}_l(i/n, j/N) \right\} \right\|_s/s_n \\ 
   & \leq\sum_{h=0}^{+\infty}C \{\sum_{l=1}^n \left\|K_{\tau}(l/n - i/n) \mathrm{tr}\left\{\proj^{l-h}{\mf U}_l(i/n, j/N) \right\} \right\|_s^2\}^{1/2} /s_n \\ 
   & = O(s_n \sqrt{n\tau}),
\end{align}
where $C$ is a positive constant that is large enough. 
Using the chaining argument as in \cref{prop:b2}, we have 
\begin{align}
     &\max_{\nt \leq i \leq n - \nt, 1 \leq j \leq N}\left|\sum_{l=1}^n K_{\tau}(l/n - i/n) \mathrm{tr}\left\{{\mf U}_l(i/n, j/N) -\E {\mf U}_l(i/n, j/N) \right\} \right|\\&=\Op(s_n^2\sqrt{n\tau}\tau^{-2/s}) = \op(s_n {n\tau}) .\label{eq:Qnvar}
\end{align}
Combining \eqref{eq:Qnlb}, \eqref{eq:QnE} and \eqref{eq:Qnvar}, for any $d \in (0,1)$, we have 
\begin{align}
    Q_n/\sqrt{s_n(n\tau)^d} \to \infty.
\end{align}

\subsection{Proof of Theorem \ref{thm:bs}}

\begin{lemma} Under \cref{ass:expo}, for a fixed $q \geq 2$, we have
\begin{align}
\left\| \sup_{t \in [b,1-b], u \in [0,1]} |\hat{\mf m}(t,u) - \mf m(t,u)|_F \right\|_q = O( l_n),\quad l_n = (nb)^{-1/2}b^{-2/q}+ b^2.
\end{align}\label{lm:supm}
\end{lemma}
\begin{proof}
Using Taylor expansion 
\begin{align}
\sup_{u \in [b, 1-b], t \in [0,1]} \left|\E(\hat{\mf m}(t,u)) - \mf m(t,u) -b^2 \int K(v) v^2 dv \frac{\partial}{\partial u^2} \mf m(t,u)/2 \right|_F \leq C \left(\frac{1}{nb} + b^4 \right), \label{eq:Em}
\end{align}
where $C$ is a positive constant that is large enough. 
Note that 
\begin{align}
\hat{\mf m}(t,u) -\E (\hat{\mf m}(t,u)) = \frac{1}{nb} \sum_{i=1}^n  \mf G \left(\frac{i}{n}, u, \FF_i \right) K_b\left( i/n - t  \right).
\end{align}
For fixed $t$ and $u$, we have 
\begin{align}
\| \hat{\mf m}(t,u) -\E (\hat{\mf m}(t,u))  \|_{F,q} = O((nb)^{-1/2}).
\end{align}
By \cref{ass:expo}, we have 
\begin{align}
\left\| \left(\frac{\partial}{\partial t} + \frac{\partial^2}{\partial t \partial u} \right) ( \hat{\mf m}(t,u) -\E (\hat{\mf m}(t,u))  ) \right\|_{F,q}  = O( (nb)^{-1/2} b^{-1}),\label{eq:dEm1}
\end{align}
and
\begin{align}
\left\|\frac{\partial }{ \partial u} ( \hat{\mf m}(t,u) -\E (\hat{\mf m}(t,u))  ) \right\|_{F,q}  = O( (nb)^{-1/2}).\label{eq:dEm2}
\end{align}
Combining \eqref{eq:Em}, \eqref{eq:dEm1} and \eqref{eq:dEm2}, by \cref{prop:b2}, we have 
\begin{align}
\left\| \sup_{t \in [b, 1-b], u \in [0,1]} |\hat{\mf m} (t, u) - \mf m(t,u) |_F  \right\|_q = O((nb)^{-1/2} b^{-2/q} + b^2).
\end{align}
\end{proof}

For $L + 1 \leq j \leq 2\nt - L + 1, 1 \leq k_1 \leq n- 2\nt + 1, 1 \leq j_1 \leq N$, define  
\begin{align}
    \tilde {S}_{j,k_1, j_1} = \frac{1}{\sqrt{2s_n^2L}} (\sum_{r = j}^{j + L - 1} - \sum_{r = j-L}^{j- 1}) \hat {V}_{r + k_1 - 1,\nt + k_1 - 1, j_1},\label{eq:deftildeS}
\end{align}
where $\hat {V}_{i,l, j}$ denotes the $j$-th element of $\hat{\mf V}_{i,l}$ as defined in \eqref{eq:hatV1}, $L \leq i \leq 2\nt -L$, $0 \leq l \leq n - 2\nt$, $1 \leq j \leq N$. Denote the corresponding moving sum without estimation error as \begin{align}
{S}_{j,k_1, j_1} = \frac{1}{\sqrt{2s_n^2L}} \left(\sum_{r = j}^{j +L - 1} - \sum_{r = j-L}^{j- 1}\right)  {V}_{r + k_1 - 1,\nt + k_1 - 1, j_1},
    \label{eq:defS}
\end{align} 
where $L + 1 \leq j \leq 2\nt - L + 1$, $1 \leq k_1 \leq n- 2\nt + 1$, $1 \leq j_1 \leq N$. 
Let 
\begin{align}
    \tilde W(t,u, i/n, \FF_{i})  = \sum_{r=1}^p (W_{r,r}(t,u, i/n, \FF_{i}) - \E W_{r,r}(t,u, i/n, \FF_{i})),  \label{eq:tildeW}
\end{align}
where $W_{r,r}(t,u, i/n, \FF_{i})$ is as defined in \eqref{eq:defineW}. Write $$\tilde W_{i,k,j} = \tilde W\left(\frac{\nt + k-1}{n}, \frac{j}{N}, \frac{i + k -1}{n}, \F_{i + k -1} \right),$$ 
where $1 \leq i \leq 2\nt$, $1 \leq k \leq n- 2\nt + 1$, $1 \leq j \leq N$.
Define the counterpart of $S_{j,k_1, j_1}$, $L + 1 \leq j \leq 2\nt - L + 1$, $1 \leq k_1 \leq n- 2\nt + 1$, $1 \leq j_1 \leq N$, with zero mean as 
\begin{align}
    \bar S_{j, k_1, j_1} =\frac{1}{\sqrt{2 L}} \left(\sum_{r = j}^{j + L - 1}  - \sum_{r = j-L}^{j- 1}\right) K_{\tau}\left(\frac{r-\nt}{n}\right)\tilde W_{r, k_1, j_1},\label{eq:Sbar}
\end{align}
Recall the definition of $\tilde{\mf S}_{j,l}$ in \cref{alg:boot2}, i.e,
\begin{align}
    \tilde{\mf S}_{j,l} = \frac{1}{\sqrt{2s_n^2 L}} \Big ( \sum_{r = j}^{j + L - 1} \hat{\mf V}_{r,l} - \sum_{r = j-L}^{j- 1} \hat{\mf V}_{r,l} \Big ),
\end{align}
and that ${\tilde {\mf S}}_{j,l} = (\tilde {S}_{j,l, j_1}, 1 \leq j_1 \leq N)^{\top}$. Let  ${\mf S}_{j,l} = ({S}_{j,l, j_1}, 1 \leq j_1 \leq N)^{\top}$,  ${\bar {\mf S}}_{j,l} = (\bar {S}_{j,l, j_1}, 1 \leq j_1 \leq N)^{\top}$ .  For $0 \leq l \leq n - 2\nt$, define $N$-dimensional vectors
\begin{align}
\tilde{\mf Z}_l = \sum_{j=L}^{2\nt -L} \tilde{\mf S}_{j,l} R_{l+j}/\sqrt{\nt - L}, \quad \bar {\mf Z}_l = \sum_{j=L}^{2\nt -L} {\mf S}_{j,l} R_{l+j}/ \sqrt{\nt - L},\label{eq:defZ}
\end{align}
and 
\begin{align}
\check{\mf Z}_l =  \sum_{j=L}^{2\nt -L} \bar{\mf S}_{j,l} R_{l+j}/\sqrt{\nt - L}.\label{eq:defZcheck}
\end{align}
\begin{lemma}\label{lm:covW}
    Under the condition of \cref{lm:delta1}, for $1 \leq i_1, i_2 \leq 2\nt$, $1 \leq k_1, k_2 \leq n - 2\nt +1$, $1 \leq j_1, j_2 \leq N$, and some $\chi \in (0,1)$,  \par
    (i) if $|i_2 + k_2 - i_1 - k_1| > 4s_n$, we have 
    \begin{align}
        \left|\E \left(K_{\tau}\left(\frac{i_1-\nt}{n}\right) \tilde W_{i_1, k_1, j_1 }K_{\tau}\left(\frac{i_2-\nt}{n}\right) \tilde W_{i_2, k_2, j_2 } \right) \right| = O(\chi_1^{|i_2+k_2-i_1-k_1|/2}),
    \end{align}
     (ii) if $|i_2 + k_2 - i_1 - k_1| \leq 4s_n$, we have 
    \begin{align}
       \max_{\substack{1 \leq i_1, i_2 \leq 2\nt, \\1 \leq k_1, k_2 \leq n - 2\nt +1,\\ 1 \leq j_1, j_2 \leq N}} \left|\E \left(K_{\tau}\left(\frac{i_1-\nt}{n}\right) \tilde W_{i_1, k_1, j_1 }K_{\tau}\left(\frac{i_2-\nt}{n}\right) \tilde W_{i_2, k_2, j_2 } \right) \right| = O (s_n^2).
    \end{align}
\end{lemma}
\begin{proof}
    Without the loss of generality, suppose $i_1 + k_1 < i_2 + k_2$ and we have
    \begin{align}
         &\left|\E \left(K_{\tau}\left(\frac{i_1-\nt}{n}\right) \tilde W_{i_1, k_1, j_1 }K_{\tau}\left(\frac{i_2-\nt}{n}\right) \tilde W_{i_2, k_2, j_2 } \right) \right|\\ 
         & \leq \sum_{l \in \mathbb Z} \E \left|K_{\tau}\left(\frac{i_1-\nt}{n}\right)K_{\tau}\left(\frac{i_2-\nt}{n}\right) (\proj^l \tilde W_{i_1, k_1, j_1})\times (\proj^l \tilde W_{i_2, k_2, j_2}) \right| \\ 
         & \leq \sum_{l \in \mathbb Z} \|\proj^l  \tilde W_{i_1, k_1, j_1}\| \times \|\proj^l  \tilde W_{i_2, k_2, j_2}\| \\
         & = \left(\sum_{l \leq i_1 + k_1 -4s_n} +  \sum_{l > i_1 + k_1 -4s_n}^{i_1 + k_1} \right) \theta_{i_1 + k_1 - l}\theta_{i_2 + k_2 - l}.
    \end{align}
    Using above arguments and  \cref{lm:delta1}, if $|i_2 + k_2 - i_1 - k_1| > 4s_n$, we have 
    \begin{align}
       & \left|\E \left(K_{\tau}\left(\frac{i_1-\nt}{n}\right) \tilde W_{i_1, k_1, j_1 }K_{\tau}\left(\frac{i_2-\nt}{n}\right) \tilde W_{i_2, k_2, j_2 } \right) \right|\\ &= O\left(\sum_{l \leq i_1 + k_1 -4s_n} \chi_1^{4s_n}\chi_1^{i_2+k_2-l} +  \Theta_{0,2}\chi_1^{i_2+k_2-i_1-k_1} \right) = O(\chi^{i_2 + k_2 - i_1-k_1}).
    \end{align}
    If $|i_2 + k_2 - i_1 - k_1| \leq 4s_n$, we have by \cref{lm:delta1}
    \begin{align}
        \left|\E \left(K_{\tau}\left(\frac{i_1-\nt}{n}\right) \tilde W_{i_1, k_1, j_1 }K_{\tau}\left(\frac{i_2-\nt}{n}\right) \tilde W_{i_2, k_2, j_2 } \right) \right| = O(\Theta^2_{0,2}) = O(s_n^2). 
    \end{align}
\end{proof}
\subsection{Proof of Theorem \ref{thm:bs}}\par
\begin{proof}
    \textbf{Proof of (i).}
We shall break the proof into three steps. Recall the notation of $\tilde{\mf Z}_l$, $\bar {\mf Z}_l$, and  $\check{\mf Z}_l$ in \eqref{eq:defZ} and \eqref{eq:defZcheck}. In Step 1, we shall investigate gap between the distributions caused by the estimation error, namely $$\mathbb{E}\left(\max_{0 \leq l \leq n-2\nt}|\tilde{\mf Z}_l  - \bar {\mf Z}_l|^q_{\infty} \mf 1(\mathbb A_n)|\FF_n \right) .$$ In Step 2, we shall show that the difference-based multiplier can largely cancel the trend function,  so that $$\mathbb{E}\left(\max_{0 \leq l \leq n-2\nt}|\tilde{\mf Z}_l  - \bar {\mf Z}_l|^q_{\infty} \mf 1(\mathbb A_n)|\FF_n \right).$$ In Step 3, we shall compute $$\sup_{x \in \R} \left|P \left(\max_{0 \leq l \leq n - 2\nt}|\check{\mf Z}_l|_{\infty} < x|\FF_n \right) -  \mathbb{P}\left(\frac{1}{\sqrt{s_n} n\tau} \left| \sum_{i=1}^{2\nt} {\mf Z}_i \right|_{\infty} < x\right)  \right|$$ the stochastic error induced by the bootstrap approximation, and derive the upper bound for
\begin{align}
    \sup_{x \in \R} \left|P \left(\max_{0 \leq l \leq n - 2\nt} |\tilde{\mf Z}_l|_{\infty}  < x|\FF_n \right) -  \mathbb{P}\left(\frac{1}{\sqrt{s_n} n\tau} \left| \sum_{i=1}^{2\nt} {\mf Z}_i \right| < x\right)  \right| .
\end{align}
\par
\textbf{Step 1}: 
Let $l_n = (nb)^{-1/2}b^{-1/q} + b^2$.   For some $\eta_n \to \infty$, let 
\begin{align}
\mathbb A_n = \left\{\sup_{t \in [b, 1-b], u \in [0,1]} |\hat{\mf m}(t,u) - {\mf m}(t,u) |_F \leq  \eta_n  l_n\right\}, 
\end{align}
By \cref{lm:supm}, we have $P(\mathbb A_n)  \to 1$.
Recall the definitions of $\tilde{\mf Z}_l$, $\bar {\mf Z}_l$ in \eqref{eq:defZ}, and the definitions of $\tilde S_{j,l,k}$ and $S_{j,l,k}$ in \eqref{eq:deftildeS} and \eqref{eq:defS} respectively. 
By Lemma 2.3.4 of \cite{gine2016mathematical}, the conditional Jensen inequality and concentration of the maximum of a Gaussian process, $\nt \geq 2L$, we have 
\begin{align}
&\mathbb{E}\left(\max_{0 \leq l \leq n-2\nt}|\tilde{\mf Z}_l  - \bar {\mf Z}_l|^q_{\infty} \mf 1(\mathbb A_n)|\FF_n \right)  \\ &\leq M \left|\sqrt{\frac{\log  (n  N)}{\nt-L}  }\left(\max_{\substack{0 \leq  l \leq n-2\nt,\\ 1 \leq k \leq N} }\sum_{j=L}^{2\nt -L} (\tilde S_{j,l,k} - S_{j,l,k})^2 \mf 1(\mathbb A_n) \right)^{1/2} \right|^q\\ 
& \leq M   \left|\sqrt{\frac{\log (n  N) }{2^{-1} s_n^2 (n\tau) L } } \left(\max_{\substack{0 \leq  l \leq n-2\nt,\\ 1 \leq k \leq N} }\sum_{j=L}^{2\nt -L} \left[ \left\{\sum_{i=j}^{j+L-1} (\hat V_{i + l -1, \nt + l - 1,k} -  V_{i+l-1,\nt+ l-1,k}) \right\}^2 \right. \right. \right.\\ & + \left. \left. \left. \left\{\sum_{i=j - L }^{j-1} (\hat V_{i+l-1,\nt+ l-1,k} - V_{i,\nt+ l-1,k}) \right\}^2 \right]\mf 1(\mathbb A_n) \right)^{1/2} \right|^q, \label{eq:step1.1}
\end{align}
where $\hat V_{i,l,j} = K_{\tau}(\frac{i-l}{n}) \mathrm{tr} \left\{\hat{\mf U}_i (l/n, j/N)\right\}$. Note that 
\begin{align}
    \hat V_{i,l,k} -  V_{i,l,k} & =  K_{\tau}\left(\frac{i-l}{n} \right) \{\mathrm{tr}(\hat{\mf U}_i (l/n, j/N)) - \mathrm{tr}({\mf U}_i (l/n, j/N))\},
\end{align}
where for some positive constant $C$
\begin{align}
&\mathrm{tr}(\hat{\mf U}_i (t, u) - {\mf U}_i (t, u)) \mf 1(\mathbb A_n) \\ &= \sum_{k=1}^{s_n}\sum_{j=1}^{n}\mathrm{tr} \left\{(\hat{\bs \epsilon}_{i-k}(u)- \bs \epsilon_{i-k}(u)) \bs \epsilon^{\top}_i(u)\left( \bs \epsilon_{j-k}(u)\bs \epsilon^{\top}_j(u)K_{\tau}(j/n - t) \right)^{\top} \right\} \mf 1(\mathbb A_n) \mf 1 (j \in \mathbb L_n) \\ 
& + \sum_{k=1}^{s_n}\sum_{j=1}^{n} \mathrm{tr} \left\{\hat{\bs \epsilon}_{i-k}(u) (\hat{\bs \epsilon}_i(u) - \bs \epsilon_i(u))^{\top} \left( \bs \epsilon_{j-k}(u)\bs \epsilon^{\top}_j(u)K_{\tau}(j/n - t) \right)^{\top}  \right\}\mf 1(\mathbb A_n) \mf 1 (j \in \mathbb L_n)\\ 
& +  \sum_{k=1}^{s_n}\sum_{j=1}^{n} \mathrm{tr}\left\{\hat{\bs \epsilon}_{i-k}(u) \hat{\bs \epsilon}^{\top}_i(u)\left( \hat{\bs \epsilon}_{j-k}(u) - \bs \epsilon_{j-k}(u)) \bs \epsilon^{\top}_j(u)K_{\tau}(j/n - t) \right)^{\top} \right\} \mf 1(\mathbb A_n) \mf 1 (j \in \mathbb L_n)\\
&  + \sum_{k=1}^{s_n}\sum_{j=1}^{n} \mathrm{tr}\left\{\hat{\bs \epsilon}_{i-k}(u) \hat{\bs \epsilon}^{\top}_i(u)\left( \hat{\bs \epsilon}_{j-k}(u)  (\hat{\bs \epsilon}^{\top}_j(u) - \bs \epsilon^{\top}_j(u))K_{\tau}(j/n - t) \right)^{\top} \right\} \mf 1(\mathbb A_n)\mf 1 (j \in \mathbb L_n)\\ 
&:= A_{i1}(t,u) + \cdots + A_{i4}(t,u),
\end{align}
where $A_{i1}(t,u) \ldots  A_{i4}(t,u)$ are defined in the obvious way.
For the sake of brevity, we only show the calculation related to $A_{i1}(t,u)$. For this purpose, notice that
\begin{align}
  &\left\|\max_{\substack{0 \leq  l \leq n-2\nt,\\ 1 \leq k \leq N} } \left(\sum_{j=L}^{2\nt -L} \left\{ \sum_{i=j}^{j+L-1} A_{(i+l-1)1}\left(\frac{\nt + l - 1}{n}, \frac{k}{N}\right)\right\}^2 \mf 1(\mathbb A_n)\right)^{1/2} \right\|_q \\
  & \leq \sqrt{n\tau} \left\| \max_{\substack{0 \leq  l \leq n-2\nt,\\ 1 \leq k \leq N,\\ L \leq j \leq 2\nt - L} } \left|\sum_{i=j}^{j+L-1} A_{(i+l-1)1}\left(\frac{\nt + l - 1}{n}, \frac{k}{N}\right)\mf 1(\mathbb A_n) \right| \right\|_q\\ 
  & \leq  \sqrt{n\tau} L \left\|\max_{\substack{1 \leq  i \leq n, 1 \leq k \leq N,\\ \nt \leq j \leq n -\nt + 1} } A_{i1}(j/n, k/N)  \right\|_q  \\ 
  & \leq \sqrt{n\tau} L \left\|\max_{\substack{1 \leq  i \leq n, 1 \leq k \leq N,\\ \nt \leq j \leq n -\nt + 1} } \sum_{h=1}^{s_n} |(\hat {\bs \epsilon}_{i-h}(k/N) - \bs \epsilon_{i-h}(k/N))1(\mathbb A_n)|_F\right.  \\
  & \times \left.  |\bs \epsilon_i^{\top}(k/N) \bs\psi_{i,h}(j/n, k/N)|_F  \right\|_q \\
  & \leq C \sqrt{n\tau} L \left[\sum_{\substack{1 \leq  i \leq n, 1 \leq k \leq N,\\ \nt \leq j \leq n -\nt + 1} } \E\left\{ \sum_{h=1}^{s_n} |(\hat {\bs \epsilon}_{i-h}(k/N) - \bs \epsilon_{i-h}(k/N))1(\mathbb A_n)|_F \right. \right. \\
  & \times \left. \left. |\bs \epsilon_i^{\top}(k/N) \bs\psi_{i,h}(j/n, k/N)|_F  \right\}^q \right]^{1/q} \\
  & = O(L s^{3/2}_n \sqrt{n\tau} l_n\eta_n(n^2N)^{1/q}), \label{eq:8A1}
\end{align}
where the third inequality follows from the definition of $\bs\psi_{i,h}(t, u)$, see \eqref{eq:defU}, $C$ is a positive constant that is large enough and the last equality follows from 
\eqref{eq:psimoment}.
Therefore, applying similar arguments to $A_{i2}(t,u), \ldots,  A_{i4}(t,u)$ in \eqref{eq:8A1} and triangle inequality, it follows that 
\begin{align}
& \left\|\max_{\substack{0 \leq  l \leq n-2\nt,\\ 1 \leq k \leq N} } \left(\sum_{j=L}^{2\nt -L} \left[ \left\{\sum_{i=j}^{j+L-1} (\hat V_{i + l -1, \nt + l - 1,k} -  V_{i+l-1,\nt+ l-1,k}) \right\}^2 \right. \right. \right.\\ & + \left. \left. \left. \left\{\sum_{i=j - L }^{j-1} (\hat V_{i+l-1,\nt+ l-1,k} - V_{i,\nt+ l-1,k}) \right\}^2 \right]\mf 1(\mathbb A_n) \right)^{1/2}  \right\|_q \\ 
 &= O (L s^{3/2}_n \sqrt{n\tau} l_n\eta_n(n^2N)^{1/q})  .\label{eq:step1.2}
\end{align}
By \eqref{eq:step1.1} and  \eqref{eq:step1.2},  we have 
\begin{align}
\left\{\mathbb{E}\left(\max_{0 \leq l \leq n-2\nt}|\tilde{\mf Z}_l  - \bar {\mf Z}_l|^q_{\infty} \mf 1(\mathbb A_n)|\FF_n \right) \right\}^{1/q} = \Op(\sqrt{L s_n \log(nN) }l_n\eta_n(n^2N)^{1/q}). 
\label{eq:step1rate}
\end{align}
\textbf{Step 2}: Using similar arguments in \eqref{eq:EUstep1} and \eqref{eq:EUstep2} of \cref{lm:4th}, under the null hypothesis, we have
 \begin{align}
    &\sup_{t\in[\tau, 1-\tau], u \in [0,1]} \left|\sum_{r=j}^{j+L-1} K_{\tau}(r/n - t) \E [\mathrm{tr} \left\{\mf U_{r}(t,u)\right\}]\right|\\
    & = \sup_{t\in[\tau, 1-\tau], u \in [0,1]}\left|\sum_{r=j}^{j+L-1} K_{\tau}(r/n - t) \sum_{k=1}^{3m} 2 \mathrm{tr}\left[\bs \Gamma_k(r/n, u) \left\{\sum_{h=1}^n\boldsymbol\Gamma_k^{\top}(h/n, u) K_{\tau}(h/n - t) \mf 1(h \in \mathbb L_r)\right\} \right]\right.\\
    &\left.+ \sum_{r=j}^{j+L-1} K_{\tau}(r/n - t) \sum_{k=3m+1}^{s_n} 2\mathrm{tr} \left\{\sum_{h=1}^n \bs \Gamma_{r-h}(r/n , u) \boldsymbol\Gamma^{\top}_{r-h}((r-k)/n, j/N) K_{\tau}(h/n - t) \mf 1(h \in \mathbb L_r)\right\} \right|\\
    &+ O(Ls_n^{3/2}(\log n +s_n)/n+  L(\log n)s_n/n)\\
    & = O(Ls_n^{3/2}(\log n +s_n)/n + L(\log n)s_n/n ) = O(Ls_n^{5/2}/n),\label{eq:rowLEu}
\end{align}
since $s_n/\log n \to \infty$.
Similar to the calculation in \eqref{eq:step1.1}, uniformly for $t$ and $u \in [0,1]$, by \eqref{eq:rowLEu}, we have 
\begin{align}
&\E\left(\max_{0 \leq l \leq n-2\nt}|\check{\mf Z}_l  -\bar{\mf Z}_l|^q_{\infty}|\F_n\right)\\  &\leq  M \left|\sqrt{\frac{\log (n  N)}{2^{-1} (n\tau) L } } \left[\max_{\substack{0 \leq  l \leq n-2\nt,\\ 1 \leq k \leq N} }\sum_{j=L}^{2\nt -L} \left\{\sum_{i=j - L }^{j-1} ( V_{i+l-1,\nt+l - 1,k}/s_n  -  V_{i+L+l-1,\nt + l - 1,k}/s_n  \right.\right. \right. \\&  \left. \left.\left. - \tilde W_{i,l,k} + \tilde W_{i+L,l,k}) \right\}^2 \right]^{1/2} \right|^q \\ 
& \leq  M \left|\sqrt{\frac{\log( n  N)}{2^{-1}s_n^2 (n\tau) L } } \left\{ \max_{\substack{0 \leq  l \leq n-2\nt,\\ 1 \leq k \leq N} }\sum_{j=L}^{2\nt -L} \right.\right.\\ & \left.\left.\left(\sum_{i=j - L }^{j-1} K_{\tau}\left(\frac{i-\nt}{n} \right)\E \left[\mathrm{tr}\left\{\mf U_{i+l-1}\left(\frac{\nt + l -1}{n}, \frac{k}{N}\right) \right\}\right]\right.\right.\right.\\ & \left.\left.\left.- K_{\tau}\left(\frac{i+L-\nt}{n}\right) \E \left[\mathrm{tr}\left\{\mf U_{i+L+l-1}\left(\frac{\nt + l -1}{n}, \frac{k}{N}\right)\right\}
\right] \right)^2 \right\}^{1/2} \right|^q \\ 
&  = O\left\{(L\log (n  N))^{q/2}(s_n^{3/2}/n)^q \right\}.\label{eq:step2diff}
\end{align}

Combining \eqref{eq:step1rate}, \eqref{eq:step2diff}, we have 
\begin{align}
   &\left\{ \E\left(\max_{0 \leq l \leq n-2\nt}|\check{\mf Z}_l  -\tilde{\mf Z}_l|^q_{\infty}|\F_n\right)\right\}^{1/q} \\ &= \Op\left\{(L\log (n  N))^{1/2} ( s_n^{3/2}/n) + \sqrt{L s_n \log(nN) }l_n\eta_n(n^2N)^{1/q}\right\}.\label{eq:step2rate}
\end{align}
\textbf{Step 3}:  
For $1 \leq k_1, k_2 \leq n - 2\nt +1$, $1 \leq j_1, j_2 \leq N$, let 
\begin{align}
\bar \sigma^{2}_{k_1, k_2, j_1, j_2} =(\nt - L)^{-1} \sum_{r=L}^{2\nt - L- (k_2-k_1)}  \bar S_{r + k_2-k_1,  k_1, j_1} \bar S_{r,  k_2, j_2},\label{eq:sigmabar}
\end{align} 
where $\bar S_{r,  k_2, j_2}$ is as defined in \eqref{eq:Sbar}, i.e.,
\begin{align}
    \bar S_{j, k_2, j_2} =\frac{1}{\sqrt{2L}} \left(\sum_{r = j}^{j + L - 1}  - \sum_{r = j-L}^{j- 1}\right) K_{\tau}\left(\frac{r-\nt}{n}\right)\tilde W_{r, k_2, j_2}.\label{eq:defSbar}
\end{align}
\par
We shall show that $\bar \sigma^2_{k_1,k_2,j_1,j_2}$ converges to the autocovariance structure of ${\mf Z}_i$. Recall the definition of ${\mf V}_{i}$ in \eqref{det27} that 
\begin{align}
\mf V_i = ({\mf V}_{i,\nt} ,\cdots, {\mf V}_{i + n - 2\nt, n-\nt} )^{\top}.
\end{align}
The autocovariance structure of ${\mf V}_{i}/(n \tau \sqrt{s_n})$ is
\begin{align}
    \sigma^2_{k_1, k_2, j_1, j_2} = \frac{1}{(n\tau) s_n^2} \mathrm{Cov}\left(\sum_{i_1 = 1}^{2 \nt} v_{i_1, (k_1 - 1)N + j_1},  \sum_{i_2 = 1}^{2 \nt} v_{i_2, (k_2 - 1)N + j_2} \right),
\end{align}
where  $1 \leq k_1, k_2 \leq n - 2\nt +1$, $1 \leq j_1, j_2 \leq N$, $v_{i,r}$ is the $r$-th component of $\mf V_i$, $1 \leq i \leq 2\nt$, $1 \leq r \leq N(n-2\nt+1)$.  Note that $v_{i_1, (k_1 -1)N+j_1} =  V_{i_1 + k_1 -1, \nt + k_1 -1, j_1}$, where $ V_{i_1 + k_1 -1, \nt + k_1 -1, j_1}$ is the $j_1$th element of $\mf  V_{i_1 + k_1 -1, \nt + k_1 -1}$.  Recall the definitions in \eqref{eq:defU} and \eqref{eq:defineW}, we can write 
\begin{align}
     \sigma^2_{k_1, k_2, j_1, j_2} = \frac{1}{n\tau} \E \left\{ \sum_{i_1, i_2 = 1}^{2\nt} K_{\tau}\left(\frac{i_1-\nt}{n}\right)K_{\tau}\left(\frac{i_2 -\nt}{n}\right) \tilde W_{i_1, k_1, j_1} \tilde W_{i_2, k_2, j_2} \right\}.\label{eq:sigmalimit}
\end{align}
At the end of the proof, we shall show that
\begin{align}
   \left\| \max_{\substack{1 \leq k_1, k_2 \leq n - 2\nt + 1,\\ 1 \leq j_1, j_2 \leq N}} \left| \sigma^2_{k_1, k_2, j_1, j_2}  -  \bar \sigma^{2}_{k_1, k_2, j_1, j_2}\right| \right\|_{q/2} = O\left(\vartheta_n\right),\label{eq:sigmaconverge}
\end{align}
where $$\vartheta_n = \frac{(s^2_n\log n)^2}{L} + \frac{s_n^3L\log n}{n\tau} + \sqrt{\frac{L}{n\tau}}s_n^2 (nN)^{4/q}.$$

By \cref{lm:cum}, there exist a $\eta_0 > 0$ such that
\begin{align}
    P(\max_{\substack{1 \leq k_1, k_2 \leq n - 2\nt + 1,\\ 1 \leq j_1, j_2 \leq N}} \bar \sigma^{2}_{k_1, k_2, j_1, j_2}> \eta_0) \geq 1 - O(\vartheta_n^{q/2}).
\end{align}

Similar to (S6.9) of \cite{dette2021confidence}, we have 
\begin{align}
&\sup_{x \in \R} \left|P \left(\max_{0 \leq l \leq n - 2\nt}|\check{\mf Z}_l|_{\infty} < x|\FF_n \right) -  \mathbb{P}\left(\frac{1}{\sqrt{s_n} n\tau} \left| \sum_{i=1}^{2\nt} {\mf Z}_i \right|_{\infty} < x\right)  \right| \\&= \Op\left(\vartheta_n^{1/3} \left\{1 \vee \log (nN/\vartheta_n) \right\}^{2/3} \right). \label{eq:step3rate}
\end{align}
Combining \eqref{eq:step2rate} and \eqref{eq:step3rate}, by the (conditional version) Lemma S1 of \cite{dette2021confidence}, we have 
\begin{align}
       & \sup_{x \in \R} \left|P \left(\max_{0 \leq l \leq n - 2\nt} |\tilde{\mf Z}_l|_{\infty}  < x|\FF_n \right) -  \mathbb{P}\left(\frac{1}{\sqrt{s_n} n\tau} \left| \sum_{i=1}^{2\nt} {\mf Z}_i \right| < x\right)  \right| \\ 
       & \leq  \sup_{x \in \R} \left|P \left(\max_{0 \leq l \leq n - 2\nt}|\check{\mf Z}_l|_{\infty} < x|\FF_n \right) -  \mathbb{P}\left(\frac{1}{\sqrt{s_n} n\tau} \left| \sum_{i=1}^{2\nt} {\mf Z}_i \right|_{\infty} < x\right)  \right|\\ 
       & + P(\max_{0 \leq l \leq n - 2\nt}|\tilde{\mf Z}_l - \check{\mf Z}_l|_{\infty}> \delta|\FF_n) + O(\delta \sqrt{1  \vee \log(nN/\delta)} )\\
       &= \Op[ \{\sqrt{L s_n \log(nN) }l_n\eta_n(n^2N)^{1/q} + (L\log (n  N))^{1/2} ( s_n^{3/2}/n) \}^{\frac{q}{q+1}} +\vartheta_n^{1/3} \left\{1 \vee \log (nN/\vartheta_n) \right\}^{2/3} ].
\end{align}
where $l_n =  (nb)^{-1/2} b^{-1/q} + b^2$, and in the last equality we take $\delta =  \{\sqrt{L s_n \log(nN) }l_n\eta_n(n^2N)^{1/q} + (L\log (n  N))^{1/2} ( s_n^{3/2}/n) \}^{\frac{q}{q+1}}$. \par
\textbf{Proof of \cref{eq:sigmaconverge}}. We follow the steps in  (5.12) of \cite{dette2019}, mainly investigating the influence of increasing dependence. 
If $k_2-k_1 > 2\nt -2L$, $\bar \sigma^2_{k_1, k_2, j_1, j_2} = 0$, where in \eqref{eq:rowLEu} we use the convention $\sum_{i=a}^b x_i = 0$, if $b < a$. By \cref{lm:covW}, for  $k_2-k_1 > 2\nt -2L$, we have
\begin{align}
    \sigma^2_{k_1, k_2, j_1, j_2} &= \frac{1}{n\tau} \sum_{i_1, i_2 = 1}^{2\nt} O\left\{\chi_1^{|i_2 + k_2 - i_1 - k_1|/2} \mf 1(|i_2 + k_2 - i_1 - k_1| > 4 s_n) \right. \\& + \left. s_n^2\mf 1(|i_2 + k_2 - i_1 - k_1| \leq 4 s_n)\right\}\\ 
    &= O(\chi_1^{4s_n})+ \frac{1}{n\tau} \sum_{i_1 = 2n\tau-2L-4s_n+1}^{2\nt}\sum_{i_2 = 1}^{2\nt} O\left\{ s_n^2\mf 1(|i_2 + k_2 - i_1 - k_1| \leq 4 s_n)\right\}\\ 
    & = O(\chi_1^{4s_n} + s_n^3 (L+s_n)/(n\tau)) = O(s_n^3 L/(n\tau)).\label{eq:5.20}
\end{align}
If $k_2-k_1 \leq 2\nt -2L$, after a careful investigation of the proof of (S6.27) in \cite{dette2021confidence}, take $a = s_n \log n$ therein, combining \cref{lm:covW} and \cref{lm:delta}, we have 
\begin{align}
   \max_{\substack{1 \leq k_1, k_2 \leq n - 2\nt + 1,\\ 1 \leq j_1, j_2 \leq N}}  |\E(\bar \sigma^2_{k_1, k_2, j_1, j_2}) - \sigma^2_{k_1, k_2, j_1, j_2} | = O\left((s_n \log n)^2 s_n^2/L + Ls^3_n 
   \log n/(n\tau)\right),\label{eq:5.25}
\end{align}
and after  a careful investigation of the proof of (S6.28) in \cite{dette2021confidence} we obtain
\begin{align}
   \left\|  \max_{\substack{1 \leq k_1, k_2 \leq n - 2\nt + 1,\\ 1 \leq j_1, j_2 \leq N}}  |\bar \sigma^2_{k_1, k_2, j_1, j_2} - \E(\bar \sigma^2_{k_1, k_2, j_1, j_2})| \right\|_{q/2} = O\left(\sqrt{\frac{L}{n\tau}} s_n^2 (nN)^{4/q}\right).\label{eq:5.32}
\end{align}
The result follows from \eqref{eq:5.20}, \eqref{eq:5.25} and \eqref{eq:5.32}. 
\par
\textbf{Proof (ii).}\par
Define $\bar{\bs \phi}_{i,k}(u) = \bs \epsilon_{i-k}(u)\bs \epsilon^{\top}_{i}(u) - \bs \Gamma_k(i/n, u)$, $\bar{\bs \psi}_{i,k}(t,u) =\sum_{j=(i-s_n)\vee 1}^{i-M-1} \bar{\bs \phi}_{j,k}(u) K_{\tau}(j/n-u)$. Recall the definition of ${\mf U}_i(t, u)$ in \eqref{eq:defU}. Under the alternative hypothesis \eqref{eq:Ha}, we have the following decomposition
\begin{align} 
&\mathrm{tr}\left\{\mf U_i(t,u)\right\} \\ &= \left[\sum_{k=1}^{s_n}\mathrm{tr} \left\{\bar{\bs \phi}_{i,k}(u)\bar{\bs \psi}^{\top}_{i,k}(t,u) + \bar{\bs \psi}_{i,k}(t,u)\bar{\bs \phi}^{\top}_{i,k}(u)\right\}  \right]\\ 
& + \left(\sum_{k=1}^{s_n}\mathrm{tr}\left[\bs \Gamma_k(i/n, u) \bar{\bs \psi}^{\top}_{i,k}(t,u) + \bar{\bs \phi}_{i,k}(u) \sum_{j=(i-s_n)\vee 1}^{i-M-1}  \bs \Gamma_k^{\top}(j/n, u) K_{\tau}(j/n - t) \right.\right. \\ &\left.\left.  +\left\{ \sum_{j=(i-s_n)\vee 1}^{i-M-1} \bs \Gamma_k(j/n, u) K_{\tau}(j/n - t) \right\}\bar{\bs \phi}^{\top}_{i,k}(u)+ \bar{\bs \psi}_{i,k}(t,u)\bs \Gamma_k^{\top}(i/n, u) \right] \right)\\
&+ \left[\sum_{k=1}^{s_n}\mathrm{tr}\left\{\bs \Gamma_k(i/n, u) \sum_{j=(i-s_n)\vee 1}^{i-M-1} \bs \Gamma_k^{\top}(j/n, u) K_{\tau}(j/n - t) \right.\right. \\ &\left.\left. + \sum_{j=(i-s_n)\vee 1}^{i-M-1} \bs \Gamma_k(j/n, u) K_{\tau}(j/n - t)  \bs \Gamma^{\top}_k(i/n, u) \right\}\right]\\ 
& := u_{i1}(t,u)  + u_{i2}(t,u) + u_{i3}(t,u),\label{eq:ualt}
\end{align}
where the formulas of $u_{i1}(t,u)$ is the same as $\mathrm{tr}\left\{\mf U_i(t,u)\right\} /\sqrt{n \tau s_n} $ under $H_0$, see \eqref{eq:recallU}, $u_{i2}(t,u)$ is the extra randomness under $H_A$,  $u_{i3}(t,u)$ is a deterministic function of $t$ and $u$.\par 

Note that the arguments for \eqref{eq:step1.1} and \eqref{eq:step1.2} do not reply on the null hypothesis. 
Following \eqref{eq:step1.1} and \eqref{eq:step1.2} of Step 1, since $P(\mathbb A_n) \to 1$,  we have
\begin{align}
   \max_{0 \leq l \leq n-2\nt}|\tilde{\mf Z}_l - \bar {\mf Z}_l|_{\infty}  = \Op(\sqrt{L \log(nN)  s_n}  l_n\eta_n (n^2N)^{1/q} ).\label{eq:orderstep1}
\end{align}

Using  \eqref{eq:rowLEu}, we have
 \begin{align}
     &\sup_{t\in[\tau, 1-\tau], u \in [0,1]} \left|\sum_{r=j}^{j+L-1} K_{\tau}(r/n - t)\E[\mathrm{tr}\{  \mf U_{r}(t,u)\}] \right. \\ &  \left.-\sum_{r=j}^{j+L-1} K_{\tau}(r/n - t) \sum_{k=1}^{3m} 2 \mathrm{tr}\left[\bs \Gamma_k(r/n, u) \left\{\sum_{h=1}^n\boldsymbol\Gamma_k^{\top}(h/n, u) K_{\tau}(h/n - t) \mf 1(h \in \mathbb L_r)\right\} \right]\right. \\ &  \left.- \sum_{r=j}^{j+L-1} K_{\tau}(r/n - t) \sum_{k=3m+1}^{s_n} 2\mathrm{tr} \left\{\sum_{h=1}^n \bs \Gamma_{r-h}(r/n , u) \boldsymbol\Gamma^{\top}_{r-h}((r-k)/n, j/N) K_{\tau}(h/n - t) \mf 1(h \in \mathbb L_r)\right\} \right|\\
    &= O(Ls_n^{3/2}(\log n +s_n)/n+  L(\log n)s_n/n).
\end{align}
Then, by the continuity of $|{\bs \Gamma}_{k}(t,u)|_F$ and $|{\bs \Gamma}_{k}(t,u)|_F = O(\chi^k)$, we can derive
 \begin{align}
     &\sup_{t\in[\tau, 1-\tau], u \in [0,1]} \left|\sum_{r=j}^{j+L-1} K_{\tau}(r/n - t)\E[\mathrm{tr}\{  \mf U_{r}(t,u)\}] -  2Ls_n K^2_\tau(j/n -t)\sum_{k =1}^{\infty}|{\bs \Gamma}_{k}(j/n,u)|_F^2\right|  \\& = O(g_n^*),\label{eq:LemmaB1}
\end{align}
where $g_n^{*} =  L^2 s_n^2/n + L s_n^{5/2} (\log n) /n + s_n^2 L /(n\tau)= O(L^2 s_n^2/n + s_n^2 L /(n\tau))$ under the condition $s_n^{1/2} (\log n) /L \to 0$.  
\par By triangle inequality, \eqref{eq:LemmaB1} and the continuity of $\mathrm{tr}\left\{\sum_{k=1}^{\infty} \bs \Gamma_i(t,u)\bs \Gamma^{\top}_i(t,u) \right\}$ with respect to $t$ and uniform for all $u \in [0,1]$ in 
\cref{ass:diff},   using the arguments in \eqref{eq:step2diff}, we have
\begin{align}
    &\E\left(\max_{0 \leq l \leq n-2\nt}|\check{\mf Z}_l  - \bar{\mf Z}_l|^q_{\infty}|\F_n\right)  \\ 
    & \leq  M \left|\sqrt{\frac{\log( n  N)}{2^{-1}s_n^2 (n\tau) L } } \right.\\&\times \left. \left\{ \max_{\substack{0 \leq  l \leq n-2\nt,\\ 1 \leq k \leq N} }\sum_{j=L}^{2\nt -L} \left(\left[\sum_{i=j - L }^{j-1} \left\{ K_{\tau}\left(\frac{i-\nt}{n} \right)\E \left[\mathrm{tr}\left\{\mf U_{i+l-1}\left(\frac{\nt + l -1}{n}, \frac{k}{N}\right) \right\}\right]\right.\right.\right.\right.\right.\\ &  \left.\left.\left.\left.\left.- K_{\tau}\left(\frac{i+L-\nt}{n}\right) \E \left[\mathrm{tr}\left\{\mf U_{i+L+l-1}\left(\frac{\nt + l -1}{n}, \frac{k}{N}\right)\right\}\right]
\right\} \right]^2 \right)\right\}^{1/2} \right|^q \\ 
&  = O\left\{ (L\log (n  N))^{q/2}( g_n^*/(s_n 
 L)+ L/(n\tau) )^q \right\}. \label{eq:orderstep2}
\end{align}
Therefore, by Markov's inequality, we have 
\begin{align}
   \max_{0 \leq l \leq n-2\nt}|\check{\mf Z}_l  - \bar{\mf Z}_l|_{\infty} = \Op \left\{ (L\log (n  N))^{1/2}( g_n^*/(s_nL)+L/(n\tau))\right\}.
\end{align}
Note that under the alternative hypothesis
\begin{align}
 \tilde W(t,u,i/n, F_k) = s_n^{-1} (u_{i1}(t,u) -\E   u_{i1}(t,u) ) + s_n^{-1} u_{i2}(t,u).
\end{align}
Recall the definition of $\tilde W_{i,k,j}$ in \eqref{eq:tildeW}, the definition of $\bar \sigma^2_{k_1, k_2, j_1, j_2}$  in \eqref{eq:sigmabar}, and the definition of $\bar S_{r,k,j}$ in \eqref{eq:defSbar}. We have 
\begin{align}
    &\max_{\substack{0 \leq k_1, k_2 \leq n - 2\nt,\\1 \leq j_1, j_2 \leq N} } \bar \sigma^2_{k_1, k_2, j_1, j_2}\\ &\leq  \max_{\substack{L\leq r \leq 2\nt -L,\\ 0 \leq k \leq n - 2\nt , 1 \leq j \leq N}} 4\| \bar S_{r,k,j}\|^2
    \\
    & \leq 4 \max_{\substack{L\leq r \leq 2\nt -L,\\0 \leq k \leq n - 2\nt , 1 \leq j \leq N} }  \left\|\sum_{i=r}^{r + L -1} K_{\tau}\left(\frac{i-\nt}{n}\right) \tilde W\left(\frac{\nt + k -1}{n}, \frac{j}{N}, \frac{i+k-1}{n}, \FF_{i+k-1}\right) \right\|^2/L\\
    & \leq \frac{4}{s^2_n L} \sup_{t \in [\tau, 1-\tau], u\in[0,1]} \max_{\substack{L\leq j \leq 2\nt -L}} \left\{\left\|\sum_{i=j}^{j+L-1} K_{\tau}\left(\frac{i - \nt}{n}\right) (u_{i1}(t,u) - \E u_{i1}(t,u)) \right\| \right.\\  &+  \left. \left\| \sum_{i=j}^{j+L-1} K_{\tau}\left(\frac{i - \nt}{n}\right) u_{i2}(t,u) \right\|  \right\}^2\\
    & \leq \frac{8}{s^2_n L} \sup_{t \in [\tau, 1-\tau], u\in[0,1]} \max_{\substack{L\leq j \leq 2\nt -L}} \left\|\sum_{i=j}^{j+L-1} K_{\tau}\left(\frac{i - \nt}{n}\right) (u_{i1}(t,u) - \E u_{i1}(t,u)) \right\|^2 \\  &+   \frac{8}{s^2_n L} \sup_{t \in [\tau, 1-\tau], u\in[0,1]} \max_{\substack{L\leq j \leq 2\nt -L}} \left\| \sum_{i=j}^{j+L-1} K_{\tau}\left(\frac{i - \nt}{n}\right) u_{i2}(t,u) \right\|^2 \\
    &:= S_1^2 + S_2^2,\label{eq:barsigmaalt}
\end{align}
where $S_1$ and $S_2$ are defined in the obvious way. For the first term, since $\E (\bar{\bs \phi}_{i,k}(u)) = 0$, using \cref{lm:delta1} by triangle inequality and Burkholder inequality, we have
\begin{align}
    S_1 & \leq 2\sqrt{\frac{2}{s^2_n L}} \sum_{l \in \mathbb Z} \left\| \sum_{i=j}^{j+L-1}\proj^{i-l}K_{\tau}\left(\frac{i - \nt}{n}\right)(u_{i1}(t,u) - \E\left\{u_{i1}(t,u)\right\}  \right\|\\ 
    &\leq 2\sqrt{\frac{2}{L}} \sum_{l \in \mathbb Z}  \left\{\left(\sum_{i=j}^{j+L-1}\theta_{i-l,2} \right)^2 \right\}^{1/2} = O(s_n). \label{eq:S1}
\end{align}
For the second term $S_2$, we first investigate $\sum_{i=j}^{j+L-1}  K_{\tau}\left(\frac{i - \nt}{n}\right)\sum_{k=1}^{s_n} \mathrm{tr}\left\{\bs \Gamma_k(i/n, u)\bar{\bs \psi}_{i,k}^{\top}(t,u)\right\}$ in $u_{i2}(t,u)$ in \eqref{eq:ualt}. Note that by triangle inequality and Burkholder inequality,for a $\chi_1 \in (0,1)$,  we have
\begin{align}
    &\left\|\sum_{i=j}^{j+L-1}K_{\tau}\left(\frac{i - \nt}{n}\right)\sum_{k=1}^{s_n} \mathrm{tr}\left\{\bs \Gamma_k(i/n, u)\bar{\bs \psi}_{i,k}^{\top}(t,u)\right\} \right\| \\ 
    & \leq \sum_{l\in\mathbb Z} \sum_{k=1}^{s_n} \sum_{r,h=1}^p \left \|\sum_{i=j}^{j+L-1}K_{\tau}\left(\frac{i - \nt}{n}\right) \Gamma_{k}^{r,h}(i/n,u)\proj^{i-l} \psi_{i,k}^{r,h}(t,u) \right\|, \\
    &\leq \left( \sum_{l\geq 4s_n+1} +  \sum_{l\leq 4s_n} \right)\sum_{k=1}^{s_n} \sum_{r,h=1}^p \left\{ \sum_{i=j}^{j+L-1}\left|K_{\tau}\left(\frac{i - \nt}{n}\right) \Gamma_{k}^{r,h}(i/n,u)\right| \left\| \psi_{i,k}^{r,h}(t,u) - \psi_{i,k}^{r,h, (i-l)}(t,u) \right\|^2 \right\}^{1/2}\\ 
    & = O(\sqrt{Ls_n} \chi_1^{s_n/2} + s_n^2 \sqrt{L} ),\label{eq:u21}
\end{align}
where the last equality follows from   \eqref{eq:F21} and similar arguments in \eqref{eq:F22}. 
Next, we investigate $\sum_{i=j}^{j+L-1}K_{\tau}\left(\frac{i - \nt}{n}\right)\sum_{k=1}^{s_n}\mathrm{tr}\left\{\bar{\bs \phi}_{i,k}(u) \sum_{j=(i-s_n)\vee 1}^{i-M-1} \bs \Gamma_k^{\top}(j/n, u)\right\} K_{\tau}(j/n - t)$ in $u_{i2}(t,u)$ of \eqref{eq:ualt}. By triangle inequality and Burkholder inequality, for a $\chi_2 \in (0,1)$,   we have 
\begin{align}
   & \left\| \sum_{i=j}^{j+L-1}K_{\tau}\left(\frac{i - \nt}{n}\right)\sum_{k=1}^{s_n}\mathrm{tr}\left\{\bar{\bs \phi}_{i,k}(u) \sum_{j=(i-s_n)\vee 1}^{i-M-1} \bs \Gamma_k^{\top}(j/n, u)\right\} K_{\tau}(j/n - t) \right\| \\
    & \leq \sum_{l \in 
 \mathbb Z}\sum_{k=1}^{s_n} \sum_{r,h=1}^p \left\|\sum_{i=j}^{j+L-1}K_{\tau}\left(\frac{i - \nt}{n}\right)\left\{\proj^{i-l} \bar{\phi}^{r,h}_{i,k}(u) \sum_{j=(i-s_n)\vee 1}^{i-M-1} K_{\tau}(j/n - t)  \Gamma_k^{r,h} (j/n, u)\right\} \right\| \\ 
 & \leq \sum_{l \in 
 \mathbb Z}\sum_{k=1}^{s_n} \sum_{r,h=1}^p \left\{\sum_{i=j}^{j+L-1} \|\bar{\phi}^{r,h}_{i,k}(u) - \bar{\phi}^{r,h,(i-l)}_{i,k}(u)\|^2 \right.\\& \times \left. \left| K_{\tau}\left(\frac{i - \nt}{n}\right) \sum_{j=(i-s_n)\vee 1}^{i-M-1} K_{\tau}(j/n - t)  \Gamma_k^{r,h} (j/n, u)\right|^2  \right\}^{1/2}\\
 & = O\left\{ s^2_n \sqrt{L}\sum_{l=4s_n+1}^{\infty} \chi_2^{l} +s_n \sqrt{L}\sum_{l=1}^{4s_n}\sum_{k=1}^{s_n} (\chi_2^{(l-k)}\mf 1(l > k)+ \chi_2^l) \right\}\\
 & = O(s_n^2\sqrt{L} ).\label{eq:u22}
\end{align}
Using \eqref{eq:ualt}, \eqref{eq:u21} and \eqref{eq:u22} as well as similar arguments therein, by triangle inequality, since $s_n/(n\tau) \to 0$, for a $\chi_2 \in (0,1)$,  we have 
\begin{align}
    S_2 &=2\sqrt{\frac{2}{s^2_n L}} \sup_{t \in [\tau, 1-\tau], u\in[0,1]} \max_{\substack{L\leq j \leq 2\nt -L}} \left\| \sum_{i=j}^{j+L-1} K_{\tau}\left(\frac{i - \nt}{n}\right) u_{i2}(t,u) \right\|\\ & =O\left(\sqrt{\frac{2}{s_n^2 L}} s_n^2\sqrt{L}\right) = O( s_n).\label{eq:S2}
\end{align}
Combining \eqref{eq:barsigmaalt}, \eqref{eq:S1} and \eqref{eq:S2}, we have 
\begin{align}
   \max_{\substack{0 \leq k_1, k_2 \leq n - 2\nt,\\1 \leq j_1, j_2 \leq N} } \bar \sigma^2_{k_1, k_2, j_1, j_2} = O(s^2_n). \label{eq:orderstep3}
\end{align}
Recall that $g_n^{*} = O(L^2 s_n^2/n + s_n^2 L /(n\tau))$.
Combining \eqref{eq:orderstep1}, \eqref{eq:orderstep2} and \eqref{eq:orderstep3}, since 
\begin{align}
    & \sqrt{L \log(nN)  s_n}  l_n\eta_n (n^2N)^{1/q} + (L\log (n  N))^{1/2}( g_n^*/(s_nL) + L/(n\tau))
    \\ 
    &=\sqrt{L \log(nN)  s_n}  l_n\eta_n (n^2N)^{1/q} + (L\log (n  N))^{1/2}\{ L/(n\tau) + Ls_n/n\}\\ 
    &=o(s_n \log(nN))
\end{align}
we have 
\begin{align}
   \tilde Z_{\bt} = \max_{0 \leq l \leq n-2\nt}|\tilde{\mf Z}_l|_{\infty} = \Op(s_n\log (nN)).
\end{align}
\end{proof}

\section{Proof of \cref{thm:extend}}\label{sec:exproof}
\subsection{Proof of \cref{thm:extend}}
\begin{proof}
Note that the point-wise upper bounds of the form $\sup_{t \in [\tau, 1-\tau], u,v \in [0,1]}\| \cdot\|_p$ for $p\geq 1$ will be the same as $\sup_{t \in [\tau, 1-\tau], u \in [0,1]}\| \cdot\|_p$  if they can be controlled by the quantities assumed in (1)-(3) in \cref{ass:expo}. 

The difference between \cref{thm:extend} and \cref{thm:ga}, \cref{thm:alt} and \cref{thm:bs} lies in the following three aspects:

(1) Deriving the uniform upper bound of the random variables of with respect to three variables  $t\in [\tau, 1-\tau], u,v \in[0,1]$ for the extended estimator \eqref{eq:testQextend} instead of only two variables $t\in [\tau, 1-\tau], u\in[0,1]$ for the estimator \eqref{eq:testQ}. The maximum inequality for three continuous variables is investigated in \cref{prop:b3} as a counterpart of \cref{prop:b2}. As a result, \cref{prop:esti} will be extended into the following: 

If \cref{ass:expo} is satisfied and  for some  $q \geq 2$,  $$\delta^*_n= s_n^2(n\tau) \{b^2 + 1/(nb)\}\tau^{-3/q} +s_n^{3/2} \sqrt{n/b} \tau^{1-3/q} = o(s_n \sqrt{n\tau}),$$ and   $\tau \to 0$,  $n\tau \to \infty$, $(\log n)^4/(nb) \to 0$, then we have  
\begin{align}
    \left\| \sup_{t \in [\tau, 1-\tau], u,v \in [0,1]}\left|  \sum_{i=1}^{n}\mathrm{tr} \left\{\hat{\mf U}_i(t,u,v)- \mf U_i(t,u,v)\right\}K_{\tau}(i/n - t) \right| \right\|_q &=  O(\delta^*_n) = o(s_n \sqrt{n\tau}).
\end{align}
(2) Extend 
$$
g(t, u) =\frac{1}{s_n\sqrt{n\tau}} \sum_{i=1}^n K_{\tau}(i/n-t) \mathrm{tr}\{\mf U_i(t, u)\}
$$
to $g(t, u, v) =\frac{1}{s_n\sqrt{n\tau}} \sum_{i=1}^n K_{\tau}(i/n-t) \mathrm{tr}\{\mf U_i(t, u, v)\}$, and derive the upper bound for approximating $\sup_{t \in [\tau, 1-\tau], u,v \in [0,1]}g(t, u, v)$ by $\max_{1 \leq i \leq n, 1 \leq j,j^{\prime} \leq N}g(i/n, j/N, j^{\prime}/N)$. 

Through tedious calculations, the rate in \eqref{eq:maxg} will be replaced by 
\begin{align}
   \left\| \underset{t \in [\tau, 1-\tau], u,v \in [0,1]}{\sup}\underset{ \substack{1 \leq i \leq n\\ 1 \leq j, j^{\prime} \leq N  }}{\max} |g(t, u, v) - g(i/n,  j/N,  j^{\prime} /N)| \right\|_q = O( s_n (nN^2)^{1/q}((n\tau)^{-1}+ N^{-1})).
\end{align}

(3) Apply Gaussian approximation and establishing bootstrap consistency for the $(n-2\nt)N^2$-dimensional vector $\mf V_i$,
\begin{align}
      {\mf V}_{i,l} &:= K_{\tau}(i/n-l/n) (\mathrm{tr}(\mf U_{i}(l/n, j/N, j^{\prime} /N)),  1 \leq j,j^{\prime} \leq N)^{\T},  \\
      \mf V_i &= ({\mf V}_{i,\nt} ,\cdots, {\mf V}_{i + n - 2\nt, n-\nt} )^{\top},
\end{align}  
instead of the $(n-2\nt)N$ -dimensional one
\begin{align}
      {\mf V}_{i,l}:&= K_{\tau}(i/n-l/n) (\mathrm{tr}(\mf U_{i}(l/n, j/N)),  1 \leq j \leq N)^{\T}, \\ \mf V_i &= ({\mf V}_{i,\nt} ,\cdots, {\mf V}_{i + n - 2\nt, n-\nt} )^{\top}.
\end{align} 
Therefore, it suffices to change $N$ into $N^2$ in the results of \eqref{eq:GA} and \cref{thm:bs}.
\end{proof}
\end{appendix}

\begin{funding}
Financial support by the Deutsche Forschungsgemeinschaft (DFG, German Research Foundation; Project-ID 520388526;  TRR 391:  Spatio-temporal Statistics for the Transition of Energy and Transport) is gratefully acknowledged. Weichi Wu is the corresponding author and is supported  by NSFC General program (No.12271287). The authors are  listed in alphabetical order.
\end{funding}

\bibliography{main}
\bibliographystyle{imsart-nameyear}

\newpage


\end{document}